\documentclass[aps,prx,superscriptaddress,notitlepage, amsfonts,longbibliography, twocolumn]{revtex4-2}

\usepackage[dvipdfmx]{graphicx}

\usepackage{amsmath,amssymb,amsthm,mathtools,mathrsfs}
\usepackage{srcltx}

\usepackage{bm,braket,comment}
\usepackage{amsmath,amssymb,amsthm,mathtools}
\usepackage{xcolor}
\usepackage{longtable}
\usepackage{enumerate}
\usepackage{ulem}
\usepackage{multirow}
\usepackage{here}
\usepackage{comment}
\usepackage{url}

\newtheorem{theorem}{Theorem}
\newtheorem{definition}{Definition}
\newtheorem{assumption}{Assumption}
\newtheorem{lemma}{Lemma}
\newtheorem{corollary}{Corollary}
\newtheorem{proposition}{Proposition}

\newcommand{\bout}[1]

\usepackage{hyperref}
\hypersetup{colorlinks=true,linkcolor=blue,citecolor=blue,urlcolor=blue}

\allowdisplaybreaks[4]
\begin{document}

\title{A General Theory of Multi-Resource Theories Involving Finite-Group Asymmetry}

\author{Yosuke Mitsuhashi}
\email{yosuke.mitsuhashi@riken.jp}
\affiliation{Analytical Quantum Complexity RIKEN Hakubi Research Team, RIKEN Center for Quantum Computing (RQC), Wako, Saitama 351-0198, Japan}

\author{Hiroyasu Tajima}
\email{hiroyasu.tajima@inf.kyushu-u.ac.jp}
\affiliation{Department of Informatics, Faculty of Information Science and Electrical Engineering, Kyushu University, 744 Motooka, Nishi-ku, Fukuoka 819-0395, Japan}
\affiliation{JST, FOREST, 4-1-8 Honcho, Kawaguchi, Saitama 332-0012, Japan}

\begin{abstract}
Resource theories provide a universal framework for quantifying properties of quantum states, such as entanglement, magic, athermality, and asymmetry, as resources under various constraints encountered in quantum science. 
Since realistic settings typically impose several constraints at once, it is important to study multi-resource theories that combine them. 
Among these, combinations with asymmetry are especially natural, as symmetry is fundamental and ubiquitous in physics. 
Here we develop a general theory that treats combinations of finite-group asymmetry with a broad range of resource theories in a unified manner. 
Our key observation is that the effect of combining the two constraints is governed by the accessibility of symmetry information under the underlying resource theory: whether the information needed for the conversion can be extracted from the input and used to control free operations. 
If this can be done with asymptotically vanishing error, we show that the conversion rate of the combined theory equals the smaller of the rates of the two constituent theories, so combining them incurs no additional cost. 
This condition holds for LOCC on arbitrary mixed multipartite states, magic-state conversion of multiqubit systems under Clifford symmetries, and Gibbs-preserving operations. 
Thermal operations, being time-translation covariant, can access only limited symmetry information, giving an additional compatibility condition between the two symmetries. 
We also analyze symmetry imposed on elementary operations rather than only on the resulting channel. 
We establish analogous results for LOCC and thermal operations under suitable assumptions, and for magic when the symmetry is represented by Pauli operators. 
Finally, we show that finite-group symmetry does not reduce the asymptotic ergotropy per copy and hence introduces no new completely passive states.

\end{abstract}

\maketitle

\let\oldaddcontentsline\addcontentsline% Store \addcontentsline
\renewcommand{\addcontentsline}[3]{}% Make \addcontentsline a no-op

\onecolumngrid

\section{Introduction}

Quantum resource theories provide a common language for studying what physical systems can accomplish when the available operations are restricted~\cite{chitambar2019quantum,coecke2016mathematical}.
A resource is determined relative to the operations under which it cannot be freely prepared.
This perspective underlies entanglement theory~\cite{horodecki2009quantum}, quantum thermodynamics~\cite{goold2016role,lostaglio2019introductory}, and the resource theory of magic~\cite{bravyi2005universal,veitch2014resource}.
A central operational question is the optimal rate at which many independent copies of one state can be converted into copies of another with vanishing error.
Such rates characterize the leading-order cost of state conversion and reveal which operational restrictions remain relevant in the asymptotic limit.

The study of these rates has revealed markedly different structures across resource theories.
In bipartite pure-state entanglement, concentration and dilution relate asymptotic conversion to entanglement entropy~\cite{bennett1996concentrating}, whereas mixed-state purification introduces additional distinctions between preparation and distillation~\cite{bennett1996mixed-state}.
Single-copy majorization criteria~\cite{nielsen1999conditions} and the structure of local operations and classical communication (LOCC)~\cite{chitambar2014everything} further illustrate that state convertibility depends on the precise operational model.
In thermodynamics, microscopic models of free operations connect state transformations to energetic restrictions~\cite{janzing2000thermodynamic}; asymptotic and finite-size analyses identify different families of constraints, including free-energy relations and their generalized counterparts~\cite{brandao2013resource,horodecki2013fundamental,brandao2015second}. 
For magic, the choice of free operations is equally important.
The stabilizer framework supports a theory of nonclassical computational resources~\cite{veitch2014resource}, while quasiprobability negativity and contextuality provide resource criteria in particular stabilizer settings, notably in odd-prime dimensions~\cite{veitch2012negative,howard2014contextuality}.
The quantification and manipulation of magic states~\cite{ahmadi2018quantification}, the geometry of the stabilizer polytope~\cite{heinrich2019robustness}, and the treatment of magic in quantum channels~\cite{seddon2019quantifying} provide complementary descriptions of these restrictions.
These developments also highlight the distinction between defining free operations at the level of quantum channels and defining them in terms of elementary physical operations.

Physical constraints, however, rarely occur one at a time.
Spatially separated laboratories may lack a shared reference frame, a thermodynamic process may have to respect additional symmetries, and a stabilizer protocol may be restricted by a prescribed group action.
General multi-resource frameworks have investigated the exchange and joint accounting of different resources~\cite{sparaciari2020first}, and thermodynamics with several conserved quantities provides an important concrete setting~\cite{guryanova2016thermodynamics,yungerhalpern2016microcanonical}.
Here we focus on a different but complementary question: when two restrictions are imposed simultaneously on the state conversion protocols, is the resulting asymptotic rate determined by the separate restrictions, or does their combination create an additional obstruction?

Symmetry is a particularly natural restriction for addressing this question.
Covariant operations describe transformations that do not rely on an external reference frame for a specified group action~\cite{bartlett2007reference,gour2008resource}.
The resource theory of asymmetry characterizes the information about this action encoded in quantum states, through transformation criteria~\cite{marvian2013theory}, information-theoretic measures of reference-frame quality~\cite{gour2009measuring}, and constraints on dynamics beyond the conservation of expectation values~\cite{marvian2014extending,marvian2014modes}.
Such constraints also exist for general quantum processes and have been quantified through trade-off relations between process irreversibility under conservation laws and quantum coherence in the ancillary system used for implementation~\cite{tajima2025universaltradeoffstructuresymmetry}.
Importantly, whether this symmetry information can be extracted and used, however, depends on the restrictions imposed on the allowed operations.

Previous studies show that combining symmetry with another operational constraint can lead to nontrivial effects.
Earlier work on entanglement under superselection rules already demonstrates that locality and symmetry must be considered jointly~\cite{verstraete2003quantum,bartlett2003entanglement}.
In particular, particle-number superselection can lead to distinct nonlocal resources and additional asymptotic structure~\cite{schuch2004nonlocal,schuch2004quantum}.
Thermodynamics supplies another instructive example.
The Gibbs-preserving property alone does not capture all restrictions imposed by thermal operations~\cite{faist2015gibbs}.
Time-translation symmetry also constrains coherence independently of the usual population-based free-energy conditions~\cite{lostaglio2015description,lostaglio2015quantum}.
These examples show that combining two resource-theoretic constraints can change the asymptotic state-conversion rate. 
This naturally raises the question of under what conditions imposing an additional symmetry constraint leaves the asymptotic conversion rate unchanged.

Reference systems and catalysts offer one route to such a removal.
Catalysis can enable transformations forbidden without an auxiliary resource~\cite{jonathan1999entanglement-assisted}, and coherent reference systems can support repeated operations under conservation constraints~\cite{aberg2014catalytic}.
The coherence-irreversibility trade-off~\cite{tajima2025universaltradeoffstructuresymmetry} can also be understood as a trade-off between the quality of a reference system used as a catalyst and the error in implementing dynamics with it.
Related results have been obtained for thermodynamic implementations: the coherence cost of implementing certain Gibbs-preserving operations by thermal operations diverges as the implementation error vanishes~\cite{Tajima2025PRLa}.
This perspective has a range of applications, including the implementation of unitary operations~\cite{tajima_coherence_2018, tajima_coherence_2020}.
However, the possibility of using such auxiliary systems depends on the symmetry group and on how the auxiliary system is required to be returned. 
In particular, no-broadcasting results impose strong restrictions on finite-dimensional reference systems for connected Lie-group symmetries~\cite{marvian2019no-broadcasting,lostaglio2019coherence}.
For a finite group, by contrast, an orthogonal orbit can encode the relevant group label in a finite-dimensional system.
The issue in a combined theory is whether such a record can be distilled and used within the underlying free operations, rather than only under unrestricted quantum operations.

Finite-group asymmetry has an especially sharp asymptotic behavior.
For pure inputs, approximate i.i.d. conversion rates under finite-group-covariant operations are either zero or infinite, according to the inclusion of the input and output symmetry subgroups~\cite{shitara2025iid}.
In this work, we extend the analysis to arbitrary mixed states and investigate asymptotic state conversion when finite-group covariance is imposed together with another resource-theoretic constraint. 
Throughout this work, the group and the underlying single-copy Hilbert spaces are fixed in the asymptotic limit.

The operational mechanism is state discrimination followed by conditional control. 
The distinct states in the orbit of $\rho$ are indexed by $G/S$, where $S$ is the symmetry subgroup 
\begin{align}
    S:=\mathrm{Sym}(\rho)
    =\{g\in G \, |\, U(g)\rho U(g)^\dag=\rho\}
\end{align}
with projective unitary representation $U$ of $G$. 
Under unrestricted measurements, the asymptotic discrimination of finitely many distinct i.i.d. states is governed by quantum Chernoff theory~\cite{audenaert2007discriminating,nussbaum2009chernoff,li2016discriminating}.
Restricted measurements require a separate analysis, as illustrated by both local indistinguishability and positive results for local discrimination~\cite{bennett1999quantum,walgate2000local}.
Our construction estimates the expectation values of the orbit states themselves, regarded as observables, and combines the outcomes of repeated tests.
Hoeffding's inequality~\cite{hoeffding1963probability} then gives an explicit exponential bound on the discrimination error.

We first establish a general state-conversion result based on measurements whose outcomes can be used to condition subsequent free operations. 
If the required orbit information can be extracted with vanishing error, a vanishing fraction of input copies suffices for discrimination, while the remaining copies support an asymptotically optimal conversion.
Group averaging then yields a covariant free channel without changing the first-order rate.
This argument does not require a physical memory storing the measurement outcome. 
We nevertheless introduce such memory states as an alternative formulation, which provides a useful operational interpretation of the discrimination-and-control mechanism.

For LOCC on arbitrary mixed multipartite states, local-observable decompositions implement the required tests.
For multiqubit magic under Clifford representations, randomized Pauli measurements provide the analogous construction.
Consequently, whenever $\mathrm{Sym}(\rho)\subset\mathrm{Sym}(\sigma)$, imposing finite-group covariance does not change the asymptotic conversion rate under LOCC, SEP, stabilizer operations, or completely stabilizer-preserving operations. 
If this condition is not satisfied, the conversion rate is zero. 
These results hold without requiring an explicit expression for the asymptotic conversion rate in the absence of the symmetry constraint. 
For fixed input state and group, our LOCC and magic protocols use $O(\log(1/\epsilon))$ copies to achieve discrimination error $\epsilon$.

Thermodynamic operations require a refinement because time-translation covariance limits accessible orbit information.
We treat Gibbs-preserving operations, time-translation-covariant Gibbs-preserving operations, and thermal operations separately.
For Gibbs-preserving operations, control admissibility suffices directly.
For the time-covariant classes, the relevant subgroup is determined by collectively energy-dephased blocks of the input, and successful conversion requires compatibility between the finite-group action and time translations on the input and output.
For thermal operations, an energy-compatible projective construction achieves the required discrimination with $O((\log(1/\epsilon))^2)$ copies.
We note that these bounds on the number of copies are not necessarily optimal.

We also distinguish symmetry imposed at the level of quantum channels from symmetry imposed on the elementary operations used to implement them. 
For LOCC and thermal operations, we establish corresponding primitive-level results under suitable assumptions.
For magic, we show that the map-level conversion rate can also be achieved at the primitive level when the symmetry is represented by Pauli operators.

Finally, we consider work extraction by symmetry-respecting unitaries.
The relation between complete passivity and equilibrium was established without additional symmetry restrictions~\cite{pusz1978passive,lenard1978thermodynamical}.
Ergotropy quantifies finite-system work extraction~\cite{allahverdyan2004maximal}, and collective control can increase the extractable work per copy~\cite{alicki2013entanglement}.
Symmetry-protected complete passivity has been characterized for connected compact Lie groups, with finite cyclic and dihedral groups providing contrasting discrete examples~\cite{mitsuhashi2022characterizing}.
We prove that an arbitrary fixed finite-group symmetry commuting with the Hamiltonian leaves the asymptotic ergotropy per copy unchanged.
It follows that such a symmetry introduces no additional completely passive states, including when the ground eigenspace is degenerate.
Overall, our results identify conditions under which imposing a finite-group symmetry constraint leaves the asymptotic conversion rate unchanged, and clarify how these conditions depend on the free operations available in the underlying resource theory.

The logical structure of our main results is summarized in Fig.~\ref{fig:overview}.
We first establish the relation between state discrimination and state conversion for finite-group asymmetry and characterize the corresponding asymptotic conversion rate.
We then extend this relation to general resource theories subject to an additional symmetry constraint, leading to our general conversion theorem.
Finally, we apply this theorem to entanglement, magic, and thermodynamics, and further investigate primitive-level implementations and symmetry-constrained work extraction.

Throughout this manuscript, we denote the set of all linear operators from a Hilbert space $\mathcal{H}$ to another Hilbert space $\mathcal{H}'$ by $\mathcal{L}(\mathcal{H}\to\mathcal{H}')$. 
If $\mathcal{H}=\mathcal{H}'$, we simply denote it by $\mathcal{L}(\mathcal{H})$. 
We denote the set of all Hermitian operators and unitary operators on $\mathcal{H}$ by $\mathcal{L}^\mathrm{H}(\mathcal{H})$ and $\mathcal{U}(\mathcal{H})$, respectively.

\begin{figure*}
    \centering
    \includegraphics[width=155mm]{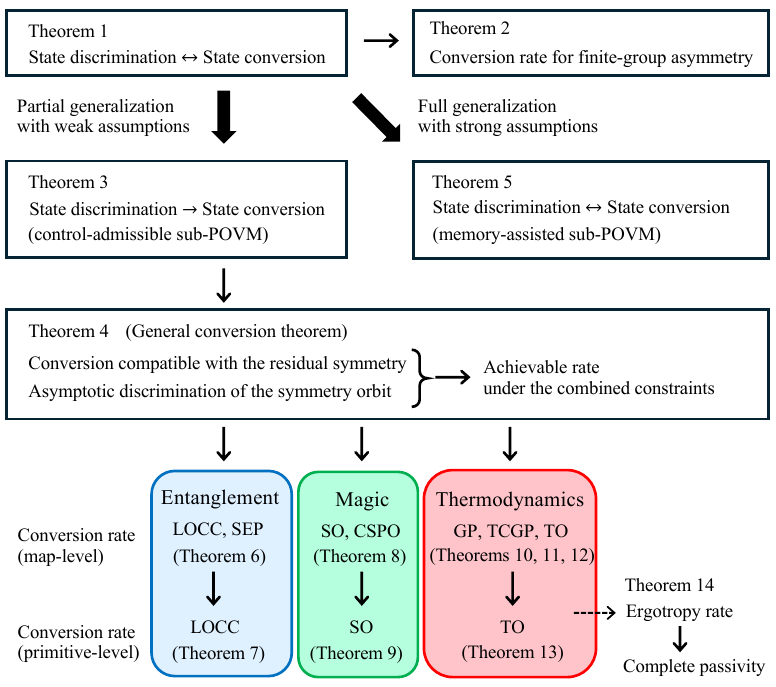}
    \caption{
Overview of the main results.
For finite-group asymmetry, state discrimination and state conversion are first related directly, leading to the characterization of the asymptotic conversion rate.
This relation is then generalized to resource theories subject to an additional symmetry constraint through control-admissible and memory-assisted measurements.
The general conversion theorem shows that compatibility of the conversion with the residual symmetry, together with asymptotic discrimination of the symmetry orbit, guarantees an achievable rate under the combined constraints.
The theorem is applied to entanglement, magic, and thermodynamics, with additional primitive-level and work-extraction results.
}
    \label{fig:overview}
\end{figure*}

\section{Resource theory of finite-group asymmetry and state discrimination}

We consider the relation between the task of distinguishing the states $\{\mathcal{U}_g(\rho)\}_{g\in G}$ and state conversion. 
Since the state set satisfies $\mathcal{U}_g(\rho)=\mathcal{U}_{g'}(\rho)$ if and only if $g^{-1}g'\in S$, the exact problem setup is distinguishing the coset $gS$ from states $\{\mathcal{U}_g(\rho)\}_{g\in G}$. 
Throughout this work, the symmetry representation on a composite system is taken to be the tensor product of the representations on its constituent subsystems.
More precisely, if systems $A$ and $B$ carry projective unitary representations $U_A$ and $U_B$ of $G$, respectively, the representation on $AB$ is
\begin{align}
    U_{AB}(g):=& U_A(g)\otimes U_B(g).
\end{align}
We use the same convention for tensor powers, so that the representation on $A^{\otimes n}$ is $U_A(g)^{\otimes n}$.

First, we define a general distinguishability measure $\Delta$, which includes measures such as trace distance, infidelity, and many versions of divergences as specific examples.

\begin{definition} (Distinguishability measure.) \label{SMdef:general_distinguishability_measure}
    Let $\Delta$ be a family of functions $\Delta_\mathcal{H}: \mathcal{S}(\mathcal{H})\times\mathcal{S}(\mathcal{H})\to [0, \infty]$ indexed by a Hilbert space $\mathcal{H}$. 
    We say that $\Delta$ is a distinguishability measure if it is monotonically decreasing under CPTP maps, i.e., 
    \begin{align}
        \Delta_\mathcal{H}(\rho, \sigma)\geq \Delta_{\mathcal{H}'}(\mathcal{E}(\rho), \mathcal{E}(\sigma)) 
    \end{align}
    for all finite-dimensional Hilbert spaces $\mathcal{H}$, $\mathcal{H}'$, $\rho, \sigma\in\mathcal{S}(\mathcal{H})$ and CPTP maps $\mathcal{E}\in\mathcal{O}(\mathcal{L}(\mathcal{H})\to\mathcal{L}(\mathcal{H}'))$. 
\end{definition}

In the following, we write $\Delta_\mathcal{H}$ as $\Delta$ by omitting the underlying Hilbert space $\mathcal{H}$, which is obvious from the input states.

Next, for a distinguishability measure $\Delta$, we define the orthogonal distinguishability function $f_\Delta$.

\begin{definition} (Orthogonal distinguishability function.) \label{SMdef:orthogonal_distinguishability_function}
    Let $\mathcal{H}$ be a finite-dimensional Hilbert space, $\rho, \sigma\in\mathcal{S}(\mathcal{H})$ be a pair of orthogonal states, and $\Delta$ be a distinguishability measure. 
    We define the orthogonal distinguishability function $f_\Delta: [0, \infty]\to [0, \infty]$ of $\Delta$ by 
    \begin{align}
        f_\Delta(x):=\Delta(\rho, e^{-x}\rho+(1-e^{-x})\sigma)\ \forall x\in [0, \infty]. \label{SMeq:SMdef:orthogonal_distinguishability_function}
    \end{align}
\end{definition}

We note that this is well-defined, because the right-hand side does not depend on $\rho$, $\sigma$, or even the underlying Hilbert space $\mathcal{H}$ as long as $\rho$ and $\sigma$ are orthogonal to each other, which is shown in Corollary~\ref{SMcor:distinguishability_state_independence}.

Important examples of distinguishability measures are the max-divergence, the min-divergence, the trace distance, and the infidelity, defined by
\begin{align}
    &D_{\max}(\rho\|\sigma):=
    \inf\{\lambda\in\mathbb{R}\,|\,\rho\leq e^\lambda\sigma\}, \\
    &D_{\min}(\rho\|\sigma):=
    -\log\left(\mathrm{tr}(\Pi_\rho\sigma)\right), \\
    &\mathrm{T}(\rho,\sigma):=
    \frac{1}{2}\|\rho-\sigma\|_1, \\
    &1-F(\rho,\sigma):=
    1-\|\sqrt{\rho}\sqrt{\sigma}\|_1^2, 
\end{align}
where $\Pi_\rho$ denotes the projection onto the support of $\rho$.
For these measures, the corresponding orthogonal distinguishability functions are
\begin{align}
    &f_{D_{\max}}(x)=f_{D_{\min}}(x)=x, \\
    &f_{\mathrm{T}}(x)=f_{1-F}(x)=1-e^{-x}.
\end{align}
Thus, all four examples satisfy the strict monotonicity condition used below.

We will also use a canonical reference state associated with each subgroup of $G$.
Let $\mathcal{H}_R$ be the left regular representation space of $G$, with orthonormal basis $\{|\xi_g\rangle\}_{g\in G}$ and representation
\begin{align}
    L(g)|\xi_h\rangle:=&|\xi_{gh}\rangle.
\end{align}
For a subgroup $S\subset G$, we define
\begin{align}
    \tau_S:=&\frac{1}{|S|}\sum_{s\in S}|\xi_s\rangle\langle\xi_s|.
    \label{eq:canonical_reference_state}
\end{align}
Under the left regular representation, we have
\begin{align}
    L(g)\tau_S L(g)^\dag
    =&\frac{1}{|S|}\sum_{s\in S}|\xi_{gs}\rangle\langle\xi_{gs}|.
\end{align}
Therefore,
\begin{align}
    \mathrm{Sym}(\tau_S)=&S.
    \label{eq:canonical_reference_symmetry}
\end{align}
Moreover, the distinct states in the orbit of $\tau_S$ are naturally indexed by the left cosets $G/S$.
Thus, $\tau_S$ provides a canonical state whose symmetry subgroup is exactly $S$.

\begin{theorem} \label{thm:conversion_state_discrimination}
    Let $\epsilon\in [0, 1)$, $\Delta$ be a distinguishability measure, its orthogonal distinguishability function $f_\Delta$ be strictly increasing, and $S$ be the symmetry subgroup of $\rho$. 
    Then, the following statements are equivalent: 
    \begin{align}
        &\textrm{(i)}\ &&\ \forall \mathcal{E}\in\mathcal{O}^{S\textrm{-}\mathrm{cov}},\ \exists \mathcal{F}\in\mathcal{O}^{G\textrm{-}\mathrm{cov}} \textrm{ s.t. } \forall \chi\in\mathcal{S}, \Delta(\mathcal{E}(\chi), \mathcal{F}(\chi\otimes \rho))\leq f_\Delta(-\log(1-\epsilon)), \\
        &\textrm{(ii)}\ &&\ \forall \sigma\in\mathcal{S}^{S\textrm{-}\mathrm{sym}},\ \exists \mathcal{F}\in\mathcal{O}^{G\textrm{-}\mathrm{cov}} \textrm{ s.t. } \Delta(\sigma, \mathcal{F}(\rho))\leq f_\Delta(-\log(1-\epsilon)), \\
        &\textrm{(iii)}\ &&\ \exists \mathcal{F}\in\mathcal{O}^{G\textrm{-}\mathrm{cov}} \textrm{ s.t. } \Delta(\tau_S, \mathcal{F}(\rho))\leq f_\Delta(-\log(1-\epsilon)), \\
        &\textrm{(iv)}\ &&\ \textrm{There exists a set of POVMs } \{M_x\}_{x\in G/S} \textrm{ such that }\min_{g\in G} \mathrm{tr}(\mathcal{U}_g(\rho) M_{gS})\geq 1-\epsilon, \\
        &\textrm{(v)}\ &&\ \textrm{There exists a set of POVMs } \{M_x\}_{x\in G/S} \textrm{ such that }\frac{1}{|G|}\sum_{g\in G} \mathrm{tr}(\mathcal{U}_g(\rho) M_{gS})\geq 1-\epsilon. 
    \end{align}
\end{theorem}

Reference~\cite{li2016discriminating} shows that the best achievable $\epsilon\sim e^{-n\xi}$ with $\xi:=\min_{g, h: g^{-1}h\not\in S} C(\mathcal{U}_g(\rho), \mathcal{U}_h(\rho))$, where $C$ is the Chernoff distance $C(\rho, \sigma):=\max_{0\leq s\leq 1} -\log(\mathrm{tr}(\rho^s\sigma^{1-s})$.

When we take max-divergence $D^\mathrm{max}$, trace distance $T$, infidelity $1-F$, or min-divergence $D^\mathrm{min}$ as $\Delta$, the statement (iii) respectively means 
\begin{align}
    &\exists \mathcal{F}\in\mathcal{O}^{G\textrm{-}\mathrm{cov}} \textrm{ s.t. } D^\mathrm{max}(\tau_S \|\mathcal{F}(\rho))\leq -\log(1-\epsilon), \\
    &\exists \mathcal{F}\in\mathcal{O}^{G\textrm{-}\mathrm{cov}} \textrm{ s.t. } T(\tau_S, \mathcal{F}(\rho))\leq \epsilon, \\
    &\exists \mathcal{F}\in\mathcal{O}^{G\textrm{-}\mathrm{cov}} \textrm{ s.t. } 1-F(\tau_S, \mathcal{F}(\rho))\leq \epsilon, \\
    &\exists \mathcal{F}\in\mathcal{O}^{G\textrm{-}\mathrm{cov}} \textrm{ s.t. } D^\mathrm{min}(\tau_S \|\mathcal{F}(\rho))\leq -\log(1-\epsilon). 
\end{align}
Corollary~\ref{cor:distingishability_relation_Dmax_general_Dmin} implies that the first condition is the strongest, and the last condition is the weakest. 
An interesting aspect of this proposition is that it implies all of these conditions are equivalent. 
We note that the strict monotonicity of $f_\Delta$ is used only in the proof of (iii) $\rightarrow$ (iv).

\begin{proof}
    We note that there are many trivial relations between these statements. 
    By considering the specific case where $\mathcal{E}$ is a map from $\mathbb{C}$ to $\mathcal{L}(\mathcal{H})$ and $\mathcal{E}(c):=c\sigma$, we can confirm that (i) $\rightarrow$ (ii) is trivial. 
    By taking the specific case of $\sigma=\tau_S$, we find that (ii) $\rightarrow$ (iii) is trivial. 
    Since (iv) gives a lower bound of the worst case, and (v) gives a lower bound of the average case, (iv) $\rightarrow$ (v) is trivial. 
    Therefore, it is sufficient to prove (v) $\rightarrow$ (i) and (iii) $\rightarrow$ (iv).

    First, we show (v) $\rightarrow$ (i). 
    We take a set of POVMs $\{M_x\}_{x\in G/S}$ such that 
    \begin{align}
        \frac{1}{|G|}\sum_{g\in G} \mathrm{tr}(M_{gS}\mathcal{U}_{g}(\rho) )\geq 1-\epsilon. \label{eq:thm:conversion_state_discrimination01}
    \end{align}
    We take an arbitrary $S$-covariant map $\mathcal{E}$. 
    We define 
    \begin{align}
        \mathcal{E}'(L):=\sum_{x\in G/S} \mathcal{E}_x(\mathrm{tr}_1((M_x\otimes I) L)) \label{eq:thm:conversion_state_discrimination02}
    \end{align}
    with 
    \begin{align}
        \mathcal{E}_x:=\frac{1}{|S|} \sum_{g\in x} \mathcal{U}_g\circ\mathcal{E}\circ\mathcal{U}_{g^{-1}}\ \forall x\in G/S. \label{eq:thm:conversion_state_discrimination03}
    \end{align}
    Then, for any $g\in G$, we have 
    \begin{align}
        \mathcal{E}'\circ(\mathcal{U}_g\otimes \mathcal{U}'_g)(\rho\otimes \chi) 
        \geq \sum_{x\in G/S} \mathrm{tr}(M_x\mathcal{U}_g(\rho))\mathcal{E}_x\circ\mathcal{U}'_g(\chi) 
        \geq \mathrm{tr}(M_{gS}\mathcal{U}_g(\rho))\mathcal{E}_{gS}\circ\mathcal{U}'_g(\chi). \label{eq:thm:conversion_state_discrimination04}
    \end{align}
    By the definition of $\mathcal{E}_{gS}$, we have 
    \begin{align}
        \mathcal{E}_{gS}\circ\mathcal{U}'_g(\chi) 
        =\frac{1}{|S|} \sum_{s\in S} \mathcal{U}_{gs}\circ\mathcal{E}\circ\mathcal{U}_{(gs)^{-1}}\circ\mathcal{U}_g(\chi) 
        =\mathcal{U}_g\circ\left(\frac{1}{|S|} \sum_{s\in S} \mathcal{U}_s\circ\mathcal{E}\circ\mathcal{U}_{s^{-1}}\right)(\chi) 
        =\mathcal{U}_g\circ\mathcal{E}(\chi), \label{eq:thm:conversion_state_discrimination05}
    \end{align}
    where we used the $S$-covariance of $\mathcal{E}$ in the final equality. 
    By Eqs.~\eqref{eq:thm:conversion_state_discrimination04} and \eqref{eq:thm:conversion_state_discrimination05}, we have 
    \begin{align}
        \mathcal{E}'\circ\left(\mathcal{U}_g\otimes \mathcal{U}'_g\right)(\rho\otimes \chi) 
        \geq \mathrm{tr}(M_{gS}\mathcal{U}_g(\rho)) \mathcal{U}_g\circ\mathcal{E}(\chi), \label{eq:thm:conversion_state_discrimination06}
    \end{align}
    which implies 
    \begin{align}
        \mathcal{U}_{g^{-1}}\circ\mathcal{E}'\circ(\mathcal{U}_g \otimes \mathcal{U}'_g)(\rho\otimes \chi) 
        \geq \mathrm{tr}(M_{gS}\mathcal{U}_g(\rho)) \mathcal{E}(\chi) \label{eq:thm:conversion_state_discrimination07}
    \end{align}
    for all $g\in G$. 
    We define a $G$-covariant CPTP map $\mathcal{F}$ by 
    \begin{align}
        \mathcal{F}:=\frac{1}{|G|}\sum_{g\in G} \mathcal{U}_{g^{-1}}\circ\mathcal{E}'\circ(\mathcal{U}_g \otimes \mathcal{U}'_g). \label{eq:thm:conversion_state_discrimination08}
    \end{align}
    By Eqs.~\eqref{eq:thm:conversion_state_discrimination07}, \eqref{eq:thm:conversion_state_discrimination08}, and \eqref{eq:thm:conversion_state_discrimination01}, we have 
    \begin{align}
        \mathcal{F}(\rho\otimes \chi)
        \geq \left(\frac{1}{|G|}\sum_{g\in G}\mathrm{tr}\left(M_{gS}\mathcal{U}_g(\rho)\right)\right) \mathcal{E}(\chi) 
        \geq (1-\epsilon)\mathcal{E}(\chi), \label{eq:thm:conversion_state_discrimination09}
    \end{align}
    which implies 
    \begin{align}
        D^\mathrm{max}(\mathcal{E}(\chi)\| \mathcal{F}(\rho\otimes\chi)) 
        \leq -\log(1-\epsilon). \label{eq:thm:conversion_state_discrimination10}
    \end{align}
    By the first statement of Corollary~\ref{cor:distingishability_relation_Dmax_general_Dmin}, Eq.~\eqref{eq:thm:conversion_state_discrimination10} implies  
    \begin{align}
        \Delta(\mathcal{E}(\chi), \mathcal{F}(\rho\otimes \chi)) 
        \leq f_\Delta(-\log(1-\epsilon)). \label{eq:thm:conversion_state_discrimination11}
    \end{align}

    Next, we show (iii) $\rightarrow$ (iv). 
    We take some $G$-covariant CPTP map $\mathcal{F}$ such that 
    \begin{align}
        \Delta(\tau_S, \mathcal{F}(\rho))\leq f_\Delta(-\log(1-\epsilon)). \label{eq:thm:conversion_state_discrimination12}
    \end{align}
    Since $f_\Delta$ is strictly increasing, the assumption of the second statement of Corollary~\ref{cor:distingishability_relation_Dmax_general_Dmin} is satisfied. 
    Therefore, by the corollary, Eq.~\eqref{eq:thm:conversion_state_discrimination12} implies 
    \begin{align}
        D^\mathrm{min}(\tau_S \| \mathcal{F}(\rho))\leq -\log(1-\epsilon), \label{eq:thm:conversion_state_discrimination13}
    \end{align}
    which is equivalent to 
    \begin{align}
        \mathrm{tr}(\Pi_{\tau_S}\mathcal{F}(\rho))\geq 1-\epsilon. \label{eq:thm:conversion_state_discrimination14}
    \end{align}
    For each $x\in G/S$, we define $M_x$ by 
    \begin{align}
        M_x:=\mathcal{F}^\dag (\Pi_{\tau_x}). \label{eq:thm:conversion_state_discrimination15}
    \end{align}
    Then, by the positivity and identity-preserving property of $\mathcal{F}^\dag$, we can confirm that $\{M_x\}_{x\in G/S}$ is a POVM. 
    By the definition of $M_x$, we have for any $g\in G$, 
    \begin{align}
        \mathrm{tr}\left(M_{gS} \mathcal{U}_g(\rho)\right) 
        =\mathrm{tr} \left(\mathcal{F}^\dag \circ\mathcal{U}_g(\Pi_{\tau_S} )\cdot\mathcal{U}_g(\rho)\right) 
        =\mathrm{tr} \left(\Pi_{\tau_S}\cdot\mathcal{U}_{g^{-1}}\circ\mathcal{F}\circ\mathcal{U}_g(\rho)\right) 
        =\mathrm{tr}\left(\Pi_{\tau_S}\mathcal{F}(\rho)\right), \label{eq:thm:conversion_state_discrimination16}
    \end{align}
    where we used the $G$-covariance of $\mathcal{F}$ in the last equality. 
    By Eqs.~\eqref{eq:thm:conversion_state_discrimination14} and \eqref{eq:thm:conversion_state_discrimination16}, we get 
    \begin{align}
        \mathrm{tr}\left(M_{gS}\mathcal{U}_g(\rho)\right) 
        \geq 1-\epsilon. \label{eq:thm:conversion_state_discrimination17}
    \end{align}
    Since this holds for all $g\in G$, we have proved that 
    \begin{align}
        \min_{g\in G} \mathrm{tr}\left(M_{gS}\mathcal{U}_g(\rho)\right) 
        \geq 1-\epsilon. \label{eq:thm:conversion_state_discrimination18}
    \end{align}
\end{proof}

We first establish an elementary discrimination protocol for a finite group orbit.
Besides proving asymptotic distinguishability under unrestricted measurements, the construction will serve as a template for the restricted-measurement protocols developed later.
Its basic ingredient is to regard each orbit state as an observable and estimate its expectation value by repeated measurements.

\begin{lemma}[Discrimination of a finite unitary orbit]
    \label{lem:asymmetry_state_discrimination}
    Let $G$ be a finite group, $U$ be a projective unitary representation of $G$ on a finite-dimensional Hilbert space $\mathcal{H}$, and $\rho\in\mathcal{S}(\mathcal{H})$.
    Set $S:=\{g\in G \,|\, \mathcal{U}_g(\rho)=\rho\}$.
    There exists a constant $\alpha>0$ such that, for every positive integer $k$, there is a POVM $\{M_x^{(k)}\}_{x\in G/S}$ on $\mathcal{H}^{\otimes k}$ satisfying
    \begin{align}
        \min_{g\in G}\mathrm{tr}\left(M_{gS}^{(k)}\mathcal{U}_g(\rho)^{\otimes k}\right)\geq& 1-e^{-\alpha k}.
        \label{eq:finite_orbit_discrimination}
    \end{align}
\end{lemma}

\begin{proof}
    Write $X:=G/S$ and $N:=|X|$.
    For each $x=gS\in X$, we define 
    \begin{align}
        &\rho_x:=\mathcal{U}_g(\rho), \\
        &O_x:=\rho_x.
    \end{align}
    The state $\rho_x$ is independent of the choice of representative $g$, and different cosets give different states.
    If $N=1$, the POVM with its only effect equal to the identity has zero error for every $k$.
    We therefore assume $N\geq2$.

    We first construct a binary test for each candidate $x\in X$.
    Set $p_\rho:=\mathrm{tr}\left(\rho^2\right)$.
    Since all orbit states have the same purity, for $x\neq y$ we have
    \begin{align}
        \mathrm{tr}\left(\rho_xO_x\right)-\mathrm{tr}\left(\rho_yO_x\right) 
        =p_\rho-\mathrm{tr}\left(\rho_y\rho_x\right) 
        =\frac{1}{2}\|\rho_x-\rho_y\|_2^2>0.
    \end{align}
    Consequently, the finite set of candidates has a strictly positive separation
    \begin{align}
        \delta_\rho:=& \min_{x,y\in X:\,x\neq y}\left(p_\rho-\mathrm{tr}\left(\rho_y\rho_x\right)\right)>0.
        \label{eq:finite_orbit_observable_gap}
    \end{align}
    Thus, the expectation of $O_x$ is $p_\rho$ on the candidate state $\rho_x$, whereas it is at most $p_\rho-\delta_\rho$ on every other orbit state.

    Fix a positive integer $l$.
    On $l$ independent copies of the unknown state, measure the observable $O_x$ separately on each copy.
    Let
    \begin{align}
        O_x=& \sum_{a\in\Lambda}aP_{x,a}
    \end{align}
    be its spectral decomposition, where the set of eigenvalues $\Lambda$ is the same for all $x$ and is contained in $[0,\|\rho\|_\infty]$.
    Denote the outcomes by $Z_{x,1},\ldots,Z_{x,l}$ and their empirical average by
    \begin{align}
        \widehat\mu_{x,l}:=& \frac1l\sum_{j=1}^{l}Z_{x,j}.
    \end{align}
    If the unknown state is $\rho_y$, then $\mathbb{E}_{\rho_y}\left(Z_{x,j}\right)=\mathrm{tr}\left(\rho_yO_x\right)$.
    We accept the candidate $x$ when $\widehat\mu_{x,l}\geq\theta$, where $\theta:=p_\rho-\delta_\rho/2$.
    The corresponding acceptance effect is the projection
    \begin{align}
        T_x^{(l)}:=& \sum_{\substack{(a_1,\ldots,a_l)\in\Lambda^l\\ l^{-1}\sum_{j=1}^{l}a_j\geq\theta}}P_{x,a_1}\otimes\cdots\otimes P_{x,a_l}.
        \label{eq:finite_orbit_binary_test}
    \end{align}
    For the correct candidate, rejection requires a downward deviation of at least $\delta_\rho/2$ from the mean.
    For an incorrect candidate, acceptance requires an upward deviation of at least $\delta_\rho/2$.
    The one-sided Hoeffding inequality~\cite{hoeffding1963probability} therefore gives
    \begin{align}
        &\mathrm{tr}\left(\left(I-T_x^{(l)}\right)\rho_x^{\otimes l}\right) 
        \leq e^{-c_\rho l}, \\
        &\mathrm{tr}\left(T_x^{(l)}\rho_y^{\otimes l}\right) 
        \leq e^{-c_\rho l}\ \text{for }y\neq x,
        \label{eq:finite_orbit_binary_errors}
    \end{align}
    where
    \begin{align}
        c_\rho:=& \frac{\delta_\rho^2}{2\|\rho\|_\infty^2}>0.
    \end{align}

    We next combine these binary tests to identify the unknown coset.
    Divide $Nl$ copies into $N$ disjoint blocks of $l$ copies, and apply the test $\{T_x^{(l)},I-T_x^{(l)}\}$ to the block labeled by $x$.
    In particular, observables associated with different candidates are never required to be measured on the same copy.
    Write $b=(b_x)_{x\in X}\in\{0,1\}^X$ for the resulting bit string, with $b_x=1$ denoting acceptance.
    The effect corresponding to $b$ is
    \begin{align}
        &F_b^{(l)}:=\bigotimes_{x\in X}R_{x,b_x}^{(l)}, \\
        &R_{x,1}^{(l)}:=T_x^{(l)}, \\
        &R_{x,0}^{(l)}:=I-T_x^{(l)}. 
    \end{align}
    Let $e_x$ denote the bit string whose only nonzero entry is at $x$.
    Choose a decoding function $d:\{0,1\}^X\to X$ satisfying $d(e_x)=x$ for every $x$, and assign all remaining bit strings to an arbitrary fixed candidate.
    This defines a POVM
    \begin{align}
        \widetilde M_x^{(l)}:=& \sum_{b:\,d(b)=x}F_b^{(l)}.
    \end{align}
    For an input $\rho_x^{\otimes Nl}$, decoding is correct whenever the test for $x$ accepts and all other tests reject.
    By the union bound and Eq.~\eqref{eq:finite_orbit_binary_errors},
    \begin{align}
        \mathrm{tr}\left(\widetilde M_x^{(l)}\rho_x^{\otimes Nl}\right)\geq& 1-Ne^{-c_\rho l}
        \label{eq:finite_orbit_block_error}
    \end{align}
    for every $x\in X$.

    For a total of $k\geq N$ copies, set $l:=\lfloor k/N\rfloor$, use the above protocol on $Nl$ copies, and discard the remaining copies.
    The resulting POVM satisfies
    \begin{align}
        1-\min_{x\in X}\mathrm{tr}\left(M_x^{(k)}\rho_x^{\otimes k}\right)
        \leq Ne^{-c_\rho\lfloor k/N\rfloor}
        \leq Ne^{c_\rho}e^{-c_\rho k/N}.
    \end{align}
    Hence, with $\alpha_0:=c_\rho/(2N)$, there is an integer $k_0\geq N$ such that the error is at most $e^{-\alpha_0k}$ for all $k\geq k_0$.
    For $1\leq k<k_0$, instead use the uniformly random guess $M_x^{(k)}:=I/N$.
    Taking
    \begin{align}
        \alpha:=& \min\left\{\alpha_0,-\frac{1}{k_0}\log\left(1-\frac{1}{N}\right)\right\}>0
    \end{align}
    establishes Eq.~\eqref{eq:finite_orbit_discrimination} for every positive integer $k$.
\end{proof}

The construction separates a measurement-dependent step from a purely classical decoding step.
The former produces a bounded unbiased estimator of $\mathrm{tr}\left(\tau O_x\right)$ for each candidate, while the latter combines the resulting binary tests on disjoint blocks.
In the restricted resource theories considered below, we replace the direct spectral measurement of $O_x$ by an allowed measurement with the same expectation value and a uniformly bounded outcome.
Independent repetition then gives the same thresholding and decoding construction, with an error exponent determined by the corresponding outcome bound.

We next recall the asymptotic state-conversion structure of finite-group asymmetry.
For a finite group, the symmetry subgroup completely determines whether asymptotic conversion is possible: if the symmetry subgroup of the input is contained in that of the output, arbitrarily large finite rates are achievable, whereas otherwise no positive rate is possible.
This simple structure will provide the symmetry-side obstruction in the multi-resource theories considered below.

\begin{theorem} \label{thm:asymmetry_rate}
    Let $G$ be a finite group, $U$ and $U'$ be projective unitary representations of $G$ on the input and output space $\mathcal{H}$ and $\mathcal{H}'$, and $\rho$ and $\sigma$ be input and output states. 
    Then, 
    \begin{align}
        R_{G\text{-}\mathrm{cov}}(\rho\to\sigma)= 
        \begin{cases}
            \infty & \text{if }\mathrm{Sym}(\rho)\subset\mathrm{Sym}(\sigma), \\
            0 & \text{otherwise}. 
        \end{cases} \label{eq:thm:asymmetry_rate01}
    \end{align}
\end{theorem}

\begin{proof}
    First, we consider the case where $\mathrm{Sym}(\rho)\subset\mathrm{Sym}(\sigma)$. 
    By Lemma~\ref{lem:asymmetry_state_discrimination}, there exists some sequence of POVMs $\{M_x^{(n)}\}_{x\in G/S}$ such that 
    \begin{align}
        \mathrm{tr}\left(M_{gS}^{(n)}\mathcal{U}_{gS}\left(\rho^{\otimes n}\right)\right) 
        \geq 1-e^{-\alpha n}, \label{eq:thm:asymmetry_rate02}
    \end{align}
    By Theorem~\ref{thm:conversion_state_discrimination}, this implies that for any $r\in (0, \infty)$, there exists some $G$-covariant CPTP map $\mathcal{E}$ such that 
    \begin{align}
        \mathrm{T}(\mathcal{E}(\rho^{\otimes n}), \sigma^{\otimes \lfloor rn \rfloor})\leq e^{-\alpha n}, \label{eq:thm:asymmetry_rate03}
    \end{align}
    which implies $R_{G\text{-}\mathrm{cov}}(\rho\to\sigma)=\infty$.

    Next, we consider the case where $\mathrm{Sym}(\rho)\not\subset\mathrm{Sym}(\sigma)$. 
    We take some $g\in G$ such that $g\in\mathrm{Sym}(\rho)$ but $g\not\in\mathrm{Sym}(\sigma)$. 
    We suppose that there exists some $n\in\mathbb{N}$ and some $G$-covariant CPTP map $\mathcal{E}$ such that 
    \begin{align}
        \mathrm{T}\left(\mathcal{E}(\rho^{\otimes n}), \sigma\right) 
        <\frac{1}{2}\mathrm{T}(\sigma, \mathcal{U}'_g(\sigma)). \label{eq:thm:asymmetry_rate04}
    \end{align}
    By the triangle inequality, we have 
    \begin{align}
        \mathrm{T}\left(\sigma, \mathcal{U}_g(\sigma)\right) 
        \leq \mathrm{T}\left(\mathcal{E}(\rho^{\otimes n}), \sigma\right)
        +\mathrm{T}\left(\mathcal{E}(\rho^{\otimes n}), \mathcal{U}'_g(\sigma) \right). \label{eq:thm:asymmetry_rate05}
    \end{align}
    Since $g$ is a symmetry subgroup of $\rho$, we have 
    \begin{align}
        \mathrm{T}\left(\mathcal{E}(\rho^{\otimes n}), \mathcal{U}'_g(\sigma) \right) 
        =\mathrm{T}\left(\mathcal{E}\circ\mathcal{U}_g(\rho)^{\otimes n}, \mathcal{U}'_g(\sigma)\right)
        =\mathrm{T} \left(\mathcal{U}_{g^{-1}}\circ\mathcal{E}\circ\mathcal{U}_g(\rho)^{\otimes n}, \sigma\right) 
        =\mathrm{T}\left(\mathcal{E}(\rho^{\otimes n}), \sigma \right), \label{eq:thm:asymmetry_rate06}
    \end{align}
    where we used the $G$-covariance in the last equality. 
    By Eqs.~\eqref{eq:thm:asymmetry_rate05} and \eqref{eq:thm:asymmetry_rate06}, we get 
    \begin{align}
        \mathrm{T}\left(\sigma, \mathcal{U}'_g(\sigma)\right) 
        \leq 2\mathrm{T}\left(\mathcal{E}(\rho^{\otimes n}), \sigma \right), \label{eq:thm:asymmetry_rate07}
    \end{align}
    which contradicts with Eq.~\eqref{eq:thm:asymmetry_rate04}. 
    Therefore, for any $n\in\mathbb{N}$, there does not exist any $G$-covariant CPTP map $\mathcal{E}$ satisfying Eq.~\eqref{eq:thm:asymmetry_rate04}, which implies $R_{G\text{-}\mathrm{cov}}(\rho\to\sigma)=0$. 
\end{proof}

\section{Multi-resource theories} \label{sec:multi_resource}

We now consider a general resource theory subject to an additional finite-group symmetry constraint.
Our aim is to identify conditions under which imposing the two restrictions simultaneously does not reduce the asymptotic conversion rate beyond the restrictions imposed separately.
We formulate this principle in terms of control-admissible measurements and then derive a general lower bound on the conversion rate.

\subsection{Setup}

First, we define general resource theories.

\begin{definition}
    Let $\{\mathfrak{S}^\mathrm{A}\}_\mathrm{A}$ be a family of state sets on system $\mathrm{A}$ and $\{\mathfrak{O}^{\mathrm{A}\to\mathrm{B}}\}_{\mathrm{A}, \mathrm{B}}$ be a family of operation sets from system $\mathrm{A}$ to system $\mathrm{B}$. 
    We say that a pair $(\{\mathfrak{S}^\mathrm{A}\}_\mathrm{A}, \{\mathfrak{O}^{\mathrm{A}\to\mathrm{B}}\}_{\mathrm{A}, \mathrm{B}})$ is a resource theory if it satisfies 
    \begin{align}
        &\forall \rho\in\mathfrak{S}^{\mathrm{A}},\ \forall \mathcal{E}\in\mathfrak{O}^{\mathrm{A}\to\mathrm{B}},\ \mathcal{E}(\rho)\in\mathfrak{S}^\mathrm{B}, \\
        &\mathrm{id}^\mathrm{A}\in\mathfrak{O}^{\mathrm{A}\to\mathrm{A}}, \\
        &\forall \mathcal{E}_1\in\mathfrak{O}^{\mathrm{A}\to\mathrm{B}},\ \mathcal{E}_2\in\mathfrak{O}^{\mathrm{B}\to\mathrm{C}},\ \mathcal{E}_2\circ\mathcal{E}_1\in\mathfrak{O}^{\mathrm{A}\to\mathrm{C}}. 
    \end{align}
\end{definition}

Throughout this paper, we restrict our attention to the resource theories that satisfy the following conditions.

\begin{assumption}
    Let $(\{\mathfrak{S}^\mathrm{A}\}_\mathrm{A}, \{\mathfrak{O}^{\mathrm{A}\to\mathrm{B}}\}_{\mathrm{A}, \mathrm{B}})$ be a resource theory. 
    We assume that the resource theory satisfies the following four properties. \\
    (i) Tensor closure: 
    \begin{align}
        &\forall \rho\in\mathfrak{S}^\mathrm{A}_1, \ \sigma\in\mathfrak{S}^\mathrm{A}_2, \ \rho\otimes\sigma\in\mathfrak{S}^{\mathrm{A}_1\mathrm{A}_2}, \\
        &\forall\mathcal{E}_1\in\mathfrak{O}^{\mathrm{A}_1\to\mathrm{B}_1}, \ \forall\mathcal{E}_2\in\mathfrak{O}^{\mathrm{A}_2\to\mathrm{B}_2}, \ \mathcal{E}_1\otimes \mathcal{E}_2\in\mathfrak{O}^{\mathrm{A}_1\mathrm{A}_2\to\mathrm{B}_1\mathrm{B}_2}
    \end{align}
    (ii) Ancilla preparation: 
    \begin{align}
        \forall\sigma\in\mathfrak{S}^\mathrm{X},\ \mathcal{E}: \rho\mapsto\rho\otimes \sigma\in\mathfrak{O}^{\mathrm{A}\to\mathrm{AX}}. 
    \end{align}
    (iii) Closure under discarding: 
    \begin{align}
        \mathrm{tr}_\mathrm{X}\in\mathfrak{O}^{\mathrm{AX}\to\mathrm{A}}. 
    \end{align}
    (iv) Convexity: 
    \begin{align}
        \forall\mathcal{E}_1, \mathcal{E}_2\in\mathfrak{O}^{\mathrm{A}\to\mathrm{{B}}},\ \forall p\in[0, 1],\ p\mathcal{E}_1+(1-p)\mathcal{E}_2\in\mathfrak{O}^{\mathrm{A}\to\mathrm{{B}}}. 
    \end{align}
\end{assumption}

Next, we introduce symmetry in resource theories. 
We assume that we are given some fixed group $G$ and every physical system $\mathrm{A}$ has projective unitary representation $U^\mathrm{A}$ of $G$.  
In this paper, we deal with the symmetry constraints that satisfy the following conditions.

\begin{assumption} \label{asm:resource_symmetry_compatibility}
    Let $(\{\mathfrak{S}^\mathrm{A}\}_\mathrm{A}, \{\mathfrak{O}^{\mathrm{A}\to\mathrm{B}}\}_{\mathrm{A}, \mathrm{B}})$ be a resource theory, $G$ be a group and $\{U^\mathrm{A}\}_\mathrm{A}$ be a family of projective unitary representations $U^\mathrm{A}$ of $G$ in the system $\mathrm{A}$. 
    We assume that the resource theory $(\{\mathfrak{S}^\mathrm{A}\}_\mathrm{A}, \{\mathfrak{O}^{\mathrm{A}\to\mathrm{B}}\}_{\mathrm{A}, \mathrm{B}})$ is compatible with the symmetry $(G, \{U^\mathrm{A}\}_\mathrm{A})$ if 
    \begin{align}
        &\mathcal{U}_g^\mathrm{A}(\mathfrak{S}^\mathrm{A})=\mathfrak{S}^\mathrm{A}, \\
        &\mathcal{U}_g^\mathrm{B} \circ\mathfrak{O}^{\mathrm{A}\to\mathrm{B}} \circ\mathcal{U}_{g^{-1}}^\mathrm{A}=\mathfrak{O}^{\mathrm{A}\to\mathrm{B}}. 
    \end{align}
    for all $g\in G$. 
\end{assumption}

This condition means that the sets of free states and free operations does not change depending on the symmetry operations. 
We stress that this condition does not imply that every free state is $G$-invariant or every free operation is $G$-covariant.

Finally, We define the symmetry constrained version of resource theories.

\begin{definition}
    Let $G$ be a group, and $(\{\mathfrak{S}^\mathrm{A}\}_\mathrm{A}, \{\mathfrak{O}^{\mathrm{A}\to\mathrm{B}}\}_{\mathrm{A}, \mathrm{B}})$ be a $G$-compatible resource theory. 
    The free states in the $G$-constrained resource theory $\mathrm{R}_{G\textrm{-}\mathrm{sym}}$ is defined by free states $\{\mathfrak{S}_{G\textrm{-}\mathrm{inv}}^\mathrm{A}\}_\mathrm{A}$ and free operations $\{\mathfrak{O}^{G\textrm{-}\mathrm{cov}}(\mathrm{A}\to\mathrm{B})\}_{\mathrm{A}, \mathrm{B}}$, which are defined by 
    \begin{align}
        &\mathfrak{S}_{G\textrm{-}\mathrm{inv}}^\mathrm{A}:=\{\rho\in\mathfrak{S}^\mathrm{A}\ |\ \forall g\in G,\ \mathcal{U}_g^\mathrm{A}(\rho)=\rho\}, \\
        &\mathfrak{O}_{G\textrm{-}\mathrm{cov}}^{\mathrm{A}\to\mathrm{B}}:=\{\mathcal{E}\in\mathfrak{O}^{\mathrm{A}\to\mathrm{B}}\ |\ \forall g\in G,\ \mathcal{E}\circ\mathcal{U}_g^\mathrm{A}=\mathcal{U}_g^\mathrm{B}\circ\mathcal{E}\}. 
    \end{align}
\end{definition}

We will compare asymptotic state-conversion rates under the original resource theory and under its symmetry-constrained counterpart.
We therefore introduce the notation used throughout the paper.

\begin{definition}[Asymptotic conversion rate]
    \label{def:conversion_rate}
    Let $(\{\mathfrak{S}^\mathrm{A}\}_\mathrm{A}, \{\mathfrak{O}^{\mathrm{A}\to \mathrm{B}}\}_\mathrm{A,B})$ be a resource theory.
    For states $\rho$ and $\sigma$, a rate $r\geq0$ is achievable under $\mathfrak{O}$ if there exists a sequence of channels $\mathcal{E}^{(n)}\in\mathfrak{O}$ such that
    \begin{align}
        \lim_{n\to\infty}
        \mathrm{T}\left(
        \mathcal{E}^{(n)}\left(\rho^{\otimes n}\right),
        \sigma^{\otimes\lfloor rn\rfloor}
        \right)=&0.
    \end{align}
    The asymptotic conversion rate from $\rho$ to $\sigma$ under $\mathfrak{O}$ is defined by
    \begin{align}
        R_{\mathfrak{O}}(\rho\to\sigma):=&
        \sup\{r\geq0 \,|\, r\textrm{ is achievable under }\mathcal{O}\}.
    \end{align}
\end{definition}

For the symmetry-constrained resource theory introduced below, we write
\begin{align}
    R_{G\text{-cov}, \mathfrak{O}}(\rho\to\sigma)
\end{align}
for the corresponding conversion rate.
Thus, the central question of this work is whether the relation 
\begin{align}
    R_{{G\text{-cov}}, \mathfrak{O}}(\rho\to\sigma)
    =R_\mathfrak{O}(\rho\to\sigma) 
\end{align}
holds or not. 
Since $\mathfrak{O}_{G\text{-cov}}\subset\mathfrak{O}$, we always have
\begin{align}
    R_{{G\text{-cov}}, \mathfrak{O}}(\rho\to\sigma)\leq
    R_\mathfrak{O}(\rho\to\sigma).
\end{align}
The purpose of the results below is to identify conditions under which the converse inequality holds.

\subsection{Control-admissible POVM}

The discrimination step required for state conversion involves more than obtaining a classical estimate of the group element.
The measurement outcome must also be usable to select a subsequent free operation without leaving the underlying resource theory.
We therefore introduce an operational class of measurements defined by this conditional-control property.

\begin{definition} \label{def:control_admissible_POVM}
    Let $X$ be a finite set and $\{M_x\}_{x\in X}$ be a sub-POVM. 
    $\{M_x\}_{x\in X}$ is called control-admissible sub-POVM if a map defined by 
    \begin{align}
        \mathcal{E}(L):=\sum_{x\in X\cup\{\perp\}} \mathcal{E}_x\left(\mathrm{tr}_1((M_x\otimes I)L)\right) 
    \end{align}
    is a free operation for all free operations $\mathcal{E}_x$ with $x\in X\cup\{\perp\}$, where 
    \begin{align}
        M_\perp:=I-\sum_{x\in X} M_x. 
    \end{align}
\end{definition}

We now show how accessible symmetry information can be converted into covariance.
Suppose that the input orbit can be identified by a control-admissible measurement and that the desired free operation is already covariant under a subgroup $K$.
The measurement outcome can then be used to rotate this operation into the appropriate symmetry sector, after which group averaging produces a fully $G$-covariant free operation.
The following lemma quantifies the resulting approximation error.

\begin{theorem} \label{thm:conversion_control_admissible_POVM}
    Let $G$ be a finite group, $K$ be a subgroup of $G$, $\epsilon\in [0, 1)$, $\rho$ be a state, there exist some control-admissible sub-POVM $\{M_x\}_{x\in G/K}$ such that 
    \begin{align}
        \frac{1}{|G|} \sum_{g\in G} \mathrm{tr}(M_{gK} \mathcal{U}_g(\rho))\geq 1-\epsilon, 
    \end{align}
    $\mathcal{E}\in\mathfrak{O}$ be a free operation and $K$-covariant. 
    Then, there exists some $G$-covariant free operation $\mathcal{F}\in\mathfrak{O}_{G\text{-}\mathrm{cov}}$ such that for any state $\chi$, 
    \begin{align}
        \Delta(\mathcal{E}(\chi), \mathcal{F}(\rho\otimes \chi))\leq f_\Delta(-\log(1-\epsilon)). 
    \end{align}
\end{theorem}

We can prove this theorem in a similar way to the former half of the proof of Theorem~\ref{thm:conversion_state_discrimination}.

\begin{proof}
    We take an arbitrary sub-POVM $\{M_x\}_{x\in G/K}$ such that 
    \begin{align}
        \frac{1}{|G|}\sum_{g\in G} \mathrm{tr}(M_{gK}\mathcal{U}_g(\rho) )\geq 1-\epsilon. \label{eq:lem:conversion_control_admissible_POVM01}
    \end{align}
    We take an arbitrary free operation $\mathcal{E}\in\mathfrak{O}$ and assume that $\mathcal{E}$ is $K$-covariant. 
    We define 
    \begin{align}
        \mathcal{E}'(L):=\sum_{x\in G/K\cup\{\perp\}} \mathcal{E}_x(\mathrm{tr}_1((M_x\otimes I) L)) \label{eq:lem:conversion_control_admissible_POVM02}
    \end{align}
    with 
    \begin{align}
        &\mathcal{E}_x:=\frac{1}{|K|} \sum_{g\in x} \mathcal{U}_g\circ\mathcal{E}\circ\mathcal{U}_{g^{-1}}\ \forall x\in G/K, \label{eq:lem:conversion_control_admissible_POVM03}\\
        &\mathcal{E}_\perp:=\frac{|K|}{|G|} \sum_{x\in G/K} \mathcal{E}_x, \\
        &M_\perp:=I-\sum_{x\in G/K} M_x. 
    \end{align}
    Since we assume that the resource theory satisfies the compatibility condition given in Assumption~\ref{asm:resource_symmetry_compatibility}, and the the class of free operations is closed under tensor product, $\mathcal{E}_x$ is a free operation. 
    Since we assume that the class of free operations is closed under convexity, $\mathcal{E}_\perp$ is also a free operation. 
    By the definition of control-admissible POVM, $\mathcal{E}'$ is a free operation. 
    By Eq.~\eqref{eq:lem:conversion_control_admissible_POVM02}, for any $g\in G$, we have 
    \begin{align}
        \mathcal{E}'\circ(\mathcal{U}_g\otimes \mathcal{U}'_g)(\rho\otimes \chi) 
        \geq \sum_{x\in G/K\cup\{\perp\}} \mathrm{tr}(M_x\mathcal{U}_g(\rho))\mathcal{E}_x\circ\mathcal{U}'_g(\chi) 
        \geq \mathrm{tr}(M_{gK}\mathcal{U}_g(\rho))\mathcal{E}_{gK}\circ\mathcal{U}'_g(\chi). \label{eq:lem:conversion_control_admissible_POVM04}
    \end{align}
    By the definition of $\mathcal{E}_{gK}$, we have 
    \begin{align}
        \mathcal{E}_{gK}\circ\mathcal{U}'_g(\chi) 
        =\frac{1}{|K|} \sum_{k\in K} \mathcal{U}_{gk}\circ\mathcal{E}\circ\mathcal{U}_{(gk)^{-1}}\circ\mathcal{U}_g(\chi) 
        =\mathcal{U}_g\circ\left(\frac{1}{|K|} \sum_{k\in K} \mathcal{U}_k\circ\mathcal{E}\circ\mathcal{U}_{k^{-1}}\right)(\chi) 
        =\mathcal{U}_g\circ\mathcal{E}(\chi), \label{eq:lem:conversion_control_admissible_POVM05}
    \end{align}
    where we used the $K$-covariance of $\mathcal{E}$ in the final equality. 
    By Eqs.~\eqref{eq:lem:conversion_control_admissible_POVM04} and \eqref{eq:lem:conversion_control_admissible_POVM05}, we have 
    \begin{align}
        \mathcal{E}'\circ\left(\mathcal{U}_g\otimes \mathcal{U}'_g\right)(\rho\otimes \chi) 
        \geq \mathrm{tr}(M_{gK}\mathcal{U}_g(\rho)) \mathcal{U}_g\circ\mathcal{E}(\chi), \label{eq:lem:conversion_control_admissible_POVM06}
    \end{align}
    which implies 
    \begin{align}
        \mathcal{U}_{g^{-1}}\circ\mathcal{E}'\circ(\mathcal{U}_g \otimes \mathcal{U}'_g)(\rho\otimes \chi) 
        \geq \mathrm{tr}(M_{gK}\mathcal{U}_g(\rho)) \mathcal{E}(\chi) \label{eq:lem:conversion_control_admissible_POVM07}
    \end{align}
    for all $g\in G$. 
    We define a CPTP map $\mathcal{F}$ by 
    \begin{align}
        \mathcal{F}:=\frac{1}{|G|}\sum_{g\in G} \mathcal{U}_{g^{-1}}\circ\mathcal{E}'\circ(\mathcal{U}_g \otimes \mathcal{U}'_g). \label{eq:lem:conversion_control_admissible_POVM08}
    \end{align}
    We note that $\mathcal{F}$ is a free operation and $G$-covariant. 
    By Eqs.~\eqref{eq:lem:conversion_control_admissible_POVM07}, \eqref{eq:lem:conversion_control_admissible_POVM08}, and \eqref{eq:lem:conversion_control_admissible_POVM01}, we have 
    \begin{align}
        \mathcal{F}(\rho\otimes \chi)
        \geq \left(\frac{1}{|G|}\sum_{g\in G}\mathrm{tr}\left(M_{gK}\mathcal{U}_g(\rho)\right)\right) \mathcal{E}(\chi) 
        \geq (1-\epsilon)\mathcal{E}(\chi), \label{eq:lem:conversion_control_admissible_POVM09}
    \end{align}
    which implies 
    \begin{align}
        D^\mathrm{max}(\mathcal{E}(\chi)\| \mathcal{F}(\rho\otimes\chi)) 
        \leq -\log(1-\epsilon). \label{eq:lem:conversion_control_admissible_POVM10}
    \end{align}
    By the first statement of Corollary~\ref{cor:distingishability_relation_Dmax_general_Dmin}, Eq.~\eqref{eq:lem:conversion_control_admissible_POVM10} implies  
    \begin{align}
        \Delta(\mathcal{E}(\chi), \mathcal{F}(\rho\otimes \chi)) 
        \leq f_\Delta(-\log(1-\epsilon)). \label{eq:lem:conversion_control_admissible_POVM11}
    \end{align}
\end{proof}

Theorem~\ref{thm:conversion_control_admissible_POVM} is a one-shot statement.
We now apply it to asymptotic state conversion.
The key observation is that the copies used to identify the symmetry sector can be chosen to form a vanishing fraction of the input, while their number still diverges.
Consequently, asymptotically vanishing discrimination error is sufficient to recover the original first-order conversion rate.

\begin{theorem}[General conversion theorem] \label{thm:rate_lower_bound_general}
    Let $G$ be a finite group, $K$ be a subgroup of $G$, $r>0$, there exist a sequence of free operations $(\mathcal{E}^{(n)})_{n\in\mathbb{N}}$ be such that 
    \begin{align}
        \lim_{n\to\infty} \frac{1}{|K|}\sum_{k\in K} \mathrm{T}
        \left(
            \mathcal{E}^{(n)}\left(\mathcal{U}_k(\rho)^{\otimes n}\right), \mathcal{U}'_k(\sigma)^{\otimes \lfloor rn\rfloor}
        \right)
        =0, \label{eq:thm:rate_lower_bound_general01}
    \end{align}
    and there exist a sequence of control-admissible sub-POVMs $(\{M_x^{(l)}\}_{x\in G/K})_{l\in\mathbb{N}}$ such that 
    \begin{align}
        \lim_{l\to\infty} \frac{1}{|G|}\sum_{g\in G}\mathrm{tr}
        \left(
            M_{gK}^{(l)} \mathcal{U}_g(\rho)^{\otimes l}
        \right)
        =1. \label{eq:thm:rate_lower_bound_general02}
    \end{align}
    Then, 
    \begin{align}
        R_{G\text{-}\mathrm{cov}, \mathfrak{O}}(\rho\to\sigma)\geq r. \label{eq:thm:rate_lower_bound_general03}
    \end{align}
\end{theorem}

\begin{proof}
    For convenience, we denote 
    \begin{align}
        &\delta_n:=\frac{1}{|K|}\sum_{k\in K} \mathrm{T}
        \left(
            \mathcal{E}^{(n)}\left(\mathcal{U}_k(\rho)^{\otimes n}\right), \mathcal{U}'_k(\sigma)^{\otimes \lfloor rn\rfloor}
        \right), \\
        &\epsilon_l:=1-\frac{1}{|G|}\sum_{g\in G}\mathrm{tr}
        \left(
            M_{gK}^{(l)} \mathcal{U}_g(\rho)^{\otimes l}
        \right). 
    \end{align}
    We take arbitrary $r'\in (0, r)$ and an arbitrary sequence $(l_n)_{n\in\mathbb{N}}$ that satisfies $l_n\in (0, n)\cap\mathbb{N}$, $\lim_{n\to\infty} l_n=\infty$, and $\lim_{n\to\infty} l_n/n=0$. 
    We define 
    \begin{align}
        {\mathcal{E}'}^{(n)}:=\frac{1}{|K|}\sum_{k\in K} {\mathcal{U}_{k^{-1}}'}^{\otimes \lfloor rn \rfloor}\circ\mathcal{E}^{(n)}\circ\mathcal{U}_k^{\otimes n}. 
    \end{align}
    Then, we have 
    \begin{align}
        \mathrm{T}\left({\mathcal{E}'}^{(n)}(\rho^{\otimes n}), \sigma^{\otimes \lfloor rn \rfloor}\right) 
        =&\mathrm{T}\left(\frac{1}{|K|}\sum_{k\in K} {\mathcal{U}_{k^{-1}} '}^{\otimes \lfloor rn \rfloor}\circ\mathcal{E}^{(n)}\circ\mathcal{U}_k^{\otimes n}(\rho^{\otimes n}), \sigma^{\otimes \lfloor rn \rfloor}\right) \nonumber\\
        \leq &\frac{1}{|K|}\sum_{k\in K} \mathrm{T}\left({\mathcal{U}_{k^{-1}} '}^{\otimes \lfloor rn \rfloor}\circ\mathcal{E}^{(n)}\circ\mathcal{U}_k^{\otimes n}(\rho^{\otimes n}), \sigma^{\otimes \lfloor rn \rfloor}\right) \nonumber\\
        =&\frac{1}{|K|}\sum_{k\in K} \mathrm{T}
        \left(
            \mathcal{E}^{(n)}\left(\mathcal{U}_k(\rho)^{\otimes n}\right), \mathcal{U}'_k(\sigma)^{\otimes \lfloor rn\rfloor}
        \right) \nonumber\\
        =&\delta_n. \label{eq:thm:rate_lower_bound_general04}
    \end{align}
    For sufficiently large $n\in\mathbb{N}$ such that $\lfloor r(n-l_n)\rfloor\geq \lfloor r'n\rfloor$, we can take a CPTP map $\mathcal{N}$ that discard $\lfloor r(n-l_n)\rfloor-\lfloor r'n\rfloor$ copies out of $\lfloor r(n-l_n)\rfloor$ copies. 
    Then, by the monotonicity of the trace distance, we have 
    \begin{align}
        \mathrm{T}\left(\mathcal{N}\circ{\mathcal{E}'}^{(n-l_n)}\left(\rho^{\otimes n-l_n}\right), \sigma^{\otimes \lfloor r'n \rfloor}\right) 
        =&\mathrm{T}\left(\mathcal{N}\circ{\mathcal{E}'}^{(n-l_n)}\left(\rho^{\otimes n-l_n}\right), \mathcal{N}\left(\sigma^{\otimes \lfloor r(n-l_n)\rfloor}\right)\right) \nonumber\\
        \leq &\mathrm{T}\left({\mathcal{E}'}^{(n-l_n)}\left(\rho^{\otimes n-l_n}\right), \sigma^{\otimes \lfloor r(n-l_n) \rfloor}\right). \label{eq:thm:rate_lower_bound_general05}
    \end{align}
    By Eqs.~\eqref{eq:thm:rate_lower_bound_general04} and \eqref{eq:thm:rate_lower_bound_general05}, we get 
    \begin{align}
        \mathrm{T}\left(\mathcal{N}\circ{\mathcal{E}'}^{(n-l_n)}\left(\rho^{\otimes n-l_n}\right), \sigma^{\otimes \lfloor r'n \rfloor}\right) 
        \leq \delta_{n-l_n}. 
    \end{align}
    We note that $\mathcal{N}\circ{\mathcal{E}'}^{(n-l_n)}$ is $K$-covariant. 
    By Theorem~\ref{thm:conversion_control_admissible_POVM}, there exists some $G$-covariant free operation $\mathcal{F}^{(n)}$ such that 
    \begin{align}
        \Delta(\mathcal{N}\circ{\mathcal{E}'}^{(n-l_n)}(\rho^{\otimes n-l_n}), \mathcal{F}^{(n)}(\rho^{\otimes n}))
        =\Delta(\mathcal{N}\circ{\mathcal{E}'}^{(n-l_n)}(\rho^{\otimes n-l_n}), \mathcal{F}^{(n)}(\rho^{\otimes l_n}\otimes \rho^{\otimes n-l_n}))
        \leq f_\Delta(-\log(1-\epsilon_{l_n})). 
    \end{align}
    Since $f_\Delta(x)=1-e^{-x}$ when $\Delta$ is the trace distance, we have 
    \begin{align}
        \mathrm{T}\left(\mathcal{N}\circ{\mathcal{E}'}^{(n-l_n)}(\rho^{\otimes n-l_n}), \mathcal{F}^{(n)}(\rho^{\otimes n})\right) 
        \leq \epsilon_{l_n}. 
    \end{align}
    By the triangle inequality, we get 
    \begin{align}
        \mathrm{T}\left(\mathcal{F}^{(n)}(\rho^{\otimes n}), \sigma^{\otimes \lfloor r'n \rfloor}\right)\leq \delta_{n-l_n}+\epsilon_{l_n}. 
    \end{align}
    Since the right-hand side tends to $0$ in the limit of $n\to\infty$, we have 
    \begin{align}
        \lim_{n\to\infty} \mathrm{T}\left(\mathcal{F}^{(n)}(\rho^{\otimes n}), \sigma^{\otimes \lfloor r'n\rfloor}\right)=0. 
    \end{align}
    Therefore, we get 
    \begin{align}
        R_{G\text{-}\mathrm{cov}, \mathfrak{O}}(\rho\to\sigma)\geq r'. 
    \end{align}
    Since this holds for all $r'\in (0, r)$, we have 
    \begin{align}
        R_{G\text{-}\mathrm{cov}, \mathfrak{O}}(\rho\to\sigma)\geq r. 
    \end{align}
\end{proof}

\subsection{Memory-assisted POVM}

Control admissibility concerns the operation obtained by using a measurement outcome to select a free operation on another system.
It does not require that the outcome be stored in a physical register carrying a prescribed symmetry representation.
We now introduce a formulation in which the measurement information is encoded in an explicit physical memory.
Besides providing conditional control, such a memory allows us to formulate the distillation of a symmetry reference and its retention during a subsequent state transformation.

The relation between the two formulations depends on the free operations available on the memory.
Under the standardization and control properties introduced below, every complete memory-assisted POVM is control-admissible.
For sub-POVMs, the complementary, inconclusive sector must also be taken into account, and we will distinguish the conditions imposed on this sector from those imposed on the memory states themselves.
The conversion principle of the preceding subsection requires only control admissibility; the additional memory structure is used for statements that explicitly involve a physical record.
We first define the memory states and the free operations needed to use them.

\begin{definition}
    Let $\mathrm{M}$ be a system, $\{\omega_x\}_{x\in G/S}$ be a set of states on $\mathrm{M}$, $G$ be a finite group, and $S$ be a subgroup of $G$. 
    We say that a pair $\{\omega_x\}_{x\in G/S}$ is a $G/S$-memory if it satisfies the following two conditions: \\
    (i) $G$-equivariance: 
    \begin{align}
        \forall g, h\in G,\ \mathcal{U}_g(\omega_{hS})=\omega_{ghS}, 
    \end{align}
    (ii) Orthogonality: 
    \begin{align}
        \forall x_1, x_2\in G/S,\ \omega_{x_1}\omega_{x_2}=0 \textrm{ if }x_1\neq x_2. 
    \end{align}
\end{definition}

We impose the following additional assumption when a physical memory realization is required.
It ensures that the memory states can be freely introduced, standardized after an arbitrary encoding, and used to control subsequent free operations.
To accommodate sub-POVMs, we also introduce an orthogonal failure sector corresponding to the inconclusive outcome.

\begin{assumption}[Free memory]
    \label{ass:free_memory}
    Let $\{\omega_x\}_{x\in X}$ be memory states on a system $M$ with mutually orthogonal supports, and let $\Pi_x$ denote the projection onto the support of $\omega_x$.
    We assume that there exists a state $\omega_\perp$ whose support is orthogonal to those of all $\omega_x$.
    Denote its support projection by $\Pi_\perp$, and assume, without loss of generality, that
    \begin{align}
        \sum_{x\in X}\Pi_x+\Pi_\perp=&I_M.
        \label{eq:memory_support_decomposition}
    \end{align}
    The following properties are satisfied.
    \begin{enumerate}
        \item Free introduction.
        For every $y\in X\cup\{\perp\}$, the preparation channel
        \begin{align}
            \mathcal{P}_y(L):=&\mathrm{tr}(L)\omega_y
        \end{align}
        is a free operation.
        \item Standardization.
        There exists a free operation $\mathcal{D}:M\to M$ satisfying
        \begin{align}
            \mathcal{D}(L)=&
            \sum_{x\in X\cup\{\perp\}} \mathrm{tr}(\Pi_x L)\omega_x
            \label{eq:memory_standardization}
        \end{align}
        for every operator $L$ on $M$.
        \item Conditional control.
        For any family of free operations
        $\{\mathcal{E}_y: B\to C\}_{y\in X\cup\{\perp\}}$,
        there exists a free operation $\mathcal{C}:MB\to MC$ satisfying
        \begin{align}
            \mathcal{C}(\omega_y\otimes L)=&
            \omega_y\otimes\mathcal{E}_y(L)
        \end{align}
        for every $y\in X\cup\{\perp\}$ and every operator $L$ on $B$.
    \end{enumerate}
\end{assumption}

We emphasize that the above assumptions are based only on the original resource theory and its equipped representation, and do not depend on the free states and free operations in the symmetry-constrained version of the resource theory defined below.

By using the notion of memory states, we define memory-assisted POVMs.

We define a free sub-POVM by using free operations and memory states, where a sub-POVM $\{M_x\}_{x\in X}$ is a set of positive-valued operators satisfying $\sum_{x\in X} M_x\leq I$.

\begin{definition} [Memory-assisted sub-POVM]
    \label{def:memory_assisted_subPOVM}
    Let $\{M_x\}$ be a sub-POVM on system $\mathrm{A}$. 
    We say that $\{M_x\}_{x\in X}$ is a memory-assisted POVM if there exists some $\mathcal{E}\in\mathfrak{O}^{\mathrm{A}\to\mathrm{M}}$ such that 
    \begin{align}
        M_x=\mathcal{E}^\dag(\Pi_x) \forall x\in X, \label{eq:memory_assisted_subPOVM}
    \end{align}
    where $\Pi_x$ is a projection onto the support of $\omega_x$. 
\end{definition}

By Eq.~\eqref{eq:memory_support_decomposition} and the trace preserving property of $\mathcal{E}$, the inconclusive effect satisfies
\begin{align}
    M_\perp:=I_A-\sum_{x\in X}M_x=&\mathcal{E}^\dag(\Pi_\perp).
    \label{eq:memory_assisted_failure_effect}
\end{align}
Thus, a memory-assisted sub-POVM is equivalently a complete memory-assisted POVM on the extended outcome set $X\cup\{\perp\}$.

Under the assumption of the existence of the free memory, we can show that control-admissible sub-POVM is equivalent to memory-assited sub-POVM.

\begin{lemma}[Equivalence of control-admissible and memory-assisted sub-POVMs]
    \label{lem:memory_assisted_control_admissible}
    Suppose that the free-memory assumption holds.
    Then, a sub-POVM $\{M_x\}_{x\in X}$ is control-admissible if and only if it is memory-assisted.
\end{lemma}

\begin{proof}
    We first suppose that $\{M_x\}_{x\in X}$ is memory-assisted.
    Let $\mathcal{F}:A\to M$ be a free operation satisfying
    \begin{align}
        M_x=&\mathcal{F}^\dag(\Pi_x)
    \end{align}
    for every $x\in X$, and set
    \begin{align}
        M_\perp:=&I_A-\sum_{x\in X}M_x.
    \end{align}
    Equation~\eqref{eq:memory_support_decomposition} gives
    \begin{align}
        M_\perp=&\mathcal{F}^\dag(\Pi_\perp).
    \end{align}

    Take arbitrary free operations
    $\{\mathcal{E}_y:B\to C\}_{y\in X\cup\{\perp\}}$.
    By the standardization property, for every operator $L$ on $AB$,
    \begin{align}
        \left(
        \left(\mathcal{D}\circ\mathcal{F}\right)
        \otimes\mathrm{id}_B
        \right)(L)
        =&
        \sum_{y\in X\cup\{\perp\}}
        \omega_y\otimes
        \mathrm{tr}_A\left(
        \left(M_y\otimes I_B\right)L
        \right).
        \label{eq:standardized_memory_output}
    \end{align}
    Applying the conditional-control operation and tracing out the memory system gives the free operation
    \begin{align}
        \mathcal{G}(L):=&
        \sum_{y\in X\cup\{\perp\}}
        \mathcal{E}_y\left(
        \mathrm{tr}_A\left(
        \left(M_y\otimes I_B\right)L
        \right)
        \right).
        \label{eq:memory_induced_control}
    \end{align}
    Hence, $\{M_x\}_{x\in X}$ is control-admissible.

    Conversely, suppose that $\{M_x\}_{x\in X}$ is control-admissible, and define
    \begin{align}
        M_\perp:=&I_A-\sum_{x\in X}M_x.
    \end{align}
    By the free-introduction property, the preparation channels
    \begin{align}
        \mathcal{P}_y(L):=&\mathrm{tr}(L)\omega_y
    \end{align}
    are free for every $y\in X\cup\{\perp\}$.
    Control admissibility therefore implies that the channel
    \begin{align}
        \mathcal{F}(L):=&
        \sum_{y\in X\cup\{\perp\}}
        \mathrm{tr}(M_yL)\omega_y
        \label{eq:memory_encoding_from_control}
    \end{align}
    is free.
    Since the supports of the states $\{\omega_y\}_{y\in X\cup\{\perp\}}$ are mutually orthogonal, we have
    \begin{align}
        \mathrm{tr}(\Pi_x\omega_y)=&\delta_{x,y}
    \end{align}
    for every $x\in X$ and $y\in X\cup\{\perp\}$.
    Therefore, for every operator $L$ on $A$,
    \begin{align}
        \mathrm{tr}\left(
        \mathcal{F}^\dag(\Pi_x)L
        \right)
        =\mathrm{tr}\left(
        \Pi_x\mathcal{F}(L)
        \right) 
        =\mathrm{tr}(M_x L),
    \end{align}
    and hence
    \begin{align}
        \mathcal{F}^\dag(\Pi_x)=&M_x.
    \end{align}
    Thus, $\{M_x\}_{x\in X}$ is memory-assisted.
\end{proof}

We can now relate the different operational manifestations of accessible symmetry information, which is a generalization of Theorem~\ref{thm:conversion_state_discrimination}.
The following theorem compares state discrimination, distillation of a finite symmetry reference, and the simulation of symmetry-compatible free operations using the input state.
Unlike the rate theorem above, this result explicitly involves a physical memory and therefore relies on the corresponding memory assumptions.

\begin{theorem} \label{thm:conversion_POVM_equivalence}
    Let $G$ be a finite group, $K$ be a subgroup of $G$, $\epsilon\in [0, 1)$, $\Delta$ be a distinguishability measure, its orthogonal distinguishability function $f_\Delta$ be strictly increasing, and $S$ be the symmetry subgroup of $\rho$. 
    Then, the following statements are equivalent: 
    \begin{align}
        (i)&\ \forall\mathcal{E}\in\mathfrak{O}_{S\textrm{-}\mathrm{cov}}(\mathrm{X}\to\mathrm{Y}),\ \exists\mathcal{F}\in\mathfrak{O}_{G\textrm{-}\mathrm{cov}}(\mathrm{AX}\to\mathrm{Y})\textrm{ s.t. }\forall\chi\in\mathcal{S}(X),\ \Delta(\mathcal{E}(\chi), \mathcal{F}(\rho\otimes \chi))\leq f_\Delta(-\log(1-\epsilon)). \\
        (ii)&\ \forall \sigma\in\mathfrak{S}_{S\text{-}\mathrm{inv}}(\mathrm{B}),\ \exists\mathcal{F}\in\mathfrak{O}_{G\textrm{-}\mathrm{cov}}(\mathrm{A}\to\mathrm{B})\textrm{ s.t. }\Delta(\sigma, \mathcal{F}(\rho))\leq f_\Delta(-\log(1-\epsilon)).\\
        (iii)&\ \exists \mathcal{F}\in\mathfrak{O}_{G\textrm{-}\mathrm{cov}}(\mathrm{A}\to\mathrm{M})\textrm{ s.t. }\Delta(\omega_S, \mathcal{F}(\rho))\leq f_\Delta(-\log(1-\epsilon)). \\
        (iv)&\ \exists \textrm{memory-assisted sub-POVM }\{M_x\}_{x\in G/S} \textrm{ s.t. } \min_{g\in G} \mathrm{tr}(M_{gS}\mathcal{U}_g(\rho))\geq 1-\epsilon. \\
        (v)&\ \exists \textrm{memory-assisted sub-POVM }\{M_x\}_{x\in G/S} \textrm{ s.t. } \frac{1}{|G|}\sum_{g\in G}\mathrm{tr}(M_{gS}\mathcal{U}_g(\rho))\geq 1-\epsilon. 
    \end{align}
\end{theorem}

\begin{proof}
    We note that (i)$\rightarrow$(ii), (ii)$\rightarrow$(iii), and (iv)$\rightarrow$(v) are trivial, and (v)$\rightarrow$(i) is direct from Theorem~\ref{thm:conversion_control_admissible_POVM} and Lemma~\ref{lem:memory_assisted_control_admissible}. 
    For the proof of (i)$\rightarrow$(ii), as an $S$-covariant free operation $\mathcal{E}$, we consider a map that introduces an $S$-invariant state $\sigma$. 
    For the proof of (ii)$\rightarrow$(iii), we take $\omega_S$ as an example of an $S$-invariant free state. 
    By taking the average of (iv), we get (v). 
    By Lemma~\ref{lem:memory_assisted_control_admissible}, a memory-assisted sub-POVM is control-admissible. 
    Thus, (v) implies 
    \begin{align}
        \exists \textrm{control-admissible sub-POVM }\{M_x\}_{x\in G/S} \textrm{ s.t. } \frac{1}{|G|}\sum_{g\in G}\mathrm{tr}(M_{gS}\mathcal{U}_g(\rho))\geq 1-\epsilon, 
    \end{align}
    which implies (i) by Theorem~\ref{thm:conversion_control_admissible_POVM}. 
    Therefore, it is sufficient to prove (iii)$\rightarrow$(iv) in the following. 
    We take some $\mathcal{F}\in\mathfrak{O}^{G\textrm{-}\mathrm{cov}}(\mathrm{A}\to\mathrm{M})$ such that $\Delta(\omega_S, \mathcal{F}(\rho))\leq f_\Delta(-\log(1-\epsilon))$. 
    By Lemma~\ref{lem:distingishability_evaluation_Dmax_Dmin}, this implies 
    \begin{align}
        D^\mathrm{min}(\omega_S\|\mathcal{F}(\rho))\leq -\log(1-\epsilon), 
    \end{align}
    which means 
    \begin{align}
        \mathrm{tr}(\Pi_{\omega_S}\mathcal{F}(\rho))\geq 1-\epsilon. \label{eq:thm:conversion_POVM_equivalence1}
    \end{align}
    We define a memory-assisted sub-POVM $\{M_x\}_{x\in G/S}$ by 
    \begin{align}
        M_{gS}:=\mathcal{F}^\dag(\Pi_{gS}). 
    \end{align}
    Then, for any $g\in G$, we have 
    \begin{align}
        \mathrm{tr}(M_{gS}\mathcal{U}_g(\rho)) 
        =\mathrm{tr}(\mathcal{F}^\dag(\Pi_{gS})\mathcal{U}_g(\rho)) 
        =\mathrm{tr}(\mathcal{F}^\dag\circ\mathcal{U}_g(\Pi_S)\mathcal{U}_g(\rho)) 
        =\mathrm{tr}(\Pi_S\mathcal{U}_{g^{-1}}\circ\mathcal{F}\circ\mathcal{U}_g(\rho)) 
        =\mathrm{tr}(\Pi_S \mathcal{F}(\rho)). \label{eq:thm:conversion_POVM_equivalence2}
    \end{align}
    By Eqs.~\eqref{eq:thm:conversion_POVM_equivalence1} and \eqref{eq:thm:conversion_POVM_equivalence2}, we have 
    \begin{align}
        \mathrm{tr}(M_{gS}\mathcal{U}_g(\rho))\geq 1-\epsilon. 
    \end{align}
    Since the choice of $g\in G$ is arbitrary, we get (iv). 
\end{proof}

\section{Entanglement under finite-group symmetry}
\label{sec:entanglement}

We apply the general framework to multipartite entanglement.
We first impose covariance on the resulting LOCC channel and then examine implementations in which each elementary operation respects the symmetry.

\subsection{Map-level conversion rates}
\label{subsec:entanglement_map}

Consider input and output systems $A=A_1\cdots A_m$ and $B=B_1\cdots B_m$ shared by $m$ parties.
Let $G$ be a finite group whose projective unitary representations on these systems are tensor products of local representations:
\begin{align}
    &U_A(g)=\bigotimes_{i=1}^m U_{A_i}(g), \\
    &U_B(g)=\bigotimes_{i=1}^m U_{B_i}(g).
\end{align}
We write $\mathcal{U}_{A,g}:=U_A(g)\cdot U_A(g)^\dag$ and similarly for the output system.
The map-level operation class is
\begin{align}
    \mathcal{O}_{G\text{-cov},\mathrm{LOCC}}(A\to B):=
    \{\mathcal{E}\in\mathcal{O}_{\mathrm{LOCC}}(A\to B)
    \ |\ \mathcal{E}\circ\mathcal{U}_{A,g}
    =\mathcal{U}_{B,g}\circ\mathcal{E}\text{ for all }g\in G\}.
\end{align}
Only the complete channel is required to be covariant; its individual local operations need not be symmetric.

Local unitary conjugation preserves LOCC, so these representations are compatible with the LOCC resource theory.
Moreover, an LOCC measurement followed by an outcome-dependent LOCC channel is again LOCC.
Thus, when $\mathrm{Sym}_{G,U_A}(\rho)\subset\mathrm{Sym}_{G,U_B}(\sigma)$, the remaining ingredient needed for the general rate bound is an LOCC measurement that distinguishes the orbit states with vanishing error.
We establish this ingredient for arbitrary multipartite states.

The unrestricted discrimination protocol of Lemma~\ref{lem:asymmetry_state_discrimination} estimates the expectation value of an orbit state regarded as an observable.
Although this observable is generally nonlocal, its expectation value can be estimated by LOCC.
Indeed, every Hermitian operator on a finite-dimensional multipartite system admits a finite decomposition into tensor products of local Hermitian operators.
Randomly sampling the terms of such a decomposition therefore gives a bounded unbiased estimator that can be implemented by local measurements and classical communication.
The following lemma makes this observation precise.

\begin{lemma} \label{lem:binary_test_LOCC}
    Let $X$ be a finite set, $\{\rho_x\}_{x\in X}$ be a finite set of multipartite states indexed by $x\in X$, and Hermitian operator $O$ satisfy 
    \begin{align}
        \mathrm{tr}(\rho_{x_0} O)>\mathrm{tr}(\rho_x O)\ \forall x\in X-\{x_0\}. \label{eq:lem:binary_test_LOCC01}
    \end{align}
    Then, there exists some LOCC POVM $\{T, I-T\}$ that acts on $N$-copy system such that 
    \begin{align}
        &\mathrm{tr}(T \rho_{x_0}^{\otimes N})\geq 1-e^{-\alpha N}, \label{eq:lem:binary_test_LOCC02}\\
        &\mathrm{tr}(T \rho_x^{\otimes N})\leq e^{-\alpha N}\ \forall x\in X-\{x_0\} \label{eq:lem:binary_test_LOCC03}
    \end{align}
    with some $\alpha>0$. 
\end{lemma}

\begin{proof}
    First, we decompose $O$ into tensor product operators: 
    \begin{align}
        O:=\sum_{j=1}^J c_j L_{1, j}\otimes L_{2, j}\otimes \cdots \otimes L_{m, j} 
    \end{align}
    with $c_j\in\mathbb{R}$ and $L_{k, j}\in\mathcal{L}^\mathrm{H}(\mathcal{H}_k)$ satisfying $\|L_{k, j}\|_\infty=1$. 
    Let $L_{k, j}$ have the following spectral decomposition 
    \begin{align}
        L_{k, j}=\sum_{a\in A_{k, j}} aP_{k, j, a}, 
    \end{align}
    where $P_{k, j, a}$ is the projection operator of  $L_{k, j}$ with eigenvalue $a$. 
    We define 
    \begin{align}
        &C:=\sum_{j=1}^J |c_j|, \\
        &p_j:=\frac{|c_j|}{C}. 
    \end{align}
    We consider a protocol in which we take $j$ with probability $p_j$ and measure $L_{k, j}$ on the $k$th party. 
    The measurement result $a_1, a_2, ..., a_m$ is sent to Alice and she calculates a score 
    \begin{align}
        z(\omega):=C \mathrm{sgn}(c_j) a_1 a_2 \cdots a_m 
    \end{align}
    For a single-copy outcome $\omega=(j,a_1,\ldots,a_m)$, we define the corresponding POVM element by
    \begin{align}
        E_\omega:=p_j P_{1,j,a_1}\otimes\cdots\otimes P_{m,j,a_m}. 
    \end{align}
    We repeat this measurement on $N$ copies of states $\{\rho_x\}_{x\in X}$. 
    We denote the $N$th set of results by $\bm{\omega}:=(\omega_1, ..., \omega_n)$, and calculate the empirical average 
    \begin{align}
        \hat{\mu}_N(\bm{\omega}):=\frac{1}{N}\sum_{k=1}^N z(\omega_k). 
    \end{align}
    For a sequence of outcomes $\boldsymbol{\omega}=(\omega_1,\ldots,\omega_N)$, we define
    \begin{align}
        E_{\boldsymbol{\omega}}:=E_{\omega_1}\otimes\cdots\otimes E_{\omega_N}.
    \end{align}
    We set the threshold $\tau$ as 
    \begin{align}
        \tau:=\frac{1}{2}\left(\mathrm{tr}(\rho_{x_0} O)+\max_{x\in X-\{x_0\}}\mathrm{tr}(\rho_x O)\right)
    \end{align}
    and we define an acceptance operator $T$ by 
    \begin{align}
        T_N:=\sum_{\bm{\omega}: \hat{\mu_N}(\bm{\omega})\geq \tau} E_{\bm{\omega}}. 
    \end{align}

    In the following, we show that the acceptance operator $T$ satisfies Eqs.~\eqref{eq:lem:binary_test_LOCC02} and \eqref{eq:lem:binary_test_LOCC03}. 
    For the proof of this inequality, we mainly use Hoeffding's inequality. 
    For the preparation, we confirm that $\hat{\mu}_N(\omega)$ is an unbiased estimator of $\rho$. 
    We note that 
    \begin{align}
        \mathbb{E}[z] 
        =&\sum_{j, a_1, ..., a_m} z(j, a_1, ..., a_m)\mathrm{tr}(\rho E_{j, a_1, ..., a_m}) \\
        =&\sum_{j, a_1, ..., a_m} C \mathrm{sgn}(c_j)a_1 a_2 \cdots a_m \cdot \mathrm{tr}(\rho p_j P_{1, j, a_1}\otimes \cdots P_{m, j, a_m}) \\
        =&\mathrm{tr}\left(\rho\left(\sum_j c_j L_{1, j}\otimes \cdots\otimes L_{m, j}\right)\right) \\
        =&\mathrm{tr}(\rho O). 
    \end{align}
    This implies that the empirical average $\hat{\mu}_N(\bm{\omega})$ is also an unbiased estimator of $\mathrm{tr}(\rho O)$
    We can also confirm that $\hat{\mu}_N(\bm{\omega})$ is bounded by noting that $\|L_{k, j}\|_\infty\leq 1$ means the eigenvalues $a_{j, k, l}$ satisfies $|a_{j, k, l}|\leq 1$, which implies $|z(\omega_k)|\leq C$. 
    The empirical average $\hat{\mu}_N(\bm{\omega})$ also satisfies $|\hat{\mu}_N(\bm{\omega})|\leq C$. 
    Thus, we are ready to use Hoeffding's inequality to $\hat{\mu}_N$, and we get 
    \begin{align}
        \mathrm{tr}\left((I-T_N) \rho_{x_0}^{\otimes N}\right) 
        =\mathrm{Pr}(\hat{\mu}_N<\tau | \rho_{x_0}) 
        \leq\exp\left(-\frac{N\Delta^2}{8C^2}\right), 
    \end{align}
    where 
    \begin{align}
        \Delta:=\mathrm{tr}(\rho_{x_0} O)-\max_{x\in X-\{x_0\}}\mathrm{tr}(\rho_x O). 
    \end{align}
    For any $x\in X-\{x_0\}$, we also get 
    \begin{align}
        \mathrm{tr}\left(T_N \rho_x^{\otimes N}\right) 
        =\mathrm{Pr}(\hat{\mu}_N\geq\tau | \rho_x) 
        \leq\exp\left(-\frac{N\Delta^2}{8C^2}\right). 
    \end{align}
    Therefore, we set $\alpha:=\Delta^2/8C^2$ and we get Eqs.~Eqs.~\eqref{eq:lem:binary_test_LOCC02} and \eqref{eq:lem:binary_test_LOCC03}. 
\end{proof}

Next, by using this test operator $T$, we construct a LOCC protocol that distinguishes the elements $\{\mathcal{U}_g(\rho)\}_{g\in G}$. 
In this construction, we take the operator $\rho$ as the observable that works like $O$.

\begin{lemma} (Multiple-state discrimination from one-versus-rest LOCC tests) \label{lem:multiple_test_LOCC}
    Let $\epsilon\in[0, 1]$, $\{\rho_x\}_{x\in X}$ be a finite set of multipartite states, and for any $x\in X$, there exist some LOCC test $T_x$ acting on $N$ copies such that 
    \begin{align}
        &\mathrm{tr}\left(T_x\rho_x^{\otimes N}\right)\geq 1-\epsilon, \\
        &\mathrm{tr}\left(T_x\rho_y^{\otimes N}\right)\leq\epsilon\ \forall y\in X-\{x\}. 
    \end{align} 
    Then, using $N|X|$ copies, there exists an LOCC measurement $\left\{M_x\right\}_{x\in X}$ that discriminates the states with 
    \begin{align}
        \mathrm{tr}\left(M_x\rho_x^{\otimes N|X|} \right) \geq 1-|X|\epsilon 
    \end{align}
    for all $x\in X$. 
\end{lemma}

\begin{proof}
    We repeat the test measurement considered in Lemma~\ref{lem:binary_test_LOCC} $|X|$ times to the copies of $\rho$. 
    On the $x$ block, we perform the binary LOCC measurement $\{T_x, I-T_x\}$
    Let $b_x\in\{0, 1\}$ denote its outcome, where $b_x=1$ corresponds to accepting the hypothesis that the input state is $\rho_x$. 
    The combined outcome is therefore a bit string $\bm{b}=(b_x)_{x\in X}\in\{0, 1\}^{|X|}$. 
    More explicitly, the POVM element associated with a bit string $\boldsymbol b=(b_x)_{x\in\mathcal X}$ is 
    \begin{align}
        F_{\boldsymbol b}:=\bigotimes_{x\in X}R_{x,b_x}.
    \end{align}
    with 
    \begin{align}
        &R_{x,1}:=T_x, \\
        &R_{x,0}:=I-T_x, 
    \end{align}
    We note that the measurement $\{F_{\boldsymbol b}\}_{\boldsymbol b}$ can be implemented by LOCC because each binary measurement is applied by LOCC to a disjoint set of copies.

    For each $x\in X$, let $\bm{e}_x$ denote the bit string whose $x$-th entry is one and whose other entries are zero. 
    We choose any classical decoding function $d: \{0,1\}^X \to X$ satisfying 
    \begin{align}
        d(\bm{e}_x)=x\ \forall x\in X, 
    \end{align}
    and assign all other bit strings arbitrarily. 
    The resulting measurement operators are
    \begin{align}
        M_x^{(N)}:=\sum_{\bm{b}:\, d(\boldsymbol b)=x} F_{\boldsymbol b}.
    \end{align}
    Since these operators are obtained from an LOCC measurement by classical post-processing, $\{M_x\}_{x\in X}$ is also an LOCC measurement.

    Suppose that the input state is $\rho_x^{\otimes N|X|}$. 
    We note that $M_x\geq F_{\bm{e}_x}$. 
    Therefore, we get 
    \begin{align}
        \mathrm{tr}\left(M_x\rho_x^{\otimes N|X|}\right) 
        \geq &\mathrm{tr}\left(F_{\bm{e}_x}\rho_x^{\otimes N|X|}\right) \nonumber\\
        =&\mathrm{tr}\left(\left(\bigotimes_{y\in X}R_{y, \delta_{y, x}}\right)\rho_x^{\otimes N|X|}\right) \nonumber\\
        =&\prod_{y\in X} \mathrm{tr}\left(R_{y, \delta_{y, x}}\rho_x^{\otimes N}\right) \nonumber\\
        =&\mathrm{tr}(T_x\rho_x^{\otimes N})\cdot\prod_{y\in X-\{x\}}\mathrm{tr}\left((I-T_y)\rho_x^{\otimes N}\right) \nonumber\\
        \geq &(1-\epsilon)^{|X|} \nonumber\\
        \geq &1-|X|\epsilon. 
    \end{align}
    for all $x\in X$. 
\end{proof}

The preceding lemmas imply that the distinct states in any finite unitary orbit can be discriminated by LOCC with an exponentially decreasing error probability.

\begin{lemma}[LOCC discrimination of a finite unitary orbit] \label{lem:LOCC_discrimination}
    Let $G$ be a finite group, $U$ be a projective unitary representation of $G$ given by tensor products of local unitaries, $\rho$ be a multipartite state, and $S$ be defined by 
    \begin{align}
        S:=\mathrm{Sym}_{G,U}(\rho)=\{g\in G\ |\ U(g)\rho U(g)^\dag=\rho\}. 
    \end{align}
    Then, there exists a constant $\alpha>0$ such that, for every positive integer $k$, there exists an LOCC measurement $\{M_{gS}^{(k)}\}_{gS\in G/S}$ on $k|G/S|$ copies satisfying
    \begin{align}
        \mathrm{tr}\left(M_{gS}^{(k)}\left(U(g)\rho U(g)^\dag\right)^{\otimes k|G/S|}\right)
        \geq 1-|G/S|e^{-\alpha k} 
    \end{align}
    for all $g\in G$. 
\end{lemma}

\begin{proof}
    We fix a representative $g$ of each coset $gS\in G/S$.
    For every $gS\in G/S$, we define the orbit state 
    \begin{align}
        \rho_{gS}:=U(g)\rho U(g)^\dag. 
    \end{align}
    This definition is independent of the choice of the representative $g$, and $\rho_{gS}=\rho_{hS}$ holds if and only if $gS=hS$.

    We first construct a one-versus-rest test for the state corresponding to the identity coset $S$.
    We choose
    \begin{align}
        O_S:=\rho.
    \end{align}
    We note that 
    \begin{align}
        \mathrm{tr}\left(\rho O_S\right)-\mathrm{tr}\left(\rho_{hS}O_S\right)
        =&\mathrm{tr}\left(\rho^2\right)-\mathrm{tr}\left(\rho_{hS}\rho\right) \nonumber\\
        =&\frac{1}{2}\mathrm{tr}\left(\left(\rho-\rho_{hS}\right)^2\right) \nonumber\\
        =&\frac{1}{2}\left\|\rho-\rho_{hS}\right\|_2^2 \nonumber\\
        \geq &0.
    \end{align}
    Equality holds if and only if $hS=S$, which means $h\in S$. 
    Consequently, 
    \begin{align}
        \mathrm{tr}\left(\rho O_S\right)>\mathrm{tr}\left(\rho_{hS}O_S\right)
    \end{align}
    for every $hS\neq S$.

    Lemma~\ref{lem:binary_test_LOCC} therefore gives a constant $\alpha>0$ and an LOCC test $T_S^{(k)}$ on $k$ copies such that
    \begin{align}
        \mathrm{tr}\left(T_S^{(k)}\rho^{\otimes k}\right)
        \geq &1-e^{-\alpha k}, \\
        \mathrm{tr}\left(T_S^{(k)}\rho_{hS}^{\otimes k}\right)
        \leq &e^{-\alpha k}\ \forall h\notin S. 
    \end{align}

    For each $gS\in G/S$, we define the rotated test 
    \begin{align}
        T_{gS}^{(k)}:=U(g)^{\otimes k}T_S^{(k)}U(g)^{\dag\otimes k}. 
    \end{align}
    We note that $T_{gS}^{(k)}$ is an LOCC test because $U(g)$ is a tensor product of local unitaries.

    For the target state $\rho_{gS}$, unitary invariance of the trace gives
    \begin{align}
        \mathrm{tr}\left(T_{gS}^{(k)}\rho_{gS}^{\otimes k}\right) 
        =\mathrm{tr}\left(T_S^{(k)}\rho^{\otimes k}\right) 
        \geq 1-e^{-\alpha k}. 
    \end{align}
    For every $hS\neq gS$, we similarly obtain
    \begin{align}
        \mathrm{tr}\left(T_{gS}^{(k)}\rho_{hS}^{\otimes k}\right) 
        =\mathrm{tr}\left(T_S^{(k)}\rho_{g^{-1}hS}^{\otimes k}\right) 
        \leq e^{-\alpha k}, 
    \end{align}
    where $g^{-1}hS\neq S$ follows from $hS\neq gS$.

    Thus, the family $\{T_{gS}^{(k)}\}_{gS\in G/S}$ satisfies the assumptions of Lemma~\ref{lem:multiple_test_LOCC} with the same error exponent $\alpha$ for every $gS\in G/S$.
    Applying Lemma~\ref{lem:multiple_test_LOCC} gives an LOCC measurement $\{M_{gS}^{(k)}\}_{gS\in G/S}$ on $k|G/S|$ copies satisfying
    \begin{align}
        \mathrm{tr}\left(M_{gS}^{(k)}\rho_{gS}^{\otimes k|G/S|}\right)\geq 1-|G/S|e^{-\alpha k} 
    \end{align}
    for all $gS\in G/S$. 
\end{proof}

Applying the preceding construction to the finite orbit of $\rho$ gives an LOCC measurement that identifies the orbit with asymptotically vanishing error.
Since LOCC is closed under measurement-conditioned LOCC operations, this measurement is control-admissible.
Theorem~\ref{thm:rate_lower_bound_general} therefore gives the main map-level result of this section.

\begin{theorem}[Asymptotic conversion rates under finite-group-covariant LOCC] \label{thm:LOCC_rate_map_level}
    Let $G$ be a finite group, and $U$ and $U'$ be projective unitary representations of $G$ on the input and output multipartite Hilbert spaces, respectively. 
    Suppose that $U(g)$ and $U'(g)$ are tensor products of local unitaries for every $g\in G$. 
    Then, for arbitrary multipartite states $\rho$ and $\sigma$, 
    \begin{align}
        R_{G\text{-}\mathrm{cov}, \mathrm{LOCC}}(\rho\to\sigma)
        =
        \begin{cases}
            R_{\mathrm{LOCC}}(\rho\to\sigma) & \text{ if }\mathrm{Sym}_{G,U}(\rho)\subset\mathrm{Sym}_{G,U'}(\sigma),\\
            0 & \text{ otherwise}.
        \end{cases}
    \end{align}
\end{theorem}

This theorem implies that if the input and output states are pure bipartite states $\Phi$ and $\Psi$, $R_{G^\text{-}\mathrm{cov}, \mathrm{LOCC}}(\Phi\to\Psi)$ is explicitly given by 
\begin{align}
    R_{G^\text{-}\mathrm{cov}, \mathrm{LOCC}}(\Phi\to\Psi)=
    \begin{cases}
        \displaystyle\frac{S(\Phi)}{S(\Psi)} & \textrm{if } \mathrm{Sym}(\Phi)\subset\mathrm{Sym}(\Psi) \\
        0 & \mathrm{otherwise}, 
    \end{cases}
\end{align}
with the entanglement entropy $S(\cdot)$.

We note that the same argument for this proof extends directly from LOCC to separable operations.
Let $\mathrm{SEP}$ denote the class of multipartite quantum channels admitting a Kraus representation with product Kraus operators.
Since $\mathrm{LOCC}\subset\mathrm{SEP}$, the orbit-discrimination measurement constructed above is also available within $\mathrm{SEP}$.
Moreover, a separable measurement followed by an outcome-dependent separable operation is again separable, so the measurement is control-admissible for $\mathrm{SEP}$.
The remaining assumptions of the general conversion theorem are also preserved under separable operations.
Consequently, the same rate formula holds with $\mathrm{LOCC}$ replaced by $\mathrm{SEP}$.

\begin{proof}
    First, we consider the case where $\mathrm{Sym}_{G,U}(\rho)\not\subset\mathrm{Sym}_{G,U'}(\sigma)$. 
    Since the class of $G$-covariant LOCC is smaller than that of $G$-covariant operations, we have 
    \begin{align}
        R_{G\text{-cov}, \mathrm{LOCC}}(\rho\to\sigma)\leq
        R_{G\text{-cov}}(\rho\to\sigma). 
    \end{align}
    Theorem~\ref{thm:asymmetry_rate} implies that 
    \begin{align}
        R_{G\text{-cov}}(\rho\to\sigma)=0. 
    \end{align}
    Thus, we get 
    \begin{align}
        R_{{G\text{-cov}}, \mathrm{LOCC}}(\rho\to\sigma)=0.
    \end{align}

    Next, we consider the case where $\mathrm{Sym}_G(\rho)\subset\mathrm{Sym}_G(\sigma)$. 
    Since the class of $G$-covariant LOCC is smaller than that of LOCC, we trivially have 
    \begin{align}
        R_{{G\text{-cov}}, \mathrm{LOCC}}(\rho\to\sigma)\leq
        R_{\mathrm{LOCC}}(\rho\to\sigma). 
    \end{align}
    We take arbitrary $r\in (0, R_{\mathrm{LOCC}}(\rho\to\sigma))$. 
    Then, there exists some sequence of LOCC such that 
    \begin{align}
        \lim_{n\to\infty} T(\mathcal{E}^{(n)}(\rho^{\otimes n}), \sigma^{\otimes \lfloor rn \rfloor})=0. 
    \end{align}
    Since $\rho$ and $\sigma$ are invariant by the action of $s\in S$, we have 
    \begin{align}
        \lim_{n\to\infty} \frac{1}{|S|}\sum_{s\in S} T\left(\mathcal{E}^{(n)}(\mathcal{U}_s(\rho)^{\otimes n}), \mathcal{U}'_s(\sigma)^{\otimes \lfloor rn \rfloor}\right)=0. 
    \end{align}
    Lemma~\ref{lem:LOCC_discrimination} provides a control-admissible LOCC measurement that satisfies 
    \begin{align}
        \lim_{k\to\infty} \frac{1}{|G|}\sum_{g\in G} \mathrm{tr}\left(M_{gS}^{(k)}\mathcal{U}_g(\rho)^{\otimes k}\right)=1. 
    \end{align}
    Therefore, the general conversion theorem (Theorem~\ref{thm:rate_lower_bound_general}) gives
    \begin{align}
        R_{{G\text{-cov}}, \mathrm{LOCC}}(\rho\to\sigma)\geq
        r.
    \end{align}
    Since this holds for all $r\in(0, R_{\mathrm{LOCC}}(\rho\to\sigma))$, we get 
    \begin{align}
        R_{{G\text{-cov}}, \mathrm{LOCC}}(\rho\to\sigma) 
        =R_{\mathrm{LOCC}}(\rho\to\sigma).
    \end{align}
\end{proof}

The map-level theorem covers both global and local covariance.
For global covariance, the same group element $g\in G$ acts on every party through the tensor-product representations $U_A(g)$ and $U_B(g)$.
For independent local covariance, we instead apply the theorem to the product group $G_1\times\cdots\times G_m$, with representations
\begin{align}
    &U_A(g_1,\ldots,g_m)=\bigotimes_{i=1}^m U_{A_i}(g_i), \\
    &U_B(g_1,\ldots,g_m)=\bigotimes_{i=1}^m U_{B_i}(g_i).
\end{align}
When all local groups equal $G$, global covariance is covariance under the diagonal subgroup of $G^m$, whereas local covariance is covariance under the full product group.
The symmetry-subgroup condition in the rate theorem must therefore be evaluated for the group whose covariance is imposed.

These are different map-level restrictions.
The relation between their primitive-level realizations additionally depends on which shared ancillary states and classical communication channels are available.
We make these assumptions explicit below.

\subsection{Primitive-level implementations of covariant LOCC}
\label{sec:primitive-level}

In the preceding sections, the allowed operations were characterized at the level of quantum channels.
In this section, we show that the same class of transformations can be generated from symmetric local primitives.
This provides an operational interpretation of covariant LOCC transformations in terms of symmetric local unitaries, symmetric local projective measurements, symmetric state introduction, classical communication, and classical control.

Let $m$ be the number of parties, and let $G_i$ be a finite group acting on the system of the $i$th party.
We consider a subgroup
\begin{align}
    H\subset G_1\times\cdots\times G_m.
\end{align}
For $\boldsymbol{g}=(g_1,\ldots,g_m)\in H$, we write
\begin{align}
    U(\boldsymbol{g})
    :=\bigotimes_{i=1}^m U_i(g_i),
    \qquad
    U'(\boldsymbol{g})
    :=\bigotimes_{i=1}^m U_i'(g_i),
\end{align}
and define the corresponding actions by
\begin{align}
    \mathcal{U}_{\boldsymbol{g}}
    :=U(\boldsymbol{g})\,\cdot\,U(\boldsymbol{g})^\dag,
    \qquad
    \mathcal{U}_{\boldsymbol{g}}^{\prime}
    :=U'(\boldsymbol{g})\,\cdot\,U'(\boldsymbol{g})^\dag.
\end{align}
A quantum channel $\mathcal{E}$ is $H$-covariant if
\begin{align}
    \mathcal{E}\circ\mathcal{U}_{\boldsymbol{g}}
    =\mathcal{U}_{\boldsymbol{g}}^{\prime}\circ\mathcal{E}
\end{align}
for every $\boldsymbol{g}\in H$.

We denote by $\mathfrak{O}_{\mathrm{covLOCC},H}$ the class of finite-round LOCC channels that are $H$-covariant.
We denote by $\mathfrak{O}_{\mathrm{prim},H}$ the class of channels generated by $H$-symmetric local unitaries, $H$-symmetric local projective measurements, symmetric state introduction, symmetry-compatible classical communication and classical control, and discarding of local subsystems.
The classical registers containing the measurement outcomes and the classical transcript are taken to transform trivially under $H$.

Classical communication in this primitive model transmits only registers on which the symmetry acts trivially.
It does not include the transmission of quantum reference systems or classical registers carrying nontrivially transforming group labels.
Local symmetry-invariant ancillary states may be freely introduced.
Whenever a shared ancillary state is additionally admitted, its availability is stated explicitly, and the state is required to be separable across the parties so that its introduction does not supply entanglement.
In particular, the availability of the correlated reference state $\Gamma_H$ below is a separate assumption.

For each party, let $\{|\xi_{g_i}\rangle\}_{g_i\in G_i}$ denote the canonical basis carrying the left regular representation of $G_i$.
For $\boldsymbol{g}=(g_1,\ldots,g_m)$, define
\begin{align}
    |\boldsymbol{\xi}_{\boldsymbol{g}}\rangle
    :=\bigotimes_{i=1}^m|\xi_{g_i}\rangle,
\end{align}
and introduce the correlated reference state
\begin{align}
    \Gamma_H
    :=\frac{1}{|H|}
    \sum_{\boldsymbol{g}\in H}
    |\boldsymbol{\xi}_{\boldsymbol{g}}\rangle
    \langle\boldsymbol{\xi}_{\boldsymbol{g}}|.
\end{align}

\begin{proposition}[Equivalence of map-level and primitive-level operations]
\label{prop:primitive-map-equivalence}
    Suppose that the state $\Gamma_H$ can be freely introduced by the primitive-level theory.
    Then every finite-round $H$-covariant LOCC channel admits an exact implementation using $H$-symmetric local primitives.
    Consequently,
    \begin{align}
        \mathfrak{O}_{\mathrm{prim},H}
        =\mathfrak{O}_{\mathrm{covLOCC},H}.
    \end{align}
\end{proposition}

\begin{proof}
    Every channel generated by the primitive-level operations is LOCC because all of its quantum operations are local and the only communication between the parties is classical.
    Moreover, each primitive is $H$-covariant, and covariance is preserved under sequential composition, classical conditioning, and coarse-graining over measurement outcomes.
    It follows that
    \begin{align}
        \mathfrak{O}_{\mathrm{prim},H}
        \subset\mathfrak{O}_{\mathrm{covLOCC},H}.
    \end{align}

    We next prove the converse inclusion.
    Let $\mathcal{E}\in\mathfrak{O}_{\mathrm{covLOCC},H}$, and fix a finite-round LOCC realization of $\mathcal{E}$.
    This realization may contain local state introductions, local unitaries, local measurements, classical communication, and operations conditioned on the preceding classical transcript.

    We attach to the $i$th party a regular-representation register with basis $\{|\xi_{g_i}\rangle\}_{g_i\in G_i}$.
    Each local primitive in the chosen LOCC realization is then replaced by its symmetric extension constructed in Sec.~\ref{sec:LOCC_technical_lemmas}.
    The same group label $\boldsymbol{g}$ is retained throughout all rounds and all branches of the protocol.
    Since the classical transcript transforms trivially, operations conditioned on the transcript can be extended independently in each classical branch.
    The technical lemmas in Sec.~\ref{sec:LOCC_technical_lemmas} therefore imply that the entire LOCC realization can be replaced by a protocol composed only of $H$-symmetric local primitives.

    Let $\widetilde{\mathcal{E}}$ denote the resulting channel on the principal system and the regular-representation registers.
    For every input state $\rho$, its action on $\rho\otimes\Gamma_H$ is
    \begin{align}
        \widetilde{\mathcal{E}}\left(\rho\otimes\Gamma_H\right)
        =&\frac{1}{|H|}
        \sum_{\boldsymbol{g}\in H}
        \left(
        \mathcal{U}_{\boldsymbol{g}}^{\prime}
        \circ\mathcal{E}
        \circ\mathcal{U}_{\boldsymbol{g}^{-1}}
        \right)(\rho)
        \otimes
        |\boldsymbol{\xi}_{\boldsymbol{g}}\rangle
        \langle\boldsymbol{\xi}_{\boldsymbol{g}}|.
    \end{align}
    The $H$-covariance of $\mathcal{E}$ gives
    \begin{align}
        \mathcal{U}_{\boldsymbol{g}}^{\prime}
        \circ\mathcal{E}
        \circ\mathcal{U}_{\boldsymbol{g}^{-1}}
        =\mathcal{E}
    \end{align}
    for every $\boldsymbol{g}\in H$.
    Hence,
    \begin{align}
        \widetilde{\mathcal{E}}\left(\rho\otimes\Gamma_H\right)
        =\mathcal{E}(\rho)\otimes\Gamma_H.
    \end{align}
    Thus, the symmetric protocol implements $\mathcal{E}$ exactly on the principal system and returns the reference state $\Gamma_H$ unchanged.
    Since $\Gamma_H$ can be freely introduced, this proves that
    \begin{align}
        \mathfrak{O}_{\mathrm{covLOCC},H}
        \subset\mathfrak{O}_{\mathrm{prim},H}.
    \end{align}
    Combining the two inclusions proves the claim.
\end{proof}

\subsubsection{Local symmetry}
\label{subsec:primitive-local-symmetry}

We first specialize Proposition~\ref{prop:primitive-map-equivalence} to local symmetry.
In this case,
\begin{align}
    H_{\mathrm{loc}}
    :=G_1\times\cdots\times G_m.
\end{align}
The corresponding reference state factorizes as
\begin{align}
    \Gamma_{\mathrm{loc}}
    =\frac{1}{|G_1|\cdots|G_m|}
    \sum_{g_1\in G_1,\ldots,g_m\in G_m}
    \bigotimes_{i=1}^m
    |\xi_{g_i}\rangle\langle\xi_{g_i}| 
    =\bigotimes_{i=1}^m
    \left(
    \frac{1}{|G_i|}
    \sum_{g_i\in G_i}
    |\xi_{g_i}\rangle\langle\xi_{g_i}|
    \right).
\end{align}
Each factor is invariant under the corresponding local regular representation.
Therefore, $\Gamma_{\mathrm{loc}}$ is a symmetric product state and can be introduced by local symmetric state introduction.

\begin{corollary}[Local symmetry]
\label{cor:primitive-local-equivalence}
    Under local symmetry, the primitive-level operation class coincides with the class of locally covariant LOCC channels:
    \begin{align}
        \mathfrak{O}_{\mathrm{prim},H_{\mathrm{loc}}}
        =\mathfrak{O}_{\mathrm{covLOCC},H_{\mathrm{loc}}}.
    \end{align}
\end{corollary}

\begin{proof}
    The reference state $\Gamma_{\mathrm{loc}}$ is a symmetric product state and is therefore freely introducible.
    The result follows directly from Proposition~\ref{prop:primitive-map-equivalence}.
\end{proof}

\subsubsection{Global symmetry}
\label{subsec:primitive-global-symmetry}

We next consider the case in which the same finite group $G$ acts globally on all parties.
The relevant subgroup is the diagonal subgroup
\begin{align}
    H_{\mathrm{glob}}
    :=\{(g,\ldots,g)\in G^m\ |\ g\in G\}.
\end{align}
The corresponding reference state is
\begin{align}
    \Gamma_{\mathrm{glob}}
    :=\frac{1}{|G|}
    \sum_{g\in G}
    \bigotimes_{i=1}^m
    |\xi_g\rangle\langle\xi_g|.
\end{align}
Unlike $\Gamma_{\mathrm{loc}}$, the state $\Gamma_{\mathrm{glob}}$ is generally correlated across the parties.
Nevertheless, it is invariant under the diagonal action of $G$.

The state $\Gamma_{\mathrm{glob}}$ is separable across the parties and invariant under the diagonal action.
However, its preparation is not implied by local symmetric state introduction and symmetry-neutral classical communication.
We therefore distinguish two primitive-level models according to whether this shared reference is freely available.

First, suppose that $\Gamma_{\mathrm{glob}}$ is admitted as a free shared ancillary state.
Its introduction does not supply entanglement, but it supplies a correlation between the local reference registers that is absent from a product of locally invariant reference states.
The primitive-level implementation theorem then gives the following result.

\begin{corollary}[Global symmetry with a shared invariant reference]
    \label{cor:primitive-global-equivalence}
    Suppose that the separable, globally invariant reference state $\Gamma_{\mathrm{glob}}$ can be freely introduced.
    Then the primitive-level operation class coincides with the class of finite-round globally covariant LOCC channels:
    \begin{align}
        \mathcal{O}_{\mathrm{prim},H_{\mathrm{glob}}}=&
        \mathcal{O}_{\mathrm{covLOCC},H_{\mathrm{glob}}}.
    \end{align}
\end{corollary}

\begin{proof}
    The assumed availability of $\Gamma_{\mathrm{glob}}$ allows us to apply Proposition~\ref{prop:primitive-map-equivalence} with $H=H_{\mathrm{glob}}$.
\end{proof}

\paragraph{Only local ancillary preparation.}
Consider instead the model in which ancillary quantum states can be introduced only locally and are invariant under the corresponding local $G$ action.
Shared classical randomness is allowed, but every communicated classical register transforms trivially under the symmetry, and no quantum system or nontrivially transforming classical register is transmitted between parties.
These restrictions concern freely supplied ancillary systems, not the resource input whose conversion is being studied.

Let $\mathcal{O}^{\mathrm{local}}_{\mathrm{prim},H_{\mathrm{glob}}}$ denote the resulting class.
Although symmetry is initially described using the diagonal group, the allowed local primitives are exactly those of the independently symmetric model with $H_{\mathrm{loc}}=G^m$.
For example, for a unitary $V_i$ acting at party $i$,
\begin{align}
    \left[V_i\otimes I_{\overline i},
    \bigotimes_{j=1}^m U_j(g)\right]=&0
    \quad\text{for every }g\in G
\end{align}
is equivalent to $[V_i,U_i(g)]=0$ for every $g\in G$.
The same equivalence holds for every local measurement projection.
Each measurement branch is therefore covariant under the local action at that party and acts trivially on the other parties.

Local invariant ancillary preparation, discarding, and symmetry-neutral classical communication also respect the independent group actions.
Consequently, composition and classical conditioning generate exactly the same operation class as in the locally symmetric primitive model:
\begin{align}
    \mathcal{O}^{\mathrm{local}}_{\mathrm{prim},H_{\mathrm{glob}}}=&
    \mathcal{O}_{\mathrm{prim},H_{\mathrm{loc}}},
    \qquad H_{\mathrm{loc}}=G^m.
    \label{eq:global_local_ancilla_equivalence}
\end{align}
In particular, for two parties this is the primitive-level theory with independent $G\times G$ symmetry.
Using Corollary~\ref{cor:primitive-local-equivalence}, the implemented channels are precisely the finite-round $G^m$-covariant LOCC channels under the stated ancillary-system assumptions.
The corresponding conversion rate is therefore governed by the symmetry subgroups for the product-group action, not merely by those for the diagonal action.

This equivalence does not identify the full map-level classes of globally and locally covariant LOCC channels.
It follows from the specified restrictions on the primitives and ancillary resources.
Moreover, excluding the direct introduction of shared ancillary states alone would not suffice if classical registers carrying nontrivially transforming group labels could be transmitted.
Such communication could distribute correlations between reference registers and thereby make $\Gamma_{\mathrm{glob}}$ available.

\subsubsection{Consequences for asymptotic conversion rates}
\label{subsec:primitive-rates}

The preceding equivalence shows that replacing covariant LOCC channels by their primitive-level implementations does not change the achievable transformations.
In particular, it does not change the asymptotic conversion rate.

Let $R_{\mathrm{prim},H}(\rho\to\sigma)$ denote the conversion rate under primitive-level $H$-symmetric LOCC, and let $R_{\mathrm{covLOCC},H}(\rho\to\sigma)$ denote the rate under $H$-covariant LOCC.
Proposition~\ref{prop:primitive-map-equivalence} gives
\begin{align}
    R_{\mathrm{prim},H}(\rho\to\sigma)
    =R_{\mathrm{covLOCC},H}(\rho\to\sigma).
\end{align}
Combining this equality with the map-level conversion theorem established in the preceding section yields the following result.

\begin{theorem}[Primitive-level conversion rate]
\label{thm:primitive-conversion-rate}
    Suppose that the reference state $\Gamma_H$ can be freely introduced.
    Then
    \begin{align}
        R_{\mathrm{prim},H}(\rho\to\sigma)
        =
        \begin{cases}
            R_{\mathrm{LOCC}}(\rho\to\sigma)
            &
            \text{if }\mathrm{Sym}_{H,U}(\rho)\subset\mathrm{Sym}_{H,U'}(\sigma),
            \\
            0
            &
            \text{otherwise},
        \end{cases}
    \end{align}
    where the zero-rate case is the one established in the resource-theoretic discussion at the beginning of the paper.
    This statement applies both to $H=H_{\mathrm{loc}}$ and to $H=H_{\mathrm{glob}}$.
\end{theorem}

\begin{proof}
    Corollaries~\ref{cor:primitive-local-equivalence} and~\ref{cor:primitive-global-equivalence} show that, in the local and global settings, respectively, every covariant LOCC protocol has an exact primitive-level implementation.
    Therefore, the primitive-level and map-level achievable rates coincide.
    The claimed expression then follows from the map-level conversion theorem.
\end{proof}

If LOCC is defined as the closure of finite-round LOCC, the same conclusions hold after taking the closure of both operation classes.

\section{Magic under finite-group symmetry}
\label{sec:magic}

We next consider the resource theory of magic under finite-group symmetry constraints.
At the map level, we allow the symmetry to be represented by arbitrary qubit Clifford operators and consider both stabilizer operations and completely stabilizer-preserving operations.
We show that imposing finite-group covariance does not reduce the asymptotic conversion rate whenever the symmetry subgroup of the input is contained in that of the output.

We then impose symmetry directly on the elementary operations of a stabilizer protocol.
In this formulation, every measured Hermitian Pauli operator must commute with the symmetry representation on the entire system on which it acts, including any ancillary registers.
For symmetries represented by Pauli operators, we construct symmetric implementations of covariant stabilizer protocols and obtain the corresponding equality of asymptotic conversion rates.
For general Clifford representations, we explain why this construction does not extend directly by exhibiting an obstruction to symmetric Pauli measurements.
We do not claim that the Pauli condition is necessary for equality of conversion rates, or that this measurement obstruction alone establishes a separation between map-level and primitive-level rates.

The construction of stabilizer memory states for general Clifford representations is treated separately in Appendix~\ref{app:clifford_memory}.
There, preparation, readout, and conditional control are considered within the underlying stabilizer theory, without requiring each elementary operation to respect the symmetry.

\subsection{Magic operations and finite-group covariance}
\label{subsec:magic-setting}

Let $\mathrm{STAB}(\mathrm{A})$ denote the set of stabilizer states on a finite-dimensional stabilizer system $\mathrm{A}$.
A stabilizer operation is a quantum channel that can be implemented using Clifford unitaries, stabilizer ancillas, Pauli measurements, classical feed-forward, and the discarding of subsystems.
We denote the set of stabilizer operations from $A$ to $B$ by $\mathrm{SO}(\mathrm{A}\to \mathrm{B})$.

We also consider operation classes larger than the set of stabilizer operations.
A quantum channel $\mathcal{E}:A\to B$ is stabilizer preserving if
\begin{align}
    \mathcal{E}(\tau)\in\mathrm{STAB}(B)
\end{align}
for every $\tau\in\mathrm{STAB}(A)$.
We denote the set of such channels by $\mathrm{SPO}(A\to B)$.

A quantum channel $\mathcal{E}:A\to B$ is completely stabilizer preserving if
\begin{align}
    \mathcal{E}\otimes\operatorname{id}_R\in\mathrm{SPO}(AR\to BR)
\end{align}
for every ancillary stabilizer system $R$.
We denote the set of such channels by $\mathrm{CSPO}(A\to B)$.
These operation classes satisfy
\begin{align}
    \mathrm{SO}(A\to B)\subseteq\mathrm{CSPO}(A\to B)\subseteq\mathrm{SPO}(A\to B).
\end{align}

Let $G$ be a finite group, and let $U_A(g)$ and $U_B(g)$ be possibly projective unitary representations of $G$ on $A$ and $B$, respectively.
Throughout this section, we assume that $U_A(g)$ and $U_B(g)$ are Clifford operators for every $g\in G$.
We denote the corresponding unitary channels by
\begin{align}
    \mathcal{U}_{A,g}\coloneqq U_A(g)\,\cdot\,U_A(g)^\dag,\qquad \mathcal{U}_{B,g}\coloneqq U_B(g)\,\cdot\,U_B(g)^\dag.
\end{align}

A quantum channel $\mathcal{E}:A\to B$ is $G$-covariant if
\begin{align}
    \mathcal{E}\circ\mathcal{U}_{A,g}=\mathcal{U}_{B,g}\circ\mathcal{E}
\end{align}
for every $g\in G$.
We denote the set of $G$-covariant channels from $A$ to $B$ by $\mathrm{Cov}_G(A\to B)$.

For a map-level class of free operations $\mathcal{F}(A\to B)$, we define its $G$-covariant restriction by
\begin{align}
    \mathcal{F}^G(A\to B)\coloneqq\mathcal{F}(A\to B)\cap\mathrm{Cov}_G(A\to B).
\end{align}
Our map-level conversion results below concern the classes $\mathrm{SO}^{G}$ and $\mathrm{CSPO}^{G}$ obtained from stabilizer operations and completely stabilizer-preserving operations, respectively.

This map-level restriction requires only the resulting channel to be $G$-covariant.
It does not require a realization in which every elementary Clifford unitary, ancillary-state introduction, and Pauli measurement respects the symmetry.
We denote the class defined by the latter requirement by $\mathrm{SO}^{G}_{\mathrm{prim}}(A\to B)$.
Its precise definition is given in Sec.~\ref{sec:magic_primitive}.
By construction,
\begin{align}
    \mathrm{SO}^{G}_{\mathrm{prim}}(A\to B)\subset& \mathrm{SO}^{G}(A\to B).
    \label{eq:magic_primitive_inclusion}
\end{align}
The converse inclusion requires a separate implementation argument and does not follow from covariance of the channel alone.

For states $\rho$ and $\sigma$, we denote the asymptotic conversion rate under an operation class $\mathcal{F}$ by $R_{\mathcal{F}}(\rho\to\sigma)$.
We first establish the map-level rate equalities for Clifford representations.
We then show that the corresponding stabilizer-operation rate is also attained at the primitive level when both the input and output symmetries are represented by Pauli operators.

\subsection{Map-level conversion rates}
\label{subsec:magic-map-level}

We first show that finite-group covariance does not change the asymptotic conversion rate at the map level, except for the necessary condition imposed by the symmetry subgroups.
The discrimination principle used here is the same as that used for entanglement.
In both cases, we identify an orbit state by estimating the expectation value of the observable $\rho$.
The difference lies only in the free implementation of this measurement.
For entanglement, we used a decomposition into local observables, whereas here we use a decomposition into Pauli observables.

For a state $\rho$ on $A$ and a state $\sigma$ on $B$, we denote the symmetry subgroup of $\rho$ by 
\begin{align}
    S:=\mathrm{Sym}(\rho)=\{g\in G \ |\ U_A(g)\rho U_A(g)^\dag=\rho\}
\end{align}
For each $g\in G$, we define the orbit state
\begin{align}
    \rho_g:=U_A(g)\rho U_A(g)^\dag.
\end{align}
The orbit state $\rho_g$ depends only on the coset $gS$.

We first construct a stabilizer measurement that identifies the coset $gS$.
The construction is based on the overlap between $\rho$ and its orbit states.
Define
\begin{align}
    c_g:=\mathrm{tr}\left(\rho\rho_g\right).
\end{align}
We note that 
\begin{align}
    c_g=\mathrm{tr}\left(\rho^2\right)-\frac{1}{2}\|\rho-\rho_g\|_2^2 
    \leq\mathrm{tr}\left(\rho^2\right), 
\end{align}
where equality holds if and only if $\rho_g=\rho$, or equivalently, if and only if $g\in S$.
Since $G$ is finite, the quantity
\begin{align}
    \Delta_\rho:=\mathrm{tr}\left(\rho^2\right)-\max_{g\notin S(\rho)}\mathrm{tr}\left(\rho\rho_g\right)
\end{align}
is strictly positive whenever $S(\rho)\neq G$.
If $S(\rho)=G$, there is only one coset and no discrimination is required.

Let $\mathcal{P}_A$ be a Hermitian Pauli basis on $A$ satisfying
\begin{align}
    \mathrm{tr}\left(PQ\right)=d_A\delta_{P,Q}
\end{align}
for every $P,Q\in\mathcal{P}_A$, where $d_A$ is the dimension of $A$.
The Pauli expansion of $\rho$ is
\begin{align}
    &\rho=\frac{1}{d_A}\sum_{P\in\mathcal{P}_A}r_PP, \\
    &r_P:=\mathrm{tr}\left(P\rho\right). 
\end{align}
Define
\begin{align}
    &L_\rho:=\sum_{P\in\mathcal{P}_A}|r_P|, \\
    &q_P:=\frac{|r_P|}{L_\rho}. 
\end{align}
To estimate the expectation value of $\rho$ on a state $\tau$, we choose $P\in\mathcal{P}_A$ with probability $q_P$ and measure $P$ on $\tau$.
If $z\in\{+1,-1\}$ denotes the measurement outcome, we define the classically processed outcome by
\begin{align}
    Y:=\frac{L_\rho}{d_A}\operatorname{sgn}(r_P)z.
\end{align}
The expectation value of $Y$ is
\begin{align}
    \mathbb{E}_\tau[Y]
    =\frac{1}{d_A} \sum_{P\in\mathcal{P}_A} r_P\mathrm{tr}\left(P\tau\right) 
    =\mathrm{tr}\left(\rho\tau\right).
\end{align}
Moreover, $Y$ satisfies
\begin{align}
    &-B_\rho\leq Y\leq B_\rho, \\
    &B_\rho:=\frac{L_\rho}{d_A}. 
\end{align}
Thus, the expectation value of the observable $\rho$ can be estimated using only classical randomness, Pauli measurements, and classical post-processing.

We now use this measurement to identify the orbit state.
Let the unknown state be $\rho_h$ and consider a candidate coset $gS$.
We first apply the Clifford unitary $U_A(g)^\dag$ and then perform the randomized Pauli measurement described above.
The expectation value of the resulting random variable is
\begin{align}
    \mathrm{tr}\left(\rho\rho_{g^{-1}h}\right).
\end{align}
If $gS=hS$, this expectation value is $\mathrm{tr}\left(\rho^2\right)$.
If $gS\neq hS$, it is at most $\mathrm{tr}\left(\rho^2\right)-\Delta_\rho$.

For each candidate coset, we repeat the measurement on $k$ independent copies and denote the sample mean by
\begin{align}
    \overline{Y}_{g,k}:=\frac{1}{k}\sum_{j=1}^kY_{g,j}.
\end{align}
We accept the candidate $gS$ if
\begin{align}
    \overline{Y}_{g,k}\geq\mathrm{tr}\left(\rho^2\right)-\frac{\Delta_\rho}{2}.
\end{align}
For the correct candidate, a rejection requires a lower deviation of at least $\Delta_\rho/2$.
For an incorrect candidate, an acceptance requires an upper deviation of at least $\Delta_\rho/2$.
The one-sided Hoeffding inequality therefore gives
\begin{align}
    \Pr[\text{an incorrect decision for a fixed candidate}]\leq\exp\left(-\alpha_\rho k\right),
\end{align}
where
\begin{align}
    \alpha_\rho:=\frac{\Delta_\rho^2}{8B_\rho^2}=\frac{d_A^2\Delta_\rho^2}{8L_\rho^2}.
\end{align}

Let
\begin{align}
    N_\rho:=|G/S|
\end{align}
be the number of distinct orbit states.
We divide $kN_\rho$ copies into $N_\rho$ blocks of size $k$ and test one candidate coset on each block.
If precisely one candidate is accepted, we output that candidate, and otherwise we output an arbitrary fixed candidate.
Applying the union bound to the one-sided error events gives a stabilizer POVM $\{M_x^{(k)}\}_{x\in G/S}$ such that
\begin{align}
    \min_{g\in G}\mathrm{tr}\left(M_{gS}^{(k)}\rho_g^{\otimes kN_\rho}\right)\geq 1-\epsilon_k,
\end{align}
where
\begin{align}
    \epsilon_k\leq N_\rho\exp\left(-\alpha_\rho k\right).
\end{align}
In particular, an error not larger than $\epsilon$ is achieved whenever
\begin{align}
    k\geq\frac{1}{\alpha_\rho}\log\frac{N_\rho}{\epsilon}.
\end{align}
Since $N_\rho$ is independent of the asymptotic number of input copies, the total number of copies required by this discrimination protocol is $O(\log(1/\epsilon))$.

We summarize this construction in the following lemma.

\begin{lemma}
    \label{lem:magic-free-discrimination}
    Let $G$ be a finite group represented by Clifford operators $U_A(g)$.
    For every state $\rho$, there exists a sequence of stabilizer POVMs $\{M_x^{(k)}\}_{x\in G/S}$ acting on $kN_\rho$ copies such that
    \begin{align}
        \min_{g\in G}\mathrm{tr}\left(M_{gS}^{(k)}\rho_g^{\otimes kN_\rho}\right)\geq 1-N_\rho e^{-\alpha_\rho k},
    \end{align}
    where
    \begin{align}
        \alpha_\rho:=\frac{d_A^2\Delta_\rho^2}{8L_\rho^2}>0.
    \end{align}
\end{lemma}

The stabilizer POVMs constructed above are control-admissible for both $\mathrm{SO}$ and $\mathrm{CSPO}$.
For $\mathrm{SO}$, this follows directly from closure under stabilizer measurements, classical feed-forward, and composition.
For $\mathrm{CSPO}$, consider an arbitrary joint stabilizer input, including an arbitrary reference system.
A stabilizer measurement on one subsystem produces an unnormalized stabilizer state on the remaining systems in each branch.
Applying a completely stabilizer-preserving channel conditionally on that branch preserves this property, including the reference system.
Summing over the outcomes gives a stabilizer state.
Thus, the resulting measurement-and-control channel is completely stabilizer preserving.

Moreover, both operation classes are closed under conjugation by Clifford unitaries and under classical randomization.
The discrimination construction therefore supplies the control-admissible POVMs required by Theorem~\ref{thm:rate_lower_bound_general}.
For example, taking $k_n=\lceil\sqrt{n}\rceil$ uses $k_nN_\rho=o(n)$ copies and gives a discrimination error tending to zero.
If $S(\rho)=G$, the trivial one-outcome POVM suffices.

\begin{theorem}[Map-level conversion rates for magic]
    \label{thm:magic_map_rates}
    Let $G$ be a finite group, and let $U_A$ and $U_B$ be projective unitary representations of $G$ on multiqubit systems $A$ and $B$, respectively.
    Suppose that $U_A(g)$ and $U_B(g)$ are Clifford operators for every $g\in G$.
    For $\mathcal{F}\in\{\mathrm{SO}, \mathrm{CSPO}\}$ and arbitrary states $\rho$ and $\sigma$ on $A$ and $B$,
    \begin{align}
        R_{G\text{-cov}, \mathcal{F}}(\rho\to\sigma)=&
        \begin{cases}
            R_\mathcal{F}(\rho\to\sigma) & \text{if }\mathrm{Sym}(\rho)\subset \mathrm{Sym}(\sigma),\\
            0 & \text{otherwise}.
        \end{cases}
        \label{eq:magic_map_rates}
    \end{align}
\end{theorem}

\begin{proof}
    The inclusion $\mathcal{F}^{G}\subset\mathcal{F}$ gives 
    \begin{align}
        R_{\mathcal{F}^{G}}(\rho\to\sigma)\leq& R_{\mathcal{F}}(\rho\to\sigma).
    \end{align}

    Suppose first that $\mathrm{Sym}(\rho)\not\subset \mathrm{Sym}(\sigma)$.
    Choose $g\in \mathrm{Sym}(\rho)$ such that $\mathcal{U}_{B,g}(\sigma)\neq\sigma$.
    For any $G$-covariant channel $\mathcal{E}$ from $n$ copies of $A$ to $m\geq 1$ copies of $B$, every single-copy marginal of $\mathcal{E}(\rho^{\otimes n})$ is invariant under $\mathcal{U}_{B,g}$.
    Taking one such marginal and using the triangle inequality and contractivity of the trace distance gives
    \begin{align}
        2\mathrm{T}\left(\mathcal{E}(\rho^{\otimes n}),  \sigma^{\otimes m}\right)\geq&
        \mathrm{T}\left(\sigma, \mathcal{U}_{B, g}(\sigma)\right)>0.
    \end{align}
    Hence, no positive conversion rate is possible.

    Suppose next that $S(\rho)\subset \mathrm{Sym}(\sigma)$ and set $S:=S(\rho)$.
    If $R_{\mathcal{F}}(\rho\to\sigma)=0$, the upper bound already proves the claim.
    Otherwise, take any $r$ satisfying $0<r<R_{\mathcal{F}}(\rho\to\sigma)$.
    There is a sequence of free operations $\mathcal{E}^{(n)}\in\mathcal{F}$ such that
    \begin{align}
        \lim_{n\to\infty}\mathrm{T}\left(\mathcal{E}^{(n)}(\rho^{\otimes n}),\sigma^{\otimes\lfloor rn\rfloor}\right)=0.
    \end{align}
    Since both $\rho$ and $\sigma$ are invariant under $S$, this sequence also satisfies
    \begin{align}
        \lim_{n\to\infty}\frac1{|S|}\sum_{s\in S}
        \mathrm{T}\left(\mathcal{E}^{(n)}\left(\mathcal{U}_{A,s}(\rho)^{\otimes n}\right),
        \mathcal{U}_{B,s}(\sigma)^{\otimes\lfloor rn\rfloor}\right)=&0.
    \end{align}
    Together with the control-admissible discrimination protocol, Theorem~\ref{thm:rate_lower_bound_general} implies
    \begin{align}
        R_{\mathcal{F}^{G}}(\rho\to\sigma)\geq&r.
    \end{align}
    Taking the supremum over $r<R_{\mathcal{F}}(\rho\to\sigma)$ proves the opposite inequality.
\end{proof}

Thus, arbitrary Clifford representations are allowed in the map-level rate theorem.
We next consider the additional requirements imposed by a symmetric implementation of a stabilizer protocol.

\subsection{Primitive-level stabilizer operations}
\label{sec:magic_primitive}

\subsection{Primitive-level stabilizer operations}

We now impose symmetry on the elementary operations of a stabilizer protocol.
Throughout this subsection, the symmetry may in general be represented by
Clifford operators.
A unitary primitive is a Clifford unitary that commutes with the symmetry
representation on the current joint quantum system.
A measurement primitive is the binary projective measurement of a Hermitian
Pauli operator $P$ satisfying
\begin{align}
    [P,U_Q(g)]
    =&0
    \qquad
    \forall g\in G,
\end{align}
where $Q$ includes all ancillary registers on which the measured operator acts.
The corresponding projections are
\begin{align}
    \Pi_{P,\pm}
    :=&\frac{I_Q\pm P}{2}.
\end{align}
Thus, each outcome projection is individually invariant under the symmetry.
In particular, covariance of a measurement that merely permutes the two
outcomes is not sufficient for the measurement to be an allowed primitive.

We also allow the introduction of symmetry-invariant stabilizer states,
classical randomness, classical feed-forward, and the discarding of
subsystems.
Stabilizer states may be mixed, with mixed stabilizer states understood as
convex combinations of pure stabilizer states.
The classical registers containing measurement outcomes and the protocol
transcript transform trivially under the symmetry.
The symmetry representation on a composite system is the tensor product of
the representations on its constituent registers.
Ancillary quantum registers may carry specified Clifford representations of
$G$, provided that every elementary operation satisfies the symmetry
requirements above.

We denote the resulting class of channels by
$\mathrm{SO}_{\mathrm{prim}}^G$.
Every allowed quantum primitive is a stabilizer operation and is
$G$-covariant, including each individual measurement branch.
Consequently, classical conditioning on a symmetry-neutral transcript and
composition preserve both properties, and
\begin{align}
    \mathrm{SO}_{\mathrm{prim}}^G(A\to B)
    \subset&
    \mathrm{SO}^G(A\to B).
\end{align}

We next identify a class of symmetry representations for which the converse
inclusion can be established by an explicit primitive-level implementation.
Specifically, suppose that both $U_A(g)$ and $U_B(g)$ are Pauli operators up
to phases for every $g\in G$.
Pauli symmetries are also known to play a distinguished role in the structure
of symmetric Clifford operations~\cite{mitsuhashi2023clifford}.
The role of the Pauli assumption here is operationally different: it allows
the symmetry-dependent extensions of Clifford unitaries and Pauli
measurements constructed below to remain within the allowed stabilizer
primitives.
We define the common projective kernel by
\begin{align}
    K:=\{g\in G \,|\, U_A(g)\propto I_A
    \text{ and } U_B(g)\propto I_B\},
\end{align}
and 
\begin{align}
    \overline{G}:=G/K.
\end{align}
Since $\overline{G}$ embeds into the product of the input and output Pauli
groups modulo phases, there exists some integer $r\geq 0$ such that
\begin{align}
    \overline{G}\simeq (\mathbb{Z}_2)^r.
\end{align}
We identify the elements of $\overline{G}$ with
$x\in\mathbb{F}_2^r$. 
The use of the common projective kernel is important because the input and
output representations need not have the same projective kernel.

A stabilizer protocol can be realized by appending quantum registers, applying Clifford unitaries and Pauli measurements, and discarding unused outputs.
The prescribed input and output registers retain their given representations, while additional workspace registers may carry trivial representations.
Thus, all quantum registers needed in such a realization can be chosen to carry Pauli representations that factor projectively through $\overline G$.
For any current quantum register $Q$, write $U_Q(x)$ for a representative of its action.
The choice of its phase will not affect any construction below.

Introduce an $r$-qubit auxiliary register $M$ with representation
\begin{align}
    V(x):=&\bigotimes_{j=1}^r X_{M_j}^{x_j}.
\end{align}
Its computational basis satisfies $V(y)|x\rangle=|x+y\rangle$.
We initialize this register in
\begin{align}
    \zeta_M:=&2^{-r}\sum_{x\in\mathbb F_2^r}|x\rangle\langle x|
    =\frac{I_M}{2^r},
\end{align}
which is an invariant stabilizer state.
The basis states $|x\rangle$ will be used to describe the blocks of symmetric operations.
We do not prepare an individual $|x\rangle$ or read out $x$ by measuring the memory qubits.

For an operator $A$ on $Q$, define
\begin{align}
    \widehat A:=&
    \sum_{x\in\mathbb F_2^r}|x\rangle\langle x|_M
    \otimes U_Q(x)A U_Q(x)^\dag.
    \label{eq:magic_Pauli_extension}
\end{align}
The projective phases cancel in conjugation.
Using $V(y)|x\rangle=|x+y\rangle$ and relabeling the sum gives
\begin{align}
    \left(V(y)\otimes U_Q(y)\right)\widehat A
    \left(V(y)\otimes U_Q(y)\right)^\dag=&\widehat A
\end{align}
for every $y\in\mathbb F_2^r$.
We next verify that this extension preserves the required stabilizer primitives.

For a Clifford unitary $C$ on $Q$, the operator $\widehat C$ is a symmetric Clifford unitary.
To verify the Clifford property without discarding relative phases between memory blocks, let $e_j$ be the $j$th standard basis vector of $\mathbb F_2^r$, and choose a Hermitian Pauli operator $P_{Q,j}$ proportional to $U_Q(e_j)$.
Define
\begin{align}
    D_Q:=&
    \prod_{j=1}^r
    \left(|0\rangle\langle0|_{M_j}\otimes I_Q
    +|1\rangle\langle1|_{M_j}\otimes P_{Q,j}\right),
\end{align}
where the product is taken in a fixed order and each factor acts trivially on the remaining memory qubits.
Every factor is a controlled Hermitian Pauli gate and is therefore Clifford.
On the memory block $x$, the product $\prod_jP_{Q,j}^{x_j}$ is proportional to $U_Q(x)$.
Consequently,
\begin{align}
    \widehat C=&D_Q\left(I_M\otimes C\right)D_Q^\dag,
    \label{eq:magic_extended_Clifford}
\end{align}
which proves that $\widehat C$ is Clifford.
This factorization is used only to establish the Clifford property.
The allowed symmetric primitive is $\widehat C$ itself; the factors defining $D_Q$ are not assumed to be individually symmetric.

For a Hermitian Pauli operator $P$ on $Q$, conjugation by the Pauli representation changes only its sign.
The sign is a character of $\mathbb F_2^r$, so there is a vector $a(P)\in\mathbb F_2^r$ such that
\begin{align}
    U_Q(x)P U_Q(x)^\dag=&(-1)^{a(P)\cdot x}P.
\end{align}
Its extension is therefore
\begin{align}
    &\widehat P=Z_M^{a(P)}\otimes P, \\
    &Z_M^{a(P)}:=\bigotimes_{j=1}^r Z_{M_j}^{a(P)_j}.
    \label{eq:magic_dressed_Pauli}
\end{align}
This is a Hermitian Pauli operator commuting with $V(y)\otimes U_Q(y)$ for every $y$.
Indeed, the signs acquired by its memory and system factors cancel.
Its binary measurement is an allowed symmetric primitive, with projections
\begin{align}
    \widehat\Pi_{P,\pm}:=&\frac{I_{MQ}\pm\widehat P}{2}\nonumber\\
    =&\sum_{x\in\mathbb F_2^r}|x\rangle\langle x|_M
    \otimes U_Q(x)\Pi_{P,\pm}U_Q(x)^\dag.
    \label{eq:magic_extended_Pauli_measurement}
\end{align}
Thus, each memory block implements the corresponding conjugated Pauli measurement, while the outcome labels remain symmetry neutral.
No separate readout of the memory label is required.

Ancillary-state introduction can also be implemented while retaining the same memory label.
Let $\tau$ be a stabilizer state on an ancillary register $Q$.
First append the maximally mixed state $I_Q/\dim Q$, which is invariant and stabilizer.
An ordinary stabilizer preparation of a pure component of $\tau$ can be performed by computational-basis measurements with Pauli corrections, followed by a Clifford unitary.
Classical randomization then prepares a mixed stabilizer state.
Replace each unitary and measurement in this preparation by its symmetric extension using $M$.
For every $x$, the resulting preparation has the action
\begin{align}
    |x\rangle\langle x|_M\otimes\frac{I_Q}{\dim Q}
    \longmapsto&
    |x\rangle\langle x|_M\otimes U_Q(x)\tau U_Q(x)^\dag.
    \label{eq:magic_symmetric_ancilla_preparation}
\end{align}
All operations in this preparation are symmetric stabilizer primitives.
In particular, this construction does not assume access to an asymmetric ancillary state.

Now fix a stabilizer realization of a channel $\mathcal{E}:A\to B$.
Replace every Clifford unitary, Pauli measurement, and ancillary-state preparation by the corresponding construction above, using the same register $M$ throughout the protocol.
Classical feed-forward remains valid because the transcript is symmetry neutral.
Each extended operation preserves the memory blocks, and discarding a subsystem is covariant.
The resulting primitive-level channel $\widehat{\mathcal{E}}$ therefore satisfies
\begin{align}
    \widehat{\mathcal{E}}\left(\zeta_M\otimes L\right)=&
    2^{-r}\sum_{x\in\mathbb F_2^r}|x\rangle\langle x|_M
    \otimes
    \left(\mathcal{U}_{B,x}\circ\mathcal{E}\circ
    \mathcal{U}_{A,x}^{-1}\right)(L),
    \label{eq:magic_symmetric_protocol_extension}
\end{align}
where $\mathcal{U}_{Q,x}:=U_Q(x)\cdot U_Q(x)^\dag$.
If $\mathcal{E}$ is $G$-covariant, every channel in the sum equals $\mathcal{E}$, giving
\begin{align}
    \widehat{\mathcal{E}}\left(\zeta_M\otimes L\right)=&
    \zeta_M\otimes\mathcal{E}(L).
    \label{eq:magic_exact_covariant_implementation}
\end{align}
Since $\zeta_M$ can be freely introduced and discarded, this is an exact symmetric primitive implementation of $\mathcal{E}$.

This implementation gives the following conversion-rate result.

\begin{theorem}[Primitive-level conversion rates for Pauli symmetries]
    \label{thm:magic_Pauli_primitive_rates}
    Let $G$ be a finite group, and suppose that both $U_A(g)$ and $U_B(g)$ are Pauli operators up to phases for every $g\in G$.
    In the primitive-level theory specified above,
    \begin{align}
        R_{G\text{-cov}, \mathrm{SO}_{\mathrm{prim}}}(\rho\to\sigma)=&
        \begin{cases}
            R_{\mathrm{SO}}(\rho\to\sigma) & \text{if } \mathrm{Sym}(\rho)\subset \mathrm{Sym}(\sigma),\\
            0 & \text{otherwise}.
        \end{cases}
    \end{align}
\end{theorem}

\begin{proof}
    Every primitive-level symmetric stabilizer channel belongs to $\mathrm{SO}^{G}$.
    Conversely, Eq.~\eqref{eq:magic_exact_covariant_implementation} gives a symmetric primitive implementation of every $G$-covariant stabilizer protocol under the stated Pauli representations.
    Applying this construction to each channel in a conversion sequence gives
    \begin{align}
        R_{G\text{-cov}, \mathrm{SO}_{\mathrm{prim}}}(\rho\to\sigma) 
        =R_{G\text{-cov}, \mathrm{SO}}(\rho\to\sigma).
    \end{align}
    The claimed expression then follows from Theorem~\ref{thm:magic_map_rates}.
\end{proof}

The Pauli assumption is sufficient for this construction because it ensures that the extension of every Hermitian Pauli observable is again a single Hermitian Pauli observable.
For general Clifford representations, conjugation may instead map a Pauli observable to a distinct Pauli observable.
The extension in Eq.~\eqref{eq:magic_Pauli_extension} is then not generally a Pauli operator, so it need not define an allowed measurement primitive.

A concrete obstruction is provided by the projective representation of $\mathbb Z_3$ generated by $C=SH$ on a qubit.
Its conjugation action is $CXC^\dag=Z$, $CZC^\dag=Y$, $CYC^\dag=X$. 
Hence, for every $n\geq1$, conjugation by $C^{\otimes n}$ maps each nonidentity Pauli string to a different Pauli string, even modulo a sign.
The only Hermitian Pauli operators commuting with $C^{\otimes n}$ are therefore proportional to the identity.

This obstruction cannot be removed merely by adjoining another Clifford-represented ancillary system.
Let $V$ denote the operator representing the generator of $\mathbb Z_3$ on such an ancillary system.
If a joint Hermitian Pauli operator $P_A\otimes P_B$ satisfies
\begin{align}
    [P_A\otimes P_B,C^{\otimes n}\otimes V]=&0,
\end{align}
then equality of the two tensor-product operators implies
\begin{align}
    C^{\otimes n}P_A C^{\dag\otimes n}\propto&P_A.
\end{align}
The Pauli permutation described above therefore forces $P_A\propto I_A$.
Thus, every symmetric Pauli observable on the enlarged system acts trivially on the original $n$ qubits.
The same conclusion holds after symmetric Clifford preprocessing, since conjugation by a symmetric Clifford unitary maps a symmetric Hermitian Pauli observable to another symmetric Hermitian Pauli observable.

This example explains why the Pauli-measurement construction does not extend directly to arbitrary Clifford representations.
It does not assert that every non-Pauli Clifford representation has the same obstruction, nor does it establish a separation between map-level and primitive-level asymptotic conversion rates.
A general characterization of primitive-level conversion under Clifford symmetries is left outside the scope of this work.

\section{Thermodynamics under finite-group symmetry}

We next apply the general framework to thermodynamic resource theories.
We consider three classes of free operations: Gibbs-preserving operations, time-translation-covariant Gibbs-preserving operations, and thermal operations.
Although all three settings can be analyzed within the general framework developed above, the amount of additional structure required to verify its assumptions differs substantially among them.

For Gibbs-preserving operations, the discrimination procedure required in Theorem~\ref{thm:rate_lower_bound_general} can be incorporated directly into a Gibbs-preserving map.
The resulting application is therefore essentially immediate.
When time-translation covariance is imposed in addition to Gibbs preservation, the situation becomes more restrictive.
In this case, only the information that remains after dephasing with respect to the system Hamiltonian can be used for the discrimination step.
Consequently, the relevant distinguishability condition is determined by the orbit of the dephased state rather than by that of the original state itself.
This introduces additional cases in which the conversion rate vanishes.
Thermal operations require a more careful analysis.
Here, it is not sufficient to appeal directly to an arbitrary discrimination map.
Instead, we explicitly construct a projective measurement that is compatible with the Hamiltonian and that asymptotically resolves the relevant group orbit.
This construction verifies the assumptions of Theorem~\ref{thm:rate_lower_bound_general} within the class of thermal operations and thereby yields the corresponding map-level conversion result.

We then strengthen the thermal-operation result at the level of physical primitives.
Namely, we show that the same discrimination structure can be incorporated coherently into a single unitary implementation that is both energy preserving and invariant under the additional symmetry.
Thus, the asymptotic conversion result continues to hold even when the symmetry constraint is imposed directly on the unitary primitives rather than only on the resulting quantum channel.

Finally, we apply the same mechanism to work extraction.
We define the symmetry-protected ergotropy rate by restricting the allowed unitaries to those that are invariant under the symmetry, while not requiring them to commute with the Hamiltonian.
We show that, for a finite symmetry group, this restriction does not change the asymptotic ergotropy rate.
As a consequence, the symmetry-protected completely passive states are characterized by the usual completely passive states, which is contrasting to the result~\cite{mitsuhashi2022characterizing} with Lie-group symmetry constraints, where generalized Gibbs states are characterized as completely passive states under symmetry constraints.

Throughout this section, we assume that the finite-group symmetry is compatible with the Hamiltonian.
More precisely, for every input and output system considered below, the symmetry representation commutes with the corresponding Hamiltonian:
\begin{align}
    [U(g),H]=0,\ 
    [U'(g),H']=0\ \forall g\in G.
    \label{eq:symmetry_hamiltonian_commutation}
\end{align}
Consequently, the finite-group action commutes with time translations and with the corresponding dephasing maps.

\subsection{Gibbs-preserving operations}

We first consider how symmetry constraints affect the asymptotic conversion rate via Gibbs-preserving operations.

We denote the Gibbs state with inverse temperature $\beta$ on system $\mathrm{S}$ 
\begin{align}
    \gamma_\mathrm{S}(\beta) =\frac{e^{-\beta H_\mathrm{S}}}{\mathrm{tr}\left(e^{-\beta H_\mathrm{S}}\right)}, 
\end{align}
where $H_\mathrm{S}$ is the Hamiltonian of the system $\mathrm{S}$. 
For some fixed inverse temperature $\beta>0$, a quantum channel $\mathcal{E}$ from system $\mathrm{A}$ to system $\mathrm{B}$ is called Gibbs preserving if 
\begin{align}
    \mathcal{E}(\gamma_\mathrm{A}(\beta))=\gamma_\mathrm{B}(\beta). 
\end{align}
For states $\rho$ and $\sigma$, we denote the asymptotic conversion rate from $\rho$ to $\sigma$ via Gibbs-preserving operations by $R_\mathrm{GP}(\rho\to\sigma)$, and the asymptotic conversion rate from $\rho$ to $\sigma$ via $G$-covariant Gibbs-preserving operations by $R_{G\text{-}\mathrm{cov}, \mathrm{GP}}(\rho\to\sigma)$. 
Then, $R_{G\text{-}\mathrm{cov}, \mathrm{GP}}(\rho\to\sigma)$ is expressed in the following theorem.

\begin{lemma} \label{lem:control_admissible_POVM_GP}
    Let $G$ be a finite group, 
    Then, any sub-POVM is a control-admissible sub-POVM in the resource theory of Gibbs-preserving operations. 
    Especially, any POVM is a control-admissible POVM. 
\end{lemma}

\begin{proof}
    We take an arbitrary sub-POVM $\{M_x\}_{x\in X}$, and arbitrary Gibbs-preserving operations $\{\mathcal{E}_x\}_{x\in X}$ and $\mathcal{E}_\perp$. 
    We define operation $\mathcal{E}$ by 
    \begin{align}
        \mathcal{E}(L_1\otimes L_2)
        :=\sum_{x\in X\cup\{\perp\}} \mathrm{tr}(M_x L_1)\mathcal{E}_x(L_2) 
    \end{align}
    with 
    \begin{align}
        M_\perp:=I-\sum_{x\in X} M_x. 
    \end{align}
    Then, we can confirm that $\mathcal{E}$ is Gibbs preserving. 
    Indeed, we have 
    \begin{align}
        \mathcal{E}(\gamma_{\mathrm{A}_1\mathrm{A}_2}(\beta)) 
        =&\mathcal{E}(\gamma_{\mathrm{A}_1}(\beta)\otimes \gamma_{\mathrm{A}_2}(\beta)) \nonumber\\
        =&\sum_{x\in X\cup\{\perp\}} \mathrm{tr}(M_x \gamma_{\mathrm{A}_1}(\beta))\mathcal{E}_x(\gamma_{\mathrm{A}_2}(\beta)) \nonumber\\
        =&\left(\sum_{x\in X\cup\{\perp\}} \mathrm{tr}(M_x \gamma_{\mathrm{A}_1}(\beta))\right)\gamma_\mathrm{B}(\beta) \nonumber\\
        =&\gamma_\mathrm{B}(\beta). 
    \end{align}
    Therefore, any sub-POVM is control-admissible in the resource theory of Gibbs-preserving operations. 
    The statement about POVM is a special case of this argument. 
\end{proof}

\begin{theorem} \label{thm:GP_conversion_rate}
    Let $G$ be a finite group, $\rho$ and $\sigma$ be input and output states, and $U$ and $U'$ be projective unitary representations of $G$ on the input and output space. 
    Then,
    \begin{align}
        R_{G\text{-}\mathrm{cov}, \mathrm{GP}}(\rho\to\sigma)=
        \begin{cases}
            R_\mathrm{GP}(\rho\to\sigma) & \text{if } \mathrm{Sym}_{G, U}(\rho)\subset \mathrm{Sym}_{G, U'}(\sigma), \\
            0 & \text{otherwise}.
        \end{cases}
    \end{align}
\end{theorem}

\begin{proof}
    Since the class of $G$-covariant Gibbs-preserving operations is smaller than the class of $G$-covariant operations and Gibbs-preserving operations, we trivially have 
    \begin{align}
        R_{G\text{-}\mathrm{cov}, \mathrm{GP}}(\rho\to\sigma)
        \leq\min\left\{R_{G\text{-}\mathrm{cov}}(\rho\to\sigma), R_\mathrm{GP}(\rho\to\sigma)\right\}
    \end{align}
    Combining with the result of Theorem~\ref{thm:asymmetry_rate}, we have 
    \begin{align}
        R_{G\text{-}\mathrm{cov}, \mathrm{GP}}(\rho\to\sigma)\leq 
        \begin{cases}
            R_\mathrm{GP}(\rho\to\sigma) & \text{if } \mathrm{Sym}_{G, U}(\rho)\subset \mathrm{Sym}_{G, U'}(\sigma), \\
            0 & \text{otherwise}.
        \end{cases}
    \end{align}
    Thus, in the following, it is sufficient to show that 
    \begin{align}
        R_{G\text{-}\mathrm{cov}, \mathrm{GP}}(\rho\to\sigma)\geq R_\mathrm{GP}(\rho\to\sigma) 
    \end{align}
    when $\mathrm{Sym}_{G, U}(\rho)\subset \mathrm{Sym}_{G, U'}(\sigma)$. 
    When $R_\mathrm{GP}(\rho\to\sigma)=0$, it is trivial that $R_{G\text{-}\mathrm{cov}, \mathrm{GP}}(\rho\to\sigma)\geq R_\mathrm{GP}(\rho\to\sigma)$. 
    We assume that $R_\mathrm{GP}(\rho\to\sigma)>0$ in the following. 
    We take arbitrary $r<R_\mathrm{GP}(\rho\to\sigma)$. 
    Then, there exists some sequence of Gibbs preserving operations $(\mathcal{E}_n)_{n\in\mathbb{N}}$ such that 
    \begin{align}
        \lim_{n\to\infty} \mathrm{T}\left(\mathcal{E}_n(\rho^{\otimes n}), \sigma^{\otimes \lfloor rn\rfloor}\right)=0. \label{eq:thm:GP_conversion_rate1}
    \end{align}
    We denote $S:=\mathrm{Sym}_{G, U}(\rho)$ for convenience. 
    Since $\rho$ and $\sigma$ are invariant under the action of $s\in S$, Eq.~\eqref{eq:thm:GP_conversion_rate1} implies 
    \begin{align}
        \lim_{n\to\infty} \frac{1}{|S|}\sum_{s\in S} \mathrm{T}\left(\mathcal{E}_n(\mathcal{U}_s(\rho)^{\otimes n}), \mathcal{U}'_s(\sigma)^{\otimes \lfloor rn\rfloor}\right)=0. \label{eq:thm:GP_conversion_rate2}
    \end{align}
    By Lemma~\ref{lem:asymmetry_state_discrimination}, there exist some sequence of POVMs $(\{M_x^{(k)}\}_{x\in G/S})_{k\in\mathbb{N}}$ and $\alpha>0$ such that 
    \begin{align}
        \mathrm{tr}(M_{gS}^{(k)}\mathcal{U}_{g}(\rho^{\otimes k}))\geq 1-e^{-\alpha k} 
    \end{align}
    for all $g\in G$ and $k\in\mathbb{N}$. 
    Since this POVM is control admissible in the resource theory of Gibbs-preserving operations by Lemma~\ref{lem:control_admissible_POVM_GP}, there exist some POVM $\{M_x^{(k)}\}_{x\in G/S}$ and $\alpha>0$ such that 
    \begin{align}
        \lim_{k\to\infty} \frac{1}{|G|}\sum_{g\in G}\mathrm{tr}(M_{gS}^{(k)}\mathcal{U}_{g}(\rho^{\otimes k}))=1. \label{eq:thm:GP_conversion_rate3}
    \end{align}
    By Eqs.~\eqref{eq:thm:GP_conversion_rate2} and \eqref{eq:thm:GP_conversion_rate3}, we can confirm that the assumptions in Theorem~\ref{thm:rate_lower_bound_general} are satisfied, and we get 
    \begin{align}
        R_{G\text{-}\mathrm{cov}, \mathrm{GP}}(\rho\to\sigma)
        \geq r. 
    \end{align}
    Since this holds for all $r\in (0, R_\mathrm{GP}(\rho\to\sigma))$, we have 
    \begin{align}
        R_{G\text{-}\mathrm{cov}, \mathrm{GP}}(\rho\to\sigma)
        \geq R_\mathrm{GP}(\rho\to\sigma). 
    \end{align}
\end{proof}

\subsection{Time-translation-covariant Gibbs-preserving operations}

We next consider Gibbs-preserving operations that are covariant under time translations.
A channel $\mathcal{E}$ from a system with Hamiltonian $H$ to a system with Hamiltonian $H'$ is time-translation covariant if
\begin{align}
    \mathcal{E}\left(e^{-itH}Le^{itH}\right)
    = e^{-itH'}\mathcal{E}(L)e^{itH'}
\end{align}
for all $L$ and $t\in\mathbb{R}$.

First, we clarify the class of control-admissible sub-POVM in the resource theory of time-translation covariant Gibbs-preserving operations.

\begin{lemma} \label{lem:control_admissible_POVM_TCGP}
    Let $X$ be a finite set, $H$ be a Hamiltonian, $\{M_x\}_{x\in X}$ be a sub-POVM, each $M_x$ commute with $H$. 
    Then, $\{M_x\}_{x\in X}$ is a control-admissible sub-POVM in the resource theory of time-translation covariant Gibbs-preserving operations. 
\end{lemma}

\begin{proof}
    We take arbitrary input and output systems, and for $x\in X\cup\{\perp\}$, we take arbitrary free operations $\mathcal{E}_x$. 
    We define a map $\mathcal{E}$ by 
    \begin{align}
        \mathcal{E}(L)
        :=\sum_{x\in X\cup\{\perp\}} \mathcal{E}_x(\mathrm{tr}_1((M_x\otimes I) L)), 
    \end{align}
    where 
    \begin{align}
        M_\perp:=I-\sum_{x\in X} M_x. 
    \end{align}
    We note that $M_\perp$ is also time-translation invariant. 
    As we showed in Lemma~\ref{lem:control_admissible_POVM_GP}, the map $\mathcal{E}$ is Gibbs-preserving. 
    We denote the input and output Hamiltonian by $H^\mathrm{in}$ and $H^\mathrm{out}$. 
    Then, for any $t\in\mathbb{R}$, we have 
    \begin{align}
        \mathcal{E}\left(e^{-it(H\otimes I+I\otimes H^\mathrm{in})}L e^{it(H\otimes I+I\otimes H^\mathrm{in})}\right) 
        =&\sum_{x\in X\cup\{\perp\}} \mathcal{E}_x\left(e^{-itH^\mathrm{in}}\mathrm{tr}_1\left(\left(e^{itH}M_x e^{-itH}\otimes I\right) L \right)e^{itH^\mathrm{in}}\right) \nonumber\\
        =&\sum_{x\in X\cup\{\perp\}} e^{-itH^\mathrm{out}}\mathcal{E}_x(\mathrm{tr}_1((M_x\otimes I) L))e^{itH^\mathrm{out}} \nonumber\\
        =&e^{-itH^\mathrm{out}}\mathcal{E}(L)e^{itH^\mathrm{out}}, 
    \end{align}
    which implies that $\mathcal{E}$ is time-translation covariant. 
    Therefore, $\{M_x\}$ is control-admissible POVM in the resource theory of time-translation covariant Gibbs-preserving operations. 
\end{proof}

This time-translation invariant POVM has some limitation about state distinguishability.  
To explicitly describe the limitation, we introduce the compatibility condition between the finite-group covariance and time-translation covariance.

\begin{definition} (Compatibility of covariance conditions.) \label{def:group_time_covariance_compatibility}
    Let $G$ be a finite group, $U$ and $U'$ be projective unitary representations of $G$ on $\mathcal{H}$ and $\mathcal{H}'$, respectively, $\rho\in\mathcal{S}(\mathcal{H})$, $\sigma\in\mathcal{S}(\mathcal{H}')$, and $H\in\mathcal{L}^\mathrm{H}(\mathcal{H})$ and $H'\in\mathcal{L}^\mathrm{H}(\mathcal{H}')$ be Hamiltonians on the input and output systems. 
    We say that $(G, U, U')$-covariance and time-translation $(H, H')$-covariance are compatible for the conversion from $\rho$ to $\sigma$ if for any $g\in G$ and any real sequence $(t_j)_{j\in\mathbb{N}}$ satisfying 
    \begin{align}
        \lim_{j\to\infty} e^{-it_j H}\rho e^{it_j H}=U(g)\rho U(g)^\dag, 
    \end{align}
    we have 
    \begin{align}
        \lim_{j\to\infty} e^{-it_j H'}\sigma e^{it_j H'}=U'(g)\sigma U'(g)^\dag. 
    \end{align}
\end{definition}

We note that the compatibility condition is stronger than $\mathrm{Sym}(\rho)\subset\mathrm{Sym}(\sigma)$, which can be confirmed as follows. 
We assume the compatibility condition and take arbitrary $g\in\mathrm{Sym}(\rho)$. 
Then, $(t_j)_{j=0}^\infty$ defined by $t_j=0$ for all $j\in\mathbb{Z}^{\geq 0}$ satisfies $\lim_{j\to\infty} e^{-it_j H}\rho e^{it_j H}=U(g)\rho U(g)^\dag$. 
The compatibility condition implies $\lim_{j\to\infty} e^{-it_j H'}\sigma e^{it_j H'}=U'(g)\sigma U'(g)^\dag$, which means $g\in\mathrm{Sym}(\sigma)$. 
This compatibility condition can be generalized to arbitrary two groups, and in that case, it can be expressed as: For any sequences $(g_n)_{n=0}^\infty$ and $(g'_n)_{n=0}^\infty$ in $G$ and $G'$, $\lim_{n\to\infty} T(\mathcal{U}_{g_n}^\mathrm{in}(\rho), {\mathcal{U}'}_{g'_n}^\mathrm{in}(\rho))=0$ implies $\lim_{n\to\infty} T(\mathcal{U}_{g_n}^\mathrm{out}(\sigma), {\mathcal{U}'}_{g'_n}^\mathrm{out}(\sigma))=0$.

\begin{lemma} \label{lem:group_time_covariance_compatibility}
    Let $\rho$ and $\sigma$ be states, and the conversion rate $R_{G\text{-}\mathrm{cov}, \mathrm{TC}}(\rho\to\sigma)$ via time-translation covariant operations satisfy $R_{G\text{-}\mathrm{cov}, \mathrm{TC}}(\rho\to\sigma)>0$. 
    Then, the compatibility condition (Def.~\ref{def:group_time_covariance_compatibility}) holds. 
\end{lemma}

\begin{proof}
    We take arbitrary $g\in G$ and $(t_j)_{j=0}^\infty$ such that $\lim_{j\to\infty} e^{-it_j H}\rho e^{it_j H}=U(g)\rho U(g)^\dag$, and we show $\lim_{j\to\infty} e^{-it_j H'}\sigma e^{it_j H'}=U'(g)\sigma U'(g)^\dag$ in the following.

    We take arbitrary $\epsilon>0$. 
    Since $R_{G\text{-}\mathrm{cov}, \mathrm{TC}}(\rho\to\sigma)$ is nonzero, we can take some $r>0$ satisfying $r<R_{G\text{-}\mathrm{cov}, \mathrm{TC}}(\rho\to\sigma)$. 
    Then, by definition, there exists some $n_0\in\mathbb{N}$ such that for any $n\geq n_0$ and $m\leq rn$, $\rho^{\otimes n}$ can be converted to $\sigma^{\otimes m}$ via some $(G, U^{\otimes n}, U'^{\otimes m})$-covariant and time-translation $(H, H')$-covariant map with error $\epsilon/4$. 
    We set $n=\max\{n_0, \lceil 1/r \rceil\}$ and $m=1$. 
    Then, we can take a $(G, U^{\otimes n}, U')$-covariant and time-translation $(H, H')$-covariant map $\mathcal{E}: \mathcal{L}(\mathcal{H}^{\otimes n})\to\mathcal{L}(\mathcal{H}')$ that satisfies 
    \begin{align}
        T(\mathcal{E}(\rho^{\otimes n}), \sigma)\leq \frac{\epsilon}{4}. \label{SMeq:lem:group_time_covariance_compatibility1}
    \end{align}

    Since $g\in G$ and $(t_j)_{j=0}^\infty$ satisfies $\lim_{j\to\infty} e^{-it_j H}\rho e^{it_j H}=U(g)\rho U(g)^\dag$, we can take $j_0\in\mathbb{N}$ such that for any $j\geq j_0$, $T(e^{-it_j H}\rho e^{it_j H}, U(g)\rho U(g)^\dag)\leq \epsilon/(2n\|\mathcal{E}\|)$, where $\|\mathcal{E}\|:=\max_{\|L\|_1\leq 1}\|\mathcal{E}(L)\|_1$. 
    Thus, for any $j\geq j_0$, we have 
    \begin{align}
        T\left(\mathcal{E}\left(\left(e^{-it_j H}\rho e^{it_j H}\right)^{\otimes n}\right), \mathcal{E}\left(\left(U(g)\rho U(g)^\dag\right)^{\otimes n}\right)\right) 
        \leq &\|\mathcal{E}\|T\left(\left(e^{-it_j H}\rho e^{it_j H}\right)^{\otimes n}, \left(U(g)\rho U(g)^\dag\right)^{\otimes n}\right) \nonumber\\
        \leq &\|\mathcal{E}\|\cdot nT\left(e^{-it_j H}\rho e^{it_j H}, U(g)\rho U(g)^\dag\right) \nonumber\\
        \leq &n\|\mathcal{E}\|\cdot\frac{\epsilon}{2n\|\mathcal{E}\|} \nonumber\\
        =&\frac{\epsilon}{2}. \label{SMeq:lem:group_time_covariance_compatibility2}
    \end{align}
    Since $\mathcal{E}$ is $(G, U^{\otimes n}, U')$-covariant and time-translation $(H, H')$-covariant, we have 
    \begin{align}
        &\mathcal{E}\left(\left(U(g)\rho U(g)^\dag\right)^{\otimes n}\right)
        =U'(g)\mathcal{E}(\rho^{\otimes n})U'(g)^\dag, \label{SMeq:lem:group_time_covariance_compatibility3}\\
        &\mathcal{E}\left(\left(e^{-it_j H}\rho e^{it_j H}\right)^{\otimes n}\right)
        =e^{-it_j H'}\mathcal{E}(\rho^{\otimes n})e^{it_j H'}. \label{SMeq:lem:group_time_covariance_compatibility4}
    \end{align}
    By plugging Eqs.~\eqref{SMeq:lem:group_time_covariance_compatibility3} and \eqref{SMeq:lem:group_time_covariance_compatibility4} into Eq.~\eqref{SMeq:lem:group_time_covariance_compatibility2}, we get 
    \begin{align}
        T\left(e^{-it_j H'}\mathcal{E}(\rho^{\otimes n})e^{it_j H'}, U'(g)\mathcal{E}(\rho^{\otimes n})U'(g)^\dag\right)\leq\frac{\epsilon}{2}. \label{SMeq:lem:group_time_covariance_compatibility5}
    \end{align}
    Eq.~\eqref{SMeq:lem:group_time_covariance_compatibility1} implies 
    \begin{align}
        &T\left(e^{-it_j H'}\mathcal{E}(\rho^{\otimes n})e^{it_j H'}, e^{-it_j H'}\sigma e^{it_j H'}\right)\leq\frac{\epsilon}{4}, \label{SMeq:lem:group_time_covariance_compatibility6}\\
        &T\left(U'(g)\mathcal{E}(\rho^{\otimes n})U'(g)^\dag, U'(g)\sigma U'(g)^\dag\right)\leq\frac{\epsilon}{4}. \label{SMeq:lem:group_time_covariance_compatibility7}
    \end{align}
    By the triangle inequality and Eqs.~\eqref{SMeq:lem:group_time_covariance_compatibility5}, \eqref{SMeq:lem:group_time_covariance_compatibility6}, and \eqref{SMeq:lem:group_time_covariance_compatibility7}, we get
    \begin{align}
        T\left(e^{-it_j H'}\sigma e^{it_j H'}, U'(g)\sigma U'(g)^\dag\right) 
        \leq &T\left(e^{-it_j H'}\sigma e^{it_j H'}, e^{-it_j H'}\mathcal{E}(\rho^{\otimes n})e^{it_j H'}\right) \nonumber\\
        &+T\left(e^{-it_j H'}\mathcal{E}(\rho^{\otimes n})e^{it_j H'}, U'(g)\mathcal{E}(\rho^{\otimes n})U'(g)^\dag\right) \nonumber\\
        &+T\left(U'(g)\mathcal{E}(\rho^{\otimes n})U'(g)^\dag, U'(g)\sigma U'(g)^\dag\right) \nonumber\\
        \leq &\frac{\epsilon}{4}+\frac{\epsilon}{2}+\frac{\epsilon}{4} 
        \nonumber\\
        =&\epsilon. 
    \end{align}
    Therefore, we get $\lim_{j\to\infty} e^{-it_j H'}\sigma e^{it_j H'}=U'(g)\sigma U'(g)^\dag$. 
\end{proof}

We next show that the cosets of $K$ can be distinguished by time-translation-invariant measurements.

\begin{definition} [Dephasing map.]
    Let 
    $F$ have the following spectral decomposition: 
    \begin{align}
        F=\sum_{j=1}^J f_j \Pi_j, 
    \end{align}
    where $f_j\neq f_k$ if $j\neq k$, and $\Pi_j$ is the projection operator onto the eigenspace with eigenvalue $f_j$. 
    The Dephasing map $\mathcal{D}_F$ with respect to $F$ is defined by 
    \begin{align}
        \mathcal{D}_F(L):=\sum_{j=1}^J \Pi_j L \Pi_j\ \forall L. 
    \end{align}
\end{definition}

Time-translation covariance changes which part of the finite-group orbit is operationally accessible.
Any admissible measurement can depend only on information that survives the relevant time-translation constraint.
We therefore first identify the subgroup that remains invisible after collective energy dephasing.
For a state $\rho$, we define
\begin{align}
    K:=\left\{g\in G \ |\ U(g)\rho U(g)^\dag\in\overline{\left\{e^{-itH}\rho e^{itH} \ |\ t\in\mathbb{R}\right\}}\right\}. \label{eq:time_translation_subgroup}
\end{align}
Since the group representation commutes with the Hamiltonian, $K$ is a subgroup of $G$.
Moreover,
\begin{align}
    S\subset K.
\end{align}

\begin{lemma} \label{lem:dephased_state_symmetry_subgroup}
    Let $G$ be a finite group, $H$ be a Hamiltonian, $U$ be a projective unitary representation of $G$ satisfying $[U(h), H]=0$ for all $h\in G$, and $g\in G$. 
    Then, there exists some $m\in\mathbb{N}$ such that 
    \begin{align}
        \mathrm{Sym}\left(\mathcal{D}_{\mathcal{H}^{(m)}}(\rho^{\otimes m})\right) 
        =K, \label{eq:lem:dephased_state_symmetry_subgroup01}
    \end{align}
    where 
    \begin{align}
        H^{(m)}:=\sum_{j=1}^m I^{\otimes j-1}\otimes H\otimes I^{\otimes m-j}. 
    \end{align}
\end{lemma}

We give the proof of this lemma in Appendix~\ref{sec:dephased_state_symetry}

We can now characterize the effect of imposing the finite-group covariance condition.

\begin{theorem} \label{thm:TCGP_conversion_rate}
    Let $G$ be a finite group, $U$ and $U'$ be projective unitary representations of $G$ on $\mathcal{H}$ and $\mathcal{H}'$, respectively, $\rho\in\mathcal{S}(\mathcal{H})$, $\sigma\in\mathcal{S}(\mathcal{H}')$, $H\in\mathcal{L}^\mathrm{H}(\mathcal{H})$ and $H'\in\mathcal{L}^\mathrm{H}(\mathcal{H}')$ be Hamiltonians on the input and output systems, and $U$ and $U'$ commute with $H$ and $H'$, respectively. 
    Then, 
    \begin{align}
        R_{G\text{-}\mathrm{cov}, \mathrm{TCGP}}(\rho\to\sigma)=
        \begin{cases}
            R_{\mathrm{TCGP}}(\rho\to\sigma) & \text{ if the compatibility condition (Def.~\ref{def:group_time_covariance_compatibility}) holds,}\\
            0 & \textrm{ otherwise.}
        \end{cases}
    \end{align}
\end{theorem}

\begin{proof}
    We note that the result for the case where the compatibility condition does not hold is direct from Lemma~\ref{lem:group_time_covariance_compatibility}. 
    Thus, we consider the case where the compatibility condition holds in the following. 
    Since the class of $G$-covariant time-translation covariant Gibbs-preserving operations is smaller than the class of time-translation invariant Gibbs-preserving operations, we have 
    \begin{align}
        R_{G\text{-}\mathrm{cov}, \mathrm{TCGP}}(\rho\to\sigma) 
        \leq R_\mathrm{TCGP}(\rho\to\sigma) \label{eq:thm:TCGP_conversion_rate1}
    \end{align}
    if $R_\mathrm{TCGP}(\rho\to\sigma)=0$, we trivially get $R_{G\text{-}\mathrm{cov}, \mathrm{TCGP}}(\rho\to\sigma)=0$. 
    In the following, we assume that $R_\mathrm{TCGP}(\rho\to\sigma)>0$. 
    We take arbitrary $r\in (0, R_\mathrm{TCGP}(\rho\to\sigma))$. 
    Then, there exists some sequence of time-translation covariant and Gibbs-preserving operations $\mathcal{E}_n$ such that 
    \begin{align}
        \lim_{n\to\infty} \mathrm{T}\left(\mathcal{E}_n(\rho^{\otimes n}), \sigma^{\otimes \lfloor rn\rfloor}\right) =0
    \end{align}
    For any $k\in K$, there exists some sequence $(t_j)_{j\in\mathbb{N}}$ such that 
    \begin{align}
        \lim_{j\to\infty} e^{-it_j H}\rho e^{it_j H}=U(k)\rho U(k)^\dag. 
    \end{align}
    By the compatibility condition, we also have 
    \begin{align}
        \lim_{j\to\infty} e^{-it_j H'}\sigma e^{it_j H'}=U'(k)\sigma U'(k)^\dag. 
    \end{align}
    Since the map $\mathcal{E}_n$ is time-translation covariant, we have 
    \begin{align}
        &\mathrm{T}\left(\mathcal{E}_n\left(\left(e^{-it_j H}\rho e^{it_j H}\right)^{\otimes n}\right), \left(e^{-it_j H'}\sigma e^{it_j H'}\right)^{\otimes \lfloor rn\rfloor}\right) \nonumber\\
        =&\mathrm{T}\left(\left(e^{-it_j H'}\right)^{\otimes \lfloor rn\rfloor}
        \mathcal{E}_n(\rho^{\otimes n}) 
        \left(e^{it_j H'}\right)^{\otimes \lfloor rn\rfloor}, 
        \left(e^{-it_j H'}\right)^{\otimes \lfloor rn\rfloor} 
        \sigma^{\otimes \lfloor rn\rfloor}
        \left(e^{it_j H'}\right)^{\otimes \lfloor rn\rfloor}\right) \nonumber\\
        =&\mathrm{T}\left(\mathcal{E}_n(\rho^{\otimes n}), \sigma^{\otimes \lfloor rn\rfloor}\right). 
    \end{align}
    Taking the limit $j\to\infty$, we get 
    \begin{align}
        \mathrm{T}\left(\mathcal{E}_n\left(\mathcal{U}_k(\rho)^{\otimes n}\right), \mathcal{U}'_k(\sigma)^{\otimes \lfloor rn \rfloor}\right) 
        =\mathrm{T}\left(\mathcal{E}_n(\rho^{\otimes n}), \sigma^{\otimes \lfloor rn\rfloor}\right). 
    \end{align}
    We therefore get 
    \begin{align}
        \lim_{n\to\infty} \max_{k\in K} \mathrm{T}\left(\mathcal{E}_n\left(\mathcal{U}_k(\rho)^{\otimes n}\right), \mathcal{U}'_k(\sigma)^{\otimes \lfloor rn \rfloor}\right) 
        =0. 
    \end{align}

    By Lemma~\ref{lem:dephased_state_symmetry_subgroup}, there exists some $m\in\mathbb{N}$ such that the symmetry subgroup of $\mathcal{D}_{H^{(m)}}(\rho^{\otimes m})$ is $K$. 
    We can take a sub-POVM $\{M_x\}_{x\in G/K}$ such that 
    \begin{align}
        \lim_{n\to\infty} \mathrm{tr}((\mathcal{D}_{H^{(m)}}(\rho_x^{\otimes m}))^{\otimes n} M_x^{(n)})=1.  
    \end{align}
    Since $\mathcal{D}_{H^{(m)}}^\dag=\mathcal{D}_{H^{(m)}}$, we have 
    \begin{align}
        \lim_{n\to\infty} \mathrm{tr}\left(\rho_x^{\otimes mn} \mathcal{D}_{H^{(m)}}^{\otimes n}\left(M_x^{(n)}\right)\right)=1. 
    \end{align}
    We note that $\mathcal{D}_{H^{(m)}}^{\otimes n}(M_x^{(n)})$ commutes with $H^{(mn)}$. 
    By Lemma~\ref{lem:control_admissible_POVM_TCGP}, $\mathcal{D}_{H^{(m)}}^{\otimes n}\left(M_x^{(n)}\right)$ is a control-admissible sub-POVM in the resource theory of time-translation covariant Gibbs-preserving operations.

    Therefore, by Theorem~\ref{thm:rate_lower_bound_general}, we get 
    \begin{align}
        R_{G\text{-}\mathrm{cov}, \mathrm{TCGP}}(\rho\to\sigma)\geq r. 
    \end{align}
    Since this holds for all $r\in (0, R_\mathrm{TCGP}(\rho\to\sigma))$, we have 
    \begin{align}
        R_{G\text{-}\mathrm{cov}, \mathrm{TCGP}}(\rho\to\sigma)
        \geq R_\mathrm{TCGP}(\rho\to\sigma). \label{eq:thm:TCGP_conversion_rate2}
    \end{align}
    By Eqs.~\eqref{eq:thm:TCGP_conversion_rate1} and \eqref{eq:thm:TCGP_conversion_rate2}, we get 
    \begin{align}
        R_{G\text{-}\mathrm{cov}, \mathrm{TCGP}}(\rho\to\sigma)
        =R_\mathrm{TCGP}(\rho\to\sigma).
    \end{align}
\end{proof}

\subsection{Thermal operations}

We next consider thermal operations.
A thermal operation from a system with Hamiltonian $H$ to a system with Hamiltonian $H'$ is a channel of the form
\begin{align}
    \mathcal{E}(\rho)
    =\mathrm{tr}_\mathrm{E} \left(V\left(\rho\otimes\gamma_\mathrm{B}\right)V^\dag\right),
\end{align}
where $\gamma_\mathrm{B}$ is a Gibbs state with a fixed inverse temperature $\beta>0$ and $V$ is an energy-preserving unitary.
Every thermal operation is Gibbs preserving and time-translation covariant.
Therefore, the compatibility condition introduced in the previous subsection is necessary for a positive conversion rate.
It remains to show that, when the compatibility condition holds, the discrimination required in Theorem~\ref{thm:rate_lower_bound_general} can be implemented within thermal operations.

\subsubsection{Map-level symmetry constraints}

First, we check a sufficient condition for control-admissible sub-POVM in the resource theory of thermal operations.

\begin{lemma} \label{lem:control_admissible_POVM_TO}
    Let $\{M_x^{(m)}\}_{x\in X}$ be a sub-PVM satisfying $[M_x^{(m)}, H^{(m)}]=0$ and $\mathcal{F}$ be a thermal operation. 
    Then, $\mathcal{F}^\dag(M_x^{(m)})$ is a control-admissible sub-POVM in the resource theory of thermal operations. 
\end{lemma}

\begin{proof}
    Let $M_\perp:=I-\sum_{x\in X}M_x$.
    We first show that the channel obtained by measuring the PVM
    $\{M_x\}_{x\in X\cup\{\perp\}}$ and conditionally applying a thermal operation is itself a thermal operation.

    For each $x\in X\cup\{\perp\}$, let $\mathcal{E}_x:S\to S'$ be a thermal operation.
    We may write
    \begin{align}
        \mathcal{E}_x(L)=&
        \mathrm{tr}_{E_x}\left(
        V_x\left(L\otimes\gamma_{B_x}\right)V_x^\dag
        \right),
    \end{align}
    where $B_x$ is a heat bath in its Gibbs state $\gamma_{B_x}$ and $V_x$ is an energy-preserving unitary.

    Introduce the common bath
    \begin{align}
        &B:=\bigotimes_{x\in X\cup\{\perp\}}B_x, \\
        &\gamma_B:=\bigotimes_{x\in X\cup\{\perp\}}\gamma_{B_x}.
    \end{align}
    For each $x$, extend $V_x$ by the identity on all bath factors $B_y$ with $y\neq x$.
    We denote the resulting unitary by $\widetilde V_x$.
    The output spaces may be enlarged by trivial ancillary factors so that all $\widetilde V_x$ act between the same input and output Hilbert spaces.
    We define
    \begin{align}
        V:=&\sum_{x\in X\cup\{\perp\}}M_x\otimes\widetilde V_x.
    \end{align}
    Since $\{M_x\}_{x\in X\cup\{\perp\}}$ is a PVM, $V$ is unitary.
    Moreover, $[M_x,H_A]=0$ for every $x$, and each $\widetilde V_x$ is energy preserving.
    Hence $V$ is energy preserving with respect to the total Hamiltonian.

    After tracing out the measured system $A$ and all bath output systems, the resulting channel is
    \begin{align}
        \mathcal{C}_M(L)=&
        \sum_{x\in X\cup\{\perp\}}
        \mathcal{E}_x\left(
        \mathrm{tr}_A\left(
        \left(M_x\otimes I_S\right)L
        \right)
        \right).
    \end{align}
    Therefore, $\mathcal{C}_M$ is a thermal operation.

    Now let $\mathcal{F}$ be the thermal operation appearing in the statement.
    Since thermal operations are closed under tensor products and composition, $\mathcal{C}_M\circ \left(\mathcal{F}\otimes\mathrm{id}_S\right)$ is a thermal operation.
    Its action is
    \begin{align}
        \mathcal{C}_M\circ
        \left(\mathcal{F}\otimes\mathrm{id}_S\right)(L)
        =\sum_{x\in X\cup\{\perp\}}
        \mathcal{E}_x\left(
        \mathrm{tr}_A\left(
        \left(\mathcal{F}^\dag(M_x)\otimes I_S\right)L
        \right)
        \right).
    \end{align}
    Hence $\{\mathcal{F}^\dag(M_x)\}_{x\in X}$ is a control-admissible sub-POVM.
\end{proof}

Next, we construct a control-admissible sub-POVM that asymptotically identifies the coset $gK$. 
Since we use this idea in the proof of the next section, we describe the lemma in a general form.

\begin{lemma} \label{lem:TO_state_discrimination_general}
    Let $G$ be a finite group, $S$ be a symmetry subgroup of $\rho$, $\tau$ be a state, $\ket{\phi}$ be a pure state, $\ket{\phi_g}:=U'(g)\ket{\phi}$ are mutually orthogonal and 
    \begin{align}
        \bra{\phi}\tau\ket{\phi}>0. 
    \end{align}
    Then, there exists a sequence of sub-PVMs $(\{R_x^{(k)}\}_{x\in G/S})_{k\in\mathbb{N}}$ and a constant $\alpha>0$ such that for any $g\in G$ and sufficiently large $k\in\mathbb{N}$, 
    \begin{align}
        \mathrm{tr}\left(R_{gS}^{(k)}\left(\mathcal{U}_g(\rho)^{\otimes k}\otimes \mathcal{U}'_g(\tau)\right)^{\otimes k}\right)\geq 1-e^{-\alpha k}. 
    \end{align}
    Moreover, for any operator $A$ commuting with $\mathcal{U}_g(\rho)$ for all $g\in G$, $R_{gS}^{(k)}$ commutes with $(A^{\otimes k}\otimes I)^{\otimes k}$. 
\end{lemma}

\begin{proof}
    First, we construct a projection operator $\Pi^{(k)}$ such that 
    \begin{align}
        &\mathrm{tr}\left(\Pi^{(k)}\rho^{\otimes k}\right)
        \geq 1-e^{-\lambda k}, \nonumber\\
        &\mathrm{tr}\left(\Pi^{(k)}\mathcal{U}_g(\rho)^{\otimes k}\right)
        \leq e^{-\lambda k} \text{ if }g\not\in S. 
    \end{align}
    with some $\lambda>0$. 
    We define 
    \begin{align}
        O^{\times k}:=\frac{1}{k}\sum_{j=1}^k I^{\otimes j-1}\otimes \rho \otimes I^{\otimes k-j}. 
    \end{align}
    We note that for any $g\in G$, 
    \begin{align}
        &\mathrm{tr}\left(\mathcal{U}_g(\rho)^{\otimes k} O^{\times k}\right) 
        =\mathrm{tr}\left(\mathcal{U}_g(\rho)\rho\right), \\
        &\mathrm{tr}\left(\rho^2\right) 
        =\mathrm{tr}\left(\mathcal{U}_g(\rho)\rho\right)+\frac{1}{2}\|\rho-\mathcal{U}_g(\rho)\|_2^2 
        >\mathrm{tr}\left(\mathcal{U}_g(\rho)\rho\right) \text{ if }g\not\in S. 
    \end{align}
    We define the projection operator onto the eigenspace of $O^{\times k}$ with eigenvalue no smaller than $c$: 
    \begin{align}
        \Pi^{(k)}:=\bm{1}_{\{O^{\times k}\geq c\}}, 
    \end{align}
    where 
    \begin{align}
        &a:=\mathrm{tr}\left(\rho^2\right), \\
        &b:=\max_{g\in G-S} \mathrm{tr}\left(\mathcal{U}_g(\rho)\rho\right), \\
        &c:=\frac{a+b}{2}. 
    \end{align}
    By using Hoeffding's inequality, 
    \begin{align}
        \mathrm{tr}\left(\Pi^{(k)}\rho^{\otimes k}\right)
        \geq 1-\exp\left(-\frac{2\left(k\Delta\right)^2}{k\|\rho\|_\infty^2}\right) 
        =1-\exp\left(-\frac{2\Delta^2}{\|\rho\|_\infty^2}k\right) 
        =1-e^{-\lambda k}, \label{eq:lem:TO_state_discrimination_general1}
    \end{align}
    where 
    \begin{align}
        &\Delta:=\frac{b-a}{2}, \\
        &\lambda:=\frac{2\Delta^2}{\|\rho\|_\infty^2}. 
    \end{align}
    Similarly, we get 
    \begin{align}
        \mathrm{tr}\left(\Pi^{(k)}\mathcal{U}_g(\rho)^{\otimes k}\right)
        \leq e^{-\lambda k} \text{ if }g\not\in S. \label{eq:lem:TO_state_discrimination_general2}
    \end{align}
    For each $g\in G$, we define
    \begin{align}
        \Pi_g^{(k)}:=U(g)^{\otimes k}\Pi^{k} U(g)^{\dag\otimes k}.
    \end{align}
    We note that the projectors $\{\Pi_x\}_{x\in G/S}$ is not necessarily mutually orthogonal.

    Next, we construct a mutually orthogonal projectors, and prove several properties. 
    We define 
    \begin{align}
        Q_{gS}^{(k)}:=\Pi_{gS}^{(k)}\otimes\left(\sum_{h\in gS}\ket{\phi_h}\bra{\phi_h}\right). 
    \end{align}
    Since the states $\{\ket{\phi_g}\}_{g\in G}$ are mutually orthogonal, the projectors $\{Q_x^{(k)}\}_{x\in G/S}$ are mutually orthogonal. 
    We also define
    \begin{align}
        Q_\perp^{(k)}:=I-\sum_{x\in G/S} Q_x^{(k)}. 
    \end{align}
    Then, we have 
    \begin{align}
        \mathrm{tr}\left(Q_x^{(k)}\left(\mathcal{U}_g(\rho)^{\otimes k}\otimes\mathcal{U}'_g(\tau)\right)\right) 
        =&\mathrm{tr}\left(\left[\Pi_x^{(k)}\otimes \left(\sum_{h\in x}\ket{\phi_h}\bra{\phi_h}\right)\right]\left(\mathcal{U}_g(\rho)^{\otimes k}\otimes\mathcal{U}'_g(\tau)\right)\right) \nonumber\\
        =&\mathrm{tr}(\Pi_x^{(k)}\mathcal{U}_g(\rho)^{\otimes k})\sum_{h\in x} \mathrm{tr}(\ket{\phi_h}\bra{\phi_h}\mathcal{U}'_g(\tau)). \label{eq:lem:TO_state_discrimination_general3}
    \end{align}
    By Eqs.~\eqref{eq:lem:TO_state_discrimination_general1} and \eqref{eq:lem:TO_state_discrimination_general3}, we get 
    \begin{align}  
        \mathrm{tr}\left(Q_{gS}^{(k)}\left(\mathcal{U}_g(\rho)^{\otimes k}\otimes\mathcal{U}'_g(\tau)\right)\right) 
        =&\mathrm{tr}\left(\Pi_{gS}^{(k)}\mathcal{U}_g(\rho)^{\otimes k}\right)\sum_{h\in gS} \mathrm{tr}(\ket{\phi_h}\bra{\phi_h}\mathcal{U}'_g(\tau)) \nonumber\\
        =&\mathrm{tr}\left(\Pi_{gS}^{(k)}\mathcal{U}_g(\rho)^{\otimes k}\right) \sum_{h\in S} \braket{\phi_h | \tau | \phi_h} \nonumber\\
        \geq &\left(1-e^{-\lambda k}\right)\sum_{h\in S} \braket{\phi_h | \tau | \phi_h}, \label{eq:lem:TO_state_discrimination_general4}
    \end{align}
    If $x\in G/S-\{gS\}$, we have 
    \begin{align}
        \mathrm{tr}\left(Q_x^{(k)}\left(\mathcal{U}_g(\rho)^{\otimes k}\otimes\mathcal{U}'_g(\tau)\right)\right) 
        =\mathrm{tr}(\Pi_x^{(k)}\mathcal{U}_g(\rho)^{\otimes k})\sum_{h\in x} \mathrm{tr}(\ket{\phi_h}\bra{\phi_h}\mathcal{U}'_g(\tau)) 
        \leq e^{-\lambda k} \sum_{h\in x} \mathrm{tr}(\ket{\phi_h}\bra{\phi_h}\mathcal{U}'_g(\tau)), 
    \end{align}
    which implies 
        \begin{align}
        \sum_{x\in G/S-\{gS\}} \mathrm{tr}\left(Q_x^{(k)}\left(\mathcal{U}_g(\rho)^{\otimes k}\otimes\mathcal{U}'_g(\tau)\right)\right) 
        \leq &e^{-\lambda k} \sum_{x\in G/S-\{gS\}} \sum_{h\in x} \mathrm{tr}(\ket{\phi_h}\bra{\phi_h}\mathcal{U}'_g(\tau)) \nonumber\\
        =& e^{-\lambda k} \sum_{h\in G-gS} \mathrm{tr}(\ket{\phi_h}\bra{\phi_h}\mathcal{U}'_g(\tau)) \nonumber\\
        =& e^{-\lambda k}\mathrm{tr}\left(\left(\sum_{h\in G-gS}\ket{\phi_h}\bra{\phi_h}\right)\mathcal{U}'_g(\tau)\right) \nonumber\\
        \leq & e^{-\lambda k}\left\|\sum_{h\in G-gS}\ket{\phi_h}\bra{\phi_h}\right\|_\infty \left\|\mathcal{U}'_g(\tau)\right\|_1 \nonumber\\
        =&e^{-\lambda k}. \label{eq:lem:TO_state_discrimination_general5}
    \end{align}
    Eq.~\eqref{eq:lem:TO_state_discrimination_general5} implies 
    \begin{align}
        \mathrm{tr} \left((Q_{gS}^{(k)}+Q_\perp^{(k)})\left(\mathcal{U}_g(\rho)^{\otimes k}\otimes \mathcal{U}'(\tau)\right)\right) 
        =& \mathrm{tr} \left(\left(I-\sum_{x\in G/S-\{gS\}} Q_x^{(k)}\right) \left(\mathcal{U}_g(\rho)^{\otimes k}\otimes \mathcal{U}'(\tau)\right)\right) \nonumber\\
        =& 1-\sum_{x\in G/S-\{gS\}} \mathrm{tr} \left(Q_x^{(k)} \left(\mathcal{U}_g(\rho)^{\otimes k}\otimes \mathcal{U}'(\tau)\right)\right) \nonumber\\
        \geq & 1-e^{-\lambda k}. \label{eq:lem:TO_state_discrimination_general6}
    \end{align}
    Eq.~\eqref{eq:lem:TO_state_discrimination_general4} implies 
    \begin{align}
        \mathrm{tr} \left(Q_\perp^{(k)} \left(\mathcal{U}_g(\rho)^{\otimes k}\otimes\mathcal{U}'_g(\tau)\right)\right) 
        \leq &\mathrm{tr} \left(\left(I-Q_{gS}^{(k)}\right) \left(\mathcal{U}_g(\rho)^{\otimes k}\otimes \mathcal{U}'_g(\tau)\right)\right) \nonumber\\
        =&1-\mathrm{tr} \left(Q_{gS}^{(k)}\left(\mathcal{U}_g(\rho)^{\otimes k}\otimes \mathcal{U}'_g(\tau)\right)\right) \nonumber\\
        \leq &1-\left(1-e^{-\lambda k}\right)\sum_{h\in S} \braket{\phi_h | \tau | \phi_h}. \label{eq:lem:TO_state_discrimination_general7}
    \end{align}

    Finally, we construct a mutually orthogonal sub-PVM with vanishingly small error probability. 
    For each $x\in G/S$, we define 
    \begin{align}
        R_x^{(k, l)}:=\sum_{(x_1, ..., x_l)\in\{x, \perp\}^l-\{\perp\}^l} Q_{x_1}^{(k)}\otimes\cdots\otimes Q_{x_l}^{(k)}, 
    \end{align}
    and we also define 
    \begin{align}
        R_\perp^{(k, l)}:=I-\sum_{x\in G/S} R_x^{(k, l)}. 
    \end{align}
    Then, $\{R_x^{(k, l)}\}_{x\in G/S\cup\{\perp\}}$ is a PVM. 
    By Eqs.~\eqref{eq:lem:TO_state_discrimination_general6} and \eqref{eq:lem:TO_state_discrimination_general7}, we have 
    \begin{align}
        \mathrm{tr}\left(R_{gS}^{(k, l)}\left(\mathcal{U}_g(\rho)^{\otimes k}\otimes\mathcal{U}'_g(\tau)\right)^{\otimes l}\right) 
        \geq &\left(1-e^{-\lambda k}\right)^l-\left[1-\left(1-e^{-\lambda k}\right)\sum_{h\in S} \braket{\phi_h | \tau | \phi_h}\right]^l \nonumber\\
        \geq &1-le^{-\lambda k}-\left[\left(1-\sum_{h\in S} \braket{\phi_h | \tau | \phi_h}\right)+e^{-\lambda k}\sum_{h\in S} \braket{\phi_h | \tau | \phi_h}\right]^l. 
    \end{align}
    If we set $l=k$, we have 
    \begin{align}
        \mathrm{tr}\left(R_{gS}^{(k, k)}\left(\mathcal{U}_g(\rho)^{\otimes k}\otimes\mathcal{U}'_g(\tau)\right)^{\otimes k}\right) 
        \geq 1-ke^{-\lambda k}-\left[\left(1-\sum_{h\in S} \braket{\phi_h | \tau | \phi_h}\right)+e^{-\lambda k}\sum_{h\in S} \braket{\phi_h | \tau | \phi_h}\right]^k. 
    \end{align}
    We take arbitrary $\alpha<\min\{\lambda, -\log(1-\sum_{h\in S} \braket{\phi_h | \tau | \phi_h})\}$. 
    Since $\ket{\phi}$ satisfies $\braket{\phi | \tau | \phi}>0$, we can take $\alpha>0$. 
    Then, we can take sufficiently large $k\in\mathbb{N}$ such that 
    \begin{align}
        &ke^{-\lambda k}\leq \frac{1}{2}e^{-\alpha k}, \\
        &\left[\left(1-\sum_{h\in S} \braket{\phi_h | \tau | \phi_h}\right)+e^{-\lambda k}\sum_{h\in S} \braket{\phi_h | \tau | \phi_h}\right]^k 
        \leq \frac{1}{2}e^{-\alpha k}. 
    \end{align}
    For this sufficiently large $k$, we have 
    \begin{align}
        \mathrm{tr}\left(R_{gS}^{(k, k)}\left(\mathcal{U}_g(\rho)^{\otimes k}\otimes\mathcal{U}'_g(\tau)\right)^{\otimes k}\right) 
        \geq 1-e^{-\alpha k}. 
    \end{align}
    For any operator $A$ commuting with $\rho$, $O^{\times k}$ commutes with $A^{\otimes k}$, which implies $\Pi^{(k)}$ commutes with $A^{\otimes k}$. 
    Then, $\Pi_{gS}^{(k)}$ commutes with $\mathcal{U}_g(A)^{\otimes k}$. 
    By the construction of $Q_{gS}^{(k)}$, $Q_{gS}^{(k)}$ commutes with $\mathcal{U}_g(A)^{\otimes k}\otimes I$. 
    By the construction of $R_{gS}^{(k, k)}$, $R_{gS}^{(k, k)}$ commutes with $(\mathcal{U}_g(A)^{\otimes k}\otimes I)^{\otimes k}$. 
\end{proof}

By considering the case where we substitute $\rho$, $S$, $\ket{\phi}$, $\tau$ by $\mathcal{D}_{H^{\times m}}(\rho^{\otimes m})$, $K$, $\ket{\xi_e}$ and $\sum_{g\in G} \ket{\xi_g}\bra{\xi_g}/|G|$, respectively, we get the following lemma.

\begin{lemma} \label{lem:TO_state_discrimination}
    Let $G$ be a finite group, $K$ be a subgroup of $G$ defined by Eq.~\eqref{eq:time_translation_subgroup}. 
    Then, there exists a sequnce of control-admissible sub-POVMs $(\{\widetilde{R}_x^{(k)}\}_{x\in G/K})_{k\in\mathbb{N}}$ in the resource theory of thermal operations such that for sufficiently large $k\in\mathbb{N}$, 
    \begin{align}
        \mathrm{tr}\left(\mathcal{U}_g(\rho)^{\otimes mk^2} \widetilde{R}_{gK}^{(k)}\right)\geq 1-e^{-\alpha k}. 
    \end{align}
\end{lemma}

\begin{proof}
    By Lemma~\ref{lem:dephased_state_symmetry_subgroup}, there exists some $m\in\mathbb{N}$ such that the symmetry subgroup of $\mathcal{D}_{H^{(m)}}(\rho^{\otimes m})$ is $K$. 

    We introduce an ancillary system with the trivial Hamiltonian $H=0$ and an orthonormal basis $\{|\xi_g\rangle\}_{g\in G}$ carrying the left regular representation.
    The Gibbs state of this system is given by 
    \begin{align}
        \zeta:=\frac{1}{|G|}\sum_{g\in G}|\xi_g\rangle\langle\xi_g|.
    \end{align}
    This state satisfies 
    \begin{align}
        \braket{\xi_e | \zeta | \xi_e}=\frac{1}{|G|}>0. 
    \end{align}

    Therefore, by Lemma~\ref{lem:TO_state_discrimination_general}, there exists a sequence of PVMs $(\{R_x^{(k)}\}_{x\in G/K})_{k\in\mathbb{N}}$ and a constant $\alpha>0$ such that for any $g\in G$ and sufficiently large $k\in\mathbb{N}$, 
    \begin{align}
        \mathrm{tr}\left(R_{gK}^{(k)}\left(\mathcal{U}_g^{\otimes m}(\mathcal{D}_{H^{\times m}}(\rho^{\otimes m}))^{\otimes k}\otimes \zeta\right)^{\otimes k}\right)\geq 1-e^{-\alpha k}. 
    \end{align}
    Since $\mathcal{U}_g^{\otimes m}$ commutes with $\mathcal{D}_{H^{\times m}}$ and $\mathcal{D}_{H^{\times m}}^\dag=\mathcal{D}_{H^{\times m}}$, we get 
    \begin{align}
        \mathrm{tr}\left(R_{gK}^{(k)}\left(\mathcal{U}_g^{\otimes m}(\mathcal{D}_{H^{\times m}}(\rho^{\otimes m}))^{\otimes k}\otimes \zeta\right)^{\otimes k}\right) 
        =&\mathrm{tr}\left(R_{gK}^{(k)}\left(\mathcal{D}_{H^{\times m}}(\mathcal{U}_g^{\otimes m}(\rho^{\otimes m}))^{\otimes k}\otimes \zeta\right)^{\otimes k}\right) \nonumber\\
        =&\mathrm{tr}\left(R_{gK}^{(k)}\left(\left(\mathcal{D}_{H^{\times m}}\right)^{\otimes k^2}(\mathcal{U}_g(\rho)^{\otimes mk^2})\otimes \zeta^{\otimes k}\right)\right). 
    \end{align}
    We define a CPTP map by 
    \begin{align}
        \mathcal{E}(L):=\left(\mathcal{D}_{H^{\times m}}\right)^{\otimes k^2}(L)\otimes \zeta^{\otimes k}. 
    \end{align}
    Then, $\mathcal{E}$ is a thermal operation and we have 
    \begin{align}
        \mathrm{tr}\left(R_{gK}^{(k)}\left(\mathcal{U}_g^{\otimes m}(\mathcal{D}_{H^{(m)}}(\rho^{\otimes m}))^{\otimes k}\otimes \zeta\right)^{\otimes k}\right) 
        =\mathrm{tr}\left(R_{gK}^{(k)}\mathcal{E}(\mathcal{U}_g(\rho)^{\otimes mk^2})\right) 
        =\mathrm{tr}\left(\mathcal{E}^\dag\left(R_{gK}^{(k)}\right) \mathcal{U}_g(\rho)^{\otimes mk^2}\right). 
    \end{align}
    Since for any $t\in\mathbb{R}$, $(e^{-itH})^{\otimes m}$ commutes with $\mathcal{D}_{H^{\times m}}(\rho^{\otimes m})$, Lemma~\ref{lem:TO_state_discrimination_general} implies that $R_{gK}^{(k)}$ commutes with $[(e^{-itH})^{\otimes mk}\otimes I]^{\otimes k}$, which means that $R_{gK}^{(k)}$ is time-translation invariant. 
    Therefore, by lemma~\ref{lem:control_admissible_POVM_TO}, $\widetilde{R}_{gK}^{(k)}:=\mathcal{E}^\dag(R_{gK}^{(k)})$ is a control-admissible sub-POVM in the resource theory of thermal operations. 
\end{proof}

We can now characterize the asymptotic conversion rate under $G$-covariant thermal operations. 
We denote the asymptotic conversion rate under thermal operations by $R_\mathrm{TO}(\rho\to\sigma)$, and the corresponding rate under thermal operations that are additionally $G$-covariant by $R_{G\text{-}\mathrm{cov}, \mathrm{TO}}(\rho\to\sigma)$.

\begin{theorem} \label{thm:covariant_TO_conversion_rate}
    Let $G$ be a finite group, $U$ and $U'$ be projective unitary representations of $G$ on $\mathcal{H}$ and $\mathcal{H}'$, respectively, $\rho\in\mathcal{S}(\mathcal{H})$, $\sigma\in\mathcal{S}(\mathcal{H}')$, $H\in\mathcal{L}^\mathrm{H}(\mathcal{H})$ and $H'\in\mathcal{L}^\mathrm{H}(\mathcal{H}')$ be Hamiltonians on the input and output systems, and $[H, U(g)]=0$ and $[H', U'(g)]=0$ for all $g\in G$. 
    Then, 
    \begin{align}
        R_{G\text{-}\mathrm{cov}, \mathrm{TO}}(\rho\to\sigma)=
        \begin{cases}
            R_{\mathrm{TO}}(\rho\to\sigma) & \text{ if the compatibility condition (Def.~\ref{def:group_time_covariance_compatibility}) holds,}\\
            0 & \textrm{ otherwise.}
        \end{cases}
    \end{align}
\end{theorem}

\begin{proof}
    We can prove this theorem in a similar way to the proof of Theorem~\ref{thm:TCGP_conversion_rate}. 
    If the compatibility condition does not hold, $R_{G\text{-}\mathrm{cov}, \mathrm{TO}}(\rho\to\sigma)=0$ by Lemma~\ref{lem:group_time_covariance_compatibility}. 
    In the following, we consider the case where the compatibility condition holds. 
    In the same way as in the proof of Theorem~\ref{thm:TCGP_conversion_rate}, we have 
    \begin{align}
        R_{G\text{-}\mathrm{cov}, \mathrm{TO}}(\rho\to\sigma) 
        \leq R_{\mathrm{TO}}(\rho\to\sigma). \label{eq:thm:covariant_TO_conversion_rate1}
    \end{align}
    We take $r\in (0, R_{\mathrm{TO}}(\rho\to\sigma))$. 
    Then, there exists some sequence of thermal operations $\mathcal{E}_n$ such that 
    \begin{align}
        \lim_{n\to\infty} \mathrm{T}\left(\mathcal{E}_n(\rho^{\otimes n}), \sigma^{\otimes \lfloor rn\rfloor}\right)=0. 
    \end{align}
    Since thermal operations are time-translation covariant, by the same argument in the proof of Theorem~\ref{thm:TCGP_conversion_rate}, we get 
    \begin{align}
        \lim_{n\to\infty} \max_{k\in K} \mathrm{T}\left(\mathcal{E}_n\left(\mathcal{U}_k(\rho)^{\otimes n}\right), \mathcal{U}_k(\sigma)^{\otimes \lfloor rn \rfloor}\right) 
        =0. 
    \end{align}
    By Lemma~\ref{lem:TO_state_discrimination}, there exists some sequence of control-admissible sub-POVMs $\{R_x^{(l)}\}_{x\in G/K}$ such that 
    \begin{align}
        \lim_{l\to\infty} \mathrm{tr}\left(R_{gK}^{(l)}\mathcal{U}_g(\rho)^{\otimes l}\right)=1. 
    \end{align}
    Therefore, by Theorem~\ref{thm:rate_lower_bound_general}, we have 
    \begin{align}
        R_{G\text{-}\mathrm{cov}, \mathrm{TO}}(\rho\to\sigma) 
        \geq r. 
    \end{align}
    Since this holds for all $r\in (0, R_\mathrm{TO}(\rho\to\sigma))$, we get 
    \begin{align}
        R_{G\text{-}\mathrm{cov}, \mathrm{TO}}(\rho\to\sigma)
        \geq R_\mathrm{TO}(\rho\to\sigma). \label{eq:thm:covariant_TO_conversion_rate2}
    \end{align}
    By Eqs.~\eqref{eq:thm:covariant_TO_conversion_rate1} and \eqref{eq:thm:covariant_TO_conversion_rate2}, we get 
    \begin{align}
        R_{G\text{-}\mathrm{cov}, \mathrm{TO}}(\rho\to\sigma)
        =R_\mathrm{TO}(\rho\to\sigma). 
    \end{align}
\end{proof}

\subsubsection{Primitive-level symmetry constraints}

We next consider the class of free operations that we call $G$-invariant thermal operations, where symmetry constraints are directly imposed on the time-translation-invariant unitary appearing in the dilation of a thermal operation, i.e., a map $\mathcal{E}$ is a $G$-invariant thermal operation if it can be written as 
\begin{align}
    \mathcal{E}(L)=\mathrm{tr}_\mathrm{E}\left(V(L\otimes \gamma_\mathrm{B})V^\dag\right)
\end{align}
with a time-translation invariant and $G$-invariant unitary $V$ and a Gibbs state $\gamma_\mathrm{B}$.

We show that this stronger requirement does not further restrict the class of $G$-covariant thermal operations.

\begin{lemma} \label{lem:invariant_covariant_TO_equivalence}
    Let $G$ be a finite group. 
    Then, the class of $G$-invariant thermal operations is equivalent to the class of $G$-covariant thermal operations. 
\end{lemma}

\begin{proof}
    The class of $G$-invariant thermal operations is trivially included in the class of $G$-covariant thermal operations, because every step of $G$-invariant thermal operation is a $G$-covariant thermal operation. 
    In the following, we show that every $G$-covariant thermal operation is a $G$-invariant thermal operation. 
    We take an arbitrary $G$-covariant thermal operation $\mathcal{E}$. 
    Since $\mathcal{E}$ is a thermal operation, it can be expressed as 
    \begin{align}
        \mathcal{E}(L)=\mathrm{tr}_\mathrm{E}(V(L\otimes\gamma_\mathrm{B})V^\dag) \label{eq:lem:invariant_covariant_TO_equivalence01}
    \end{align}
    with time-translation-invariant unitary $V$ and Gibbs state $\gamma_\mathrm{B}$. 
    We introduce an ancillary system $\mathrm{A}$ with Hamiltonian $H_\mathrm{A}=0$ and an orthonormal basis $\{\ket{\xi_g}\}_{g\in G}$ carrying the left regular representation $L$ of $G$ satisfying, $L(h)\ket{\xi_g}=\ket{\xi_{hg}}$. 
    The Gibbs state of this system is 
    \begin{align}
        \zeta_\mathrm{A}
        =\frac{1}{|G|}\sum_{g\in G} \ket{\xi_g}\bra{\xi_g}. \label{eq:lem:invariant_covariant_TO_equivalence02}
    \end{align}
    We define $\widetilde{V}$ by 
    \begin{align}
        \widetilde{V}:=\sum_{g\in G} U(g)VU(g)^\dag\otimes \ket{\xi_g}\bra{\xi_g}. \label{eq:lem:invariant_covariant_TO_equivalence03}
    \end{align}
    Since $U(g)VU(g)^\dag$ is unitary for all $g\in G$ and the projectors $\{\ket{\xi_g}\bra{\xi_g}\}_{g\in G}$ are mutually orthogonal, $\widetilde{V}$ is unitary. 
    We can directly confirm that $\widetilde{V}$ is $G$-invariant. 
    Moreover, since $V$, $U(g)$, and $\ket{\xi_g}\bra{\xi_g}$ are all commute with the Hamiltonian of each system, $\widetilde{V}$ is time-translation invariant. 
    By Eqs.~\eqref{eq:lem:invariant_covariant_TO_equivalence02} and \eqref{eq:lem:invariant_covariant_TO_equivalence03}, we have 
    \begin{align}
        \widetilde{V} \left(L\otimes\gamma_\mathrm{B}\otimes\zeta_\mathrm{A}\right)\widetilde{V}^\dag 
        =&\frac{1}{|G|}\sum_{g\in G} U(g)VU(g)^\dag (L\otimes\gamma_\mathrm{B}) U(g)V^\dag U(g)^\dag\otimes\ket{\xi_g}\bra{\xi_g} \\
        =&\frac{1}{|G|}\sum_{g\in G} \mathcal{U}_{\mathrm{S}', g}\otimes\mathcal{U}_{\mathrm{E}, g}(V(\mathcal{U}_{\mathrm{S}, g^{-1}}(L)\otimes\gamma_\mathrm{B})V^\dag)\otimes\ket{\xi_g}\bra{\xi_g}, \label{eq:lem:invariant_covariant_TO_equivalence04}
    \end{align}
    which implies 
    \begin{align}
        \mathrm{tr}_{\mathrm{EA}} \left(\widetilde{V}\left(L\otimes\gamma_\mathrm{B}\otimes\zeta_\mathrm{A}\right)\widetilde{V}^\dag\right) 
        =&\mathrm{tr}_\mathrm{E}\left(\frac{1}{|G|}\sum_{g\in G} \mathcal{U}_{\mathrm{S}', g}\otimes\mathcal{U}_{\mathrm{E}, g}(V(\mathcal{U}_{\mathrm{S}, g^{-1}}(L)\otimes\gamma_\mathrm{B})V^\dag)\right) \nonumber\\
        =&\frac{1}{|G|}\sum_{g\in G} \mathcal{U}_{\mathrm{S}', g} \mathrm{tr}_\mathrm{E}\left((V(\mathcal{U}_{\mathrm{S}, g^{-1}}(L)\otimes\gamma_\mathrm{B})V^\dag)\right) \nonumber\\
        =&\frac{1}{|G|}\sum_{g\in G} \mathcal{U}_{\mathrm{S}', g}\circ\mathcal{E}\circ\mathcal{U}_{\mathrm{S}, g^{-1}}(L)  \nonumber\\
        =&\mathcal{E}(L), \label{eq:lem:invariant_covariant_TO_equivalence05}
    \end{align}
    where we used the $G$-covariance of $\mathcal{E}$ in the final equality. 
    Equation~\eqref{eq:lem:invariant_covariant_TO_equivalence05} implies that $\mathcal{E}$ is a $G$-invariant thermal operation. 
\end{proof}

As an immediate consequence, imposing the symmetry constraint at the primitive level does not further reduce the asymptotic conversion rate.

\begin{theorem} \label{thm:TO_rate_primitive_level}
    The asymptotic conversion rate under thermal operations implemented by $G$-symmetric energy-preserving unitaries is identical to the rate under $G$-covariant thermal operations, i.e., 
    \begin{align}
        R_{G\text{-}\mathrm{inv}, \mathrm{TO}}(\rho\to\sigma)
        = \begin{cases}
            R_\mathrm{TO}(\rho\to\sigma) & \text{if the compatibility condition (Def.~\ref{def:group_time_covariance_compatibility}) holds,}, \\
            0 & \text{otherwise}.
        \end{cases}
    \end{align}
\end{theorem}

\subsection{Symmetry-protected ergotropy}

We finally consider work extraction by unitary operations under a finite-group symmetry constraint.
For a state $\rho$ with Hamiltonian $H$, we define
\begin{align}
    W^{(n)}(\rho)
    := \max_V\left\{\mathrm{tr}\left(\rho^{\otimes n} H^{(n)}\right)-\mathrm{tr}\left(V\rho^{\otimes n}V^\dag H^{(n)}\right)\right\},
\end{align}
where the maximization is taken over all unitaries $V$ on $n$ copies and
\begin{align}
    H^{\times n}:=\sum_{j=1}^n I^{\otimes j-1}\otimes H\otimes I^{\otimes n-j}. 
\end{align}
We define the asymptotic ergotropy rate by 
\begin{align}
    \overline{W}^{(\infty)}(\rho):=\limsup_{n\to\infty} \overline{W}^{(n)}(\rho) 
\end{align}
with 
\begin{align}
    \overline{W}^{(n)}(\rho):=\frac{1}{n}W^{(n)}(\rho). 
\end{align}

We similarly define $W_{G\text{-}\mathrm{inv}}^{(n)}(\rho)$ by restricting the optimization to $G$-invariant unitaries satisfying $[V, U(g)^{\otimes n}]=0$ for all $g\in G$, and define symmetry-protected ergotropy rate 
\begin{align}
    W_{G\text{-}\mathrm{inv}}^{(\infty)}(\rho):=\limsup_{n\to\infty} \overline{W}_{G\text{-}\mathrm{inv}}^{(n)}(\rho). 
\end{align}
with 
\begin{align}
    \overline{W}_{G\text{-}\mathrm{inv}}^{(n)}(\rho):=\frac{1}{n}W_{G\text{-}\mathrm{inv}}^{(n)}(\rho). 
\end{align}
We emphasize that the unitaries considered here are not required to commute with the Hamiltonian.

\begin{definition}[Complete passivity]
    \label{def:complete_passivity}
    Let $G$ be a group and $\rho$ be a state. 
    A state $\rho$ is called completely passive if
    \begin{align}
        W^{(n)}(\rho)=0\ \forall n\in\mathbb{N}.
    \end{align}
    Similarly, we say that $\rho$ is $G$-completely passive if
    \begin{align}
        W_{G\text{-inv}}^{(n)}(\rho)=0\ \forall n\in\mathbb{N},
    \end{align}
    where the work extraction in $W_G^{(n)}(\rho)$ is restricted to $G$-invariant unitaries.
\end{definition}

\begin{theorem} \label{thm:symmetry_protected_ergotropy}
    Let $G$ be a finite group, $U$ be a projective unitary representation of $G$, and $H$ be a $G$-invariant Hamiltonian. 
    Then,
    \begin{align}
        \overline{W}_{G\text{-}\mathrm{inv}}^{(\infty)}(\rho)=\overline{W}^{(\infty)}(\rho). 
    \end{align}
\end{theorem}

To consider ergotropy rate, we prepare a lemma, which is similar to Lemma~\ref{lem:TO_state_discrimination} that we used in the proof of the previous section about thermal operations. 
We note that we cannot freely use Gibbs states this time, unlike Lemma~\ref{lem:TO_state_discrimination}.

To construct a $G$-invariant work-extraction unitary, we use a fine-grained version of the PVM constructed in Lemma~\ref{lem:TO_state_discrimination_general}.
The additional structure needed here is that the outcomes are labeled by individual group elements, while their sums over symmetry cosets retain the discrimination property of Lemma~\ref{lem:TO_state_discrimination_general}.

\begin{lemma}
    \label{lem:ergotropy_fine_grained_PVM}
    Let $G$ be a finite group, $U$ be a faithful projective unitary representation of $G$, and $\rho$ be a state. 
    Suppose that there exist $m\in\mathbb{N}$ and a vector $\ket{\phi}$ such that the vectors $\ket{\phi_g}:=U(g)^{\otimes m}\ket{\phi}$ with $g\in G$ are mutually orthogonal and
    \begin{align}
        \bra{\phi}\rho^{\otimes m}\ket{\phi}>0.
    \end{align}
    Then, there exist a sequence of integers $N_k=O(k^2)$, a constant $\alpha>0$, and PVMs $\{R_g^{(k)}\}_{g\in G\cup\{\perp\}}$ on $N_k$ copies such that
    \begin{align}
        U(h)^{\otimes N_k}R_g^{(k)}U(h)^{\dag\otimes N_k}
        =&R_{hg}^{(k)}
    \end{align}
    for all $g,h\in G$, and, for every $g\in G$ and all sufficiently large $k$,
    \begin{align}
        \mathrm{tr}\left(
        \rho_g^{\otimes N_k}
        \sum_{s\in S}R_{gs}^{(k)}
        \right)
        \geq&1-e^{-\alpha k},
        \qquad
        \rho_g:=U(g)\rho U(g)^\dag, 
    \end{align}
    where $S:=\mathrm{Sym}_{G,U}(\rho)$. 
\end{lemma}

\begin{proof}
    We refine the construction in Lemma~\ref{lem:TO_state_discrimination_general}.
    In that construction, instead of defining the projector associated with a coset $gS$ by summing the orthogonal orbit projectors, we define
    \begin{align}
        Q_g^{(k)}
        :=\Pi_{gS}^{(k)} \otimes\ket{\phi_g}\bra{\phi_g}\ \forall g\in G.
    \end{align}
    Since the vectors $\{\ket{\phi_g}\}_{g\in G}$ are mutually orthogonal, the projectors $\{Q_g^{(k)}\}_{g\in G}$ are mutually orthogonal.
    Moreover, we have 
    \begin{align}
        U(h)^{\otimes k+m}Q_g^{(k)}U(h)^{\dag\otimes k+m} 
        =Q_{hg}^{(k)}. 
    \end{align}

    We now apply the same amplification as in Lemma~\ref{lem:TO_state_discrimination_general}, while retaining the group label of the first conclusive outcome.
    We define
    \begin{align}
        R_g^{(k,l)}
        :=&\sum_{j=1}^l
        \left(Q_\perp^{(k)}\right)^{\otimes j-1}
        \otimes Q_g^{(k)}
        \otimes
        \left(
        Q_{gS}^{(k)}+Q_\perp^{(k)}
        \right)^{\otimes l-j} 
    \end{align}
    with 
    \begin{align}
        &Q_\perp^{(k)}:=I-\sum_{g\in G}Q_g^{(k)}, \\
        &Q_{gS}^{(k)}:=\sum_{s\in S}Q_{gs}^{(k)},
    \end{align}
    where $Q_{gS}^{(k)}$ corresponds to the projector used in the proof of Lemma~\ref{lem:TO_state_discrimination_general}.
    The projectors $\{R_g^{(k,l)}\}$ are mutually orthogonal and satisfy
    \begin{align}
        \sum_{s\in S}R_{gs}^{(k,l)}
        =&R_{gS}^{(k,l)},
    \end{align}
    where $R_{gS}^{(k,l)}$ is exactly the coset projector constructed in Lemma~\ref{lem:TO_state_discrimination_general}.
    By the definition of $R_g^{(k, l)}$, we have 
    \begin{align}
        U(h)^{\otimes (k+m)l}R_g^{(k,l)}U(h)^{\dag\otimes (k+m)l}
        =&R_{hg}^{(k,l)}.
    \end{align}
    Taking $l=k$ and defining
    \begin{align}
        R_\perp^{(k)}
        :=&I-\sum_{g\in G}R_g^{(k,k)},
    \end{align}
    the error estimate of Lemma~\ref{lem:TO_state_discrimination_general} directly gives
    \begin{align}
        \mathrm{tr}\left(
        \rho_g^{\otimes (k+m)k}\sum_{s\in S}R_{gs}^{(k)}
        \right)
        \geq &1-e^{-\alpha k}. 
    \end{align}
\end{proof}

For later use, we reindex the above construction by the actual number of input copies.
Set
\begin{align}
    N_k:=&(k+m)k,
\end{align}
and, for each sufficiently large $l$, define
\begin{align}
    k(l):=&\max\{k\in\mathbb{N} \,|\, N_k\leq l\}.
\end{align}
We extend the PVM to $l$ copies by
\begin{align}
    \widetilde R_g^{(l)}:=&
    R_g^{(k(l))}\otimes I^{\otimes(l-N_{k(l)})},
    \qquad g\in G,
\end{align}
and define $\widetilde R_\perp^{(l)}$ analogously.
Then,
\begin{align}
    \mathrm{tr}\left(
    \rho_g^{\otimes l}
    \sum_{s\in S}\widetilde R_{gs}^{(l)}
    \right)\geq&
    1-e^{-\alpha k(l)}.
\end{align}
Since $k(l)\to\infty$ as $l\to\infty$, the discrimination error vanishes.
Moreover, the covariance properties of the original PVM are preserved under this extension.

The preceding lemma provides the symmetry information required to align the work-extraction unitary.
Only a sublinear number of copies is needed for this purpose: the discrimination error vanishes as the number of reference copies diverges, while the fraction of copies consumed by the reference tends to zero.
We can therefore use the remaining copies for an asymptotically optimal unrestricted work-extraction protocol and symmetrize it without changing the work extracted per copy.

\bigskip

\noindent
\textit{Proof of Theorem~\ref{thm:symmetry_protected_ergotropy}.}
    Since the class of $G$-invariant unitaries is included in the class of unitaries, we have 
    \begin{align}
        \overline{W}_{G\text{-}\mathrm{inv}}^{(\infty)}(\rho) \leq \overline{W}^{(\infty)}(\rho). \label{eq:thm:symmetry_protected_ergotropy1}
    \end{align}
    In the following, we prove the opposite inequality.

    First, we explain that we can assume that $U$ is a faithful projective unitary representation without loss of generality. 
    We define 
    \begin{align}
        K:=\{h\in G\ |\ U(h)\propto I\}. 
    \end{align}
    Then, $K$ is a mormal subgroup of $G$, which can be confirmed by noting that for any $g\in G$ and $k\in K$, we have 
    \begin{align}
        U(ghg^{-1})\propto U(g)U(h)U(g^{-1}) 
        =U(h)U(g)U(g^{-1}) 
        \propto U(hgg^{-1}) 
        =U(h).  
    \end{align}
    Thus, $G':=G/K$ forms a group, and 
    We fix representative $g_j$ of each coset, and define a function $U'$ on $G'$ by 
    \begin{align}
        U'(g_j K):=U(g_j). 
    \end{align}
    Then, by the definition of $K$, for all $g\in G$, we have 
    \begin{align}
        U'(gK)=U'(g'K)=U(g')\propto U(g')U(g'^{-1}g) \propto U(g'g'^{-1}g) =U(g),  
    \end{align}
    where $g'$ is the representative of $gK$. 
    This property implies for any $g, g'\in G$, 
    \begin{align}
        U'(gK)U'(g'K) \propto U(g)U(g') \propto U(gg')\propto U'(gg'K), 
    \end{align}
    which means $U'$ is a projective unitary representation of $G/K$. 
    The definition of $K$ directly implies $U'$ is a faithful projective unitary representation. 
    We also note that $G$-invariance using the representation $U$ is equivalent to $G'$-invariance using the representation $U'$. 
    Therefore, in the following, we can assume that $U$ is a faithful projective unitary representation without loss of generality.

    Next, we construct a $G$-invariant unitary operator for work extraction. 
    By Lemmas~\ref{lem:TO_state_discrimination_general} and \ref{SMlem:regular_basis_const_second_prep}, we can take a sequence of PVMs $(\{\widetilde{R}_g^{(k)}\}_{g\in G})_{k\in\mathbb{N}}$ and a constant $\alpha>0$ such that for any $g\in G$ and sufficiently large $k\in\mathbb{N}$, 
    \begin{align}
        \mathrm{tr}\left(\left(\sum_{s\in S} \widetilde{R}_{gs}^{(k)}\right)\left(\mathcal{U}_g(\rho)^{\otimes k}\otimes \mathcal{U}_g(\rho)^{\otimes m}\right)^{\otimes k}\right)\geq 1-e^{-\alpha k}. 
    \end{align}
    We take a sequence of unitaries $\{V^{(n)}\}_{n\in\mathbb{N}}$ that can extract the maximal work from $\rho^{\otimes n}$, i.e., 
    \begin{align}
        \mathrm{tr}\left(\left(\rho^{\otimes n}-V^{(n)}\rho^{\otimes n}V^{(n)\dag}\right)H^{\times n}\right)
        =W^{(n)}(\rho). 
    \end{align}
    Using $\widetilde{R}_g^{(k)}$ and $V^{(n)}$, we define 
    \begin{align}
        \widetilde{V}^{(k, n)}:=\sum_{g\in G} \widetilde{R}_g^{(k)}\otimes U(g)^{\otimes n}V^{(n)} U(g)^{\dag\otimes n}+\widetilde{R}_\perp^{(k)}\otimes I^{\otimes n}, 
    \end{align}
    where $\widetilde{R}_\perp^{(k)}$ is defined by 
    \begin{align}
        \widetilde{R}_\perp^{(k)}:=I-\sum_{g\in G} \widetilde{R}_g^{(k)}. 
    \end{align}
    Since the projectors $\{\widetilde{R}_g^{(k)}\}_{g\in G\cup\{\perp\}}$ are mutually orthogonal and sum to the identity, $\widetilde{V}^{(k, n)}$ is unitary. 
    Since $\{\widetilde{R}_g^{(k)}\}_{g\in G}$ satisfies $U(h)^{\otimes k}\widetilde{R}_g^{(k)} U(h)^{\dag\otimes k}=R_{hg}^{(k)}$ for all $g, h\in G$, $\widetilde{V}~{(k, n)}$ is $G$-invariant. 
    We take an arbitrary sequence $(l_n)_{n\in N}$ such that $l_n\in (0, n)\cap\mathbb{N}$ for all $n\in\mathbb{N}$, $\lim_{n\to\infty} l_n=\infty$, and $\lim_{n\to\infty} l_n/n=0$. 
    We note that for any $n\in\mathbb{N}$, 
    \begin{align}
        &I^{\otimes l_n}\otimes H^{\times n-l_n}-\widetilde{V}^{(l_n, n-l_n)\dag}\left(I^{\otimes l_n}\otimes H^{\times n-l_n}\right)\widetilde{V}^{(l_n, n-l_n)} \nonumber\\
        =&\widetilde{V}^{(l_n, n-l_n)\dag}\left[\widetilde{V}^{(l_n, n-l_n)}, I^{\otimes l_n}\otimes H^{\times n-l_n}\right] \nonumber\\
        =&\left[\sum_{g'\in G} \widetilde{R}_{g'}^{(l_n)}\otimes U(g')^{\otimes n-l_n}V^{(n-l_n)\dag}U(g')^{\dag\otimes n-l_n} +\widetilde{R}_\perp^{(l_n)}\otimes I^{\otimes n-l_n}\right] \nonumber\\
        &\hspace{1cm}\times \left(\sum_{g\in G} \widetilde{R}_g^{(l_n)}\otimes U(g)^{\otimes n-l_n}\left[V^{(n-l_n)}, H^{\times n-l_n}\right]U(g)^{\dag\otimes n-l_n}\right) \nonumber\\
        =&\sum_{g\in G} \widetilde{R}_g^{(l_n)}\otimes U(g)^{\otimes n-l_n}V^{(n-l_n) \dag}\left[V^{(n-l_n)}, H^{\times n-l_n}\right]U(g)^{\dag\otimes n-l_n} \nonumber\\
        =&\sum_{g\in G} \widetilde{R}_g^{(l_n)}\otimes U(g)^{\otimes n-l_n}\left(H^{\times n-l_n}-V^{(n-l_n)\dag} H^{\times n-l_n}V^{(n-l_n)}\right)U(g)^{\dag\otimes n-l_n}, 
    \end{align}
    which implies that 
    \begin{align}
        &\mathrm{tr}\left(\rho^{\otimes n} \left[I^{\otimes l_n}\otimes H^{\times n-l_n}-\widetilde{V}^{(l_n, n-l_n)\dag}\left(I^{\otimes l_n}\otimes H^{\times n-l_n}\right)\widetilde{V}^{(l_n, n-l_n)}\right]\right) \nonumber\\
        =&\sum_{g\in G} \mathrm{tr}\left(\rho^{\otimes l_n}\widetilde{R}_g^{(l_n)}\right)\mathrm{tr}\left(\rho^{\otimes n-l_n} U(g)^{\otimes n-l_n}\left(H^{\times n-l_n}-V^{(n-l_n)\dag} H^{\times n-l_n}V^{(n-l_n)}\right)U(g)^{\dag\otimes n-l_n}\right) \nonumber\\
        \geq &\sum_{g\in S} \mathrm{tr}\left(\rho^{\otimes l_n}\widetilde{R}_g^{(l_n)}\right) W^{(n-l_n)}(\rho)+\sum_{g\in G-S} \mathrm{tr}\left(\rho^{\otimes l_n}\widetilde{R}_g^{(l_n)}\right)(n-l_n)(-\Delta h) \nonumber\\
        \geq &\mathrm{tr}\left(\rho^{\otimes l_n}\widetilde{R}_S^{(l_n)}\right) W^{(n-l_n)}(\rho)+\left(1-\mathrm{tr}\left(\rho^{\otimes l_n}\widetilde{R}_S^{(l_n)}\right)\right)(n-l_n)(-\Delta h) \nonumber\\
        =&(n-l_n)\left[\mathrm{tr}\left(\rho^{\otimes l_n}\widetilde{R}_S^{(l_n)}\right)\overline{W}^{(n-l_n)}(\rho)-\left(1-\mathrm{tr}\left(\rho^{\otimes l_n}\widetilde{R}_S^{(l_n)}\right)\right)\Delta h \right], \label{eq:thm:symmetry_protected_ergotropy2}
    \end{align}
    where $\Delta h:=h_\mathrm{max}-h_\mathrm{min}$ with the maximum and minimum eigenvalues of $H$. 
    On the other hand, since $l_n h_\mathrm{min}I^{\otimes n}\leq H^{\times l_n}\otimes I^{\otimes n-l_n}\leq l_n h_\mathrm{max}I^{\otimes n}$, we have 
    \begin{align}
        \mathrm{tr}\left(\rho^{\otimes n}\left[H^{\times l_n}\otimes I^{\otimes n-l_n}-\widetilde{V}^\dag\left(H^{\times l_n}\otimes I^{\otimes n-l_n}\right)\widetilde{V}\right]\right) 
        \geq -l_n(h_\mathrm{max}-h_\mathrm{min}) 
        =-l_n\Delta h, \label{eq:thm:symmetry_protected_ergotropy3}
    \end{align}
    By Eqs.~\eqref{eq:thm:symmetry_protected_ergotropy2} and \eqref{eq:thm:symmetry_protected_ergotropy3}, we get 
    \begin{align}
        \overline{W}_{G\text{-}\mathrm{inv}}^{(n)}\left(\rho\right) 
        \geq &\left(1-\frac{l_n}{n}\right)\left[\mathrm{tr}\left(\rho^{\otimes l_n}R_S^{(l_n)}\right)\overline{W}^{(n-l_n)}(\rho)-\left(1-\mathrm{tr}\left(\rho^{\otimes l_n}R_S^{(l_n)}\right)\right)\Delta_H \right]-\frac{l_n}{n}(h_\mathrm{max}-h_\mathrm{min}). 
    \end{align}
    Since the right-hand side converges to $\overline{W}^{(\infty)}(\rho)$ in the limit of $n \to\infty$, we have 
    \begin{align}
        \overline{W}_{G\text{-}\mathrm{inv}}^{(\infty)}(\rho)
        \geq \overline{W}^{(\infty)}(\rho). \label{eq:thm:symmetry_protected_ergotropy4}
    \end{align}
    By Eqs.~\eqref{eq:thm:symmetry_protected_ergotropy1} and \eqref{eq:thm:symmetry_protected_ergotropy4}, we have  
    \begin{align}
        \overline{W}_{G\text{-}\mathrm{inv}}^{(\infty)}(\rho)
        =\overline{W}^{(\infty)}(\rho). 
    \end{align}
\hfill $\square$

As a direct consequence, finite-group symmetry does not introduce additional completely passive states.

\begin{corollary} \label{cor:complete_passivity}
    Let $G$ be a finite group, $U$ be a projective unitary representation of $G$, $H$ be a $G$-invariant Hamiltonian, and $\rho$ be a state. 
    Then, $\rho$ is completely passive under $G$-invariant unitaries if and only if it is completely passive, i.e., it is a gound state or a Gibbs state with nonnegative inverse temperature $\beta\in [0, \infty)$. 
\end{corollary}

\begin{proof}
    Since the statement for the ``if'' part is trivial, we show the ``only if'' part. 
    We take arbitrary state $\rho$ and suppose that $\rho$ is not completely passive but $G$-symmetry-protected completely passive. 
    Then, there exists some $k\in\mathbb{N}$ such that 
    \begin{align}
        W^{(k)}(\rho)>0,  
    \end{align}
    We note that for any $n\in\mathbb{N}$, 
    \begin{align}
        W^{(nk)}(\rho)\geq n W^{(k)}(\rho), 
    \end{align}
    or equivalently, 
    \begin{align}
        \overline{W}^{(nk)}(\rho)\geq \overline{W}^{(k)}(\rho), 
    \end{align}
    which implies 
    \begin{align}
        \overline{W}^{(\infty)}(\rho)\geq \overline{W}^{(k)}(\rho). 
    \end{align}
     By Theorem~\ref{thm:symmetry_protected_ergotropy}, we have 
    \begin{align}
        \overline{W}_{G^\text{-}\mathrm{inv}}^{(\infty)}(\rho)\geq \overline{W}^{(k)}(\rho)>0. \label{eq:cor:complete_passivity1}
    \end{align}
    On the other hand, since we assume that $\rho$ is $G$-symmetry-protected completely passive, we have 
    \begin{align}
        W_{G^\text{-}\mathrm{inv}}^{(n)}(\rho)=0 
    \end{align}
    for all $n\in\mathbb{N}$, which implies 
    \begin{align}
        \overline{W}_{G^\text{-}\mathrm{inv}}^{(\infty)}(\rho)=0. 
    \end{align}
    This contradicts with Eq.~\eqref{eq:cor:complete_passivity1}. 
    Therefore, we have proved that if $\rho$ is $G$-symmetry-protected completely passive, then $\rho$ is completely passive. 
\end{proof}

\section{Discussion}

In this work, we investigated how finite-group symmetry constraints interact with several quantum resource theories.
The common structure behind our results is that the information required to compensate for a finite symmetry action is itself finite.
Theorem~\ref{thm:conversion_state_discrimination} establishes a direct relation between state discrimination and state conversion for finite-group asymmetry, while Theorem~\ref{thm:conversion_control_admissible_POVM} extends the discrimination-and-control construction to an underlying resource theory.
The general conversion theorem, Theorem~\ref{thm:rate_lower_bound_general}, shows that asymptotically reliable discrimination, together with compatibility of the conversion with the residual symmetry, is sufficient to achieve the desired rate under the combined constraints.
The symmetry information can be extracted from a sublinear number of input copies, leaving the remaining copies available for the resource conversion.

This asymptotic rate preservation does not mean that a finite-group symmetry constraint is operationally trivial.
First, Theorem~\ref{thm:asymmetry_rate} establishes an exact obstruction: if $\mathrm{Sym}(\rho)\not\subset\mathrm{Sym}(\sigma)$, then no positive asymptotic conversion rate is possible under $G$-covariant operations.
For LOCC, stabilizer operations, completely stabilizer-preserving operations, and Gibbs-preserving operations, Theorems~\ref{thm:LOCC_rate_map_level}, \ref{thm:magic_map_rates}, and~\ref{thm:GP_conversion_rate} show that the symmetry-subgroup condition is sufficient to preserve the corresponding unconstrained rate.
For time-translation-covariant Gibbs-preserving operations and thermal operations, Theorems~\ref{thm:TCGP_conversion_rate} and \ref{thm:covariant_TO_conversion_rate} instead identify the compatibility condition in Definition~\ref{def:group_time_covariance_compatibility} as the relevant condition.
Second, even when the asymptotic rate is unchanged, the discrimination step consumes copies and introduces an error at finite block length.
Finite-copy conversion can therefore remain sensitive to the symmetry constraint even when its first-order asymptotic effect disappears.

The rate equalities established here should also be distinguished from an explicit solution of the underlying resource-conversion problem.
For example, for general mixed multipartite states, the optimal LOCC conversion rate is not known in closed form.
Theorem~\ref{thm:LOCC_rate_map_level} states that, whenever the symmetry-subgroup condition is satisfied, imposing finite-group covariance does not further reduce this rate.
Theorem~\ref{thm:magic_map_rates} gives the corresponding statement for magic-state conversion under stabilizer operations and completely stabilizer-preserving operations, and Theorems~\ref{thm:GP_conversion_rate}, \ref{thm:TCGP_conversion_rate}, and \ref{thm:covariant_TO_conversion_rate} treat the three thermodynamic operation classes.
The role of the symmetry analysis is therefore to determine the additional restriction generated by the symmetry constraint, rather than to solve the underlying resource theory itself.

A central operational ingredient in these results is state discrimination.
Lemma~\ref{lem:LOCC_discrimination} shows that a finite unitary orbit can be distinguished by LOCC with exponentially decreasing error by estimating suitable observables through local measurements.
Lemma~\ref{lem:magic-free-discrimination} provides the corresponding construction using Pauli measurements for magic.
In the time-covariant thermodynamic setting, only the orbit information that survives collective energy dephasing is available to the discrimination step.
Lemma~\ref{lem:dephased_state_symmetry_subgroup} identifies the residual subgroup associated with this restriction.
For thermal operations, Lemmas~\ref{lem:control_admissible_POVM_TO}, \ref{lem:TO_state_discrimination_general}, \ref{lem:TO_state_discrimination} establish control admissibility through an energy-compatible projective construction.
In our protocols, an error $\epsilon$ can be achieved using $O(\log(1/\epsilon))$ copies for LOCC and stabilizer operations and $O((\log(1/\epsilon))^2)$ copies for thermal operations.
These bounds are sufficient to make the discrimination cost sublinear in the total number of input copies.
We have not established their optimality.

More generally, our analysis distinguishes several operational requirements that arise when symmetry is combined with another resource theory.
One must determine which symmetry-related states can be distinguished and whether the measurement outcome can be used to select subsequent free operations.
Definition~\ref{def:control_admissible_POVM} captures this second requirement through control admissibility, without requiring a physical memory that stores the outcome.
Under the free-memory assumption, Lemma~\ref{lem:memory_assisted_control_admissible} relates this formulation to memory-assisted measurements, and Theorem~\ref{thm:conversion_POVM_equivalence} connects discrimination to reference-state distillation and the simulation of symmetry-compatible free operations.
A further question arises when the allowed operations are specified in terms of elementary physical primitives: whether the resulting covariant channel admits a realization using symmetric primitives.

This last distinction is reflected in our separate treatment of map-level and primitive-level symmetry constraints.
At the map level, the resulting channel must belong both to the original free-operation class and to the class of $G$-covariant channels.
At the primitive level, each elementary operation in a realization must itself respect the symmetry.
For LOCC, Proposition~\ref{prop:primitive-map-equivalence} and Theorem~\ref{thm:primitive-conversion-rate} establish the corresponding implementation and rate results under the stated assumptions on ancillary systems and classical communication.
In particular, the globally symmetric setting requires a distinction between freely available shared invariant reference states and only local invariant ancillary preparation.
For thermal operations, Lemma~\ref{lem:invariant_covariant_TO_equivalence} establishes equality of the map-level and primitive-level operation classes, leading to the rate formula in Theorem~\ref{thm:TO_rate_primitive_level}.

For magic, Theorem~\ref{thm:magic_map_rates} allows arbitrary Clifford representations at the map level, whereas the primitive-level result in Theorem~\ref{thm:magic_Pauli_primitive_rates} concerns Pauli representations.
The latter uses symmetric Clifford extensions and joint system-memory Pauli measurements whose observables commute with the full symmetry representation.
The example generated by $SH$ shows that this measurement construction cannot be extended directly to all Clifford representations, even by adjoining Clifford-represented ancillary systems, while leaving open whether the map-level and primitive-level asymptotic conversion rates are actually different.
Theorem~\ref{thm:general-Clifford-memory} provides a complementary result by constructing stabilizer memory states with stabilizer readout and conditional control for arbitrary Clifford representations, without requiring the readout operations to be symmetric.
Determining whether a map-level and primitive-level rate separation occurs under the specified symmetric primitive set, or identifying more general sufficient conditions for equality, remains a separate question.

Our results on work extraction provide a complementary manifestation of the same finite-information principle.
At finite copy number, restricting the work-extracting unitary to be symmetric can reduce the ergotropy.
Nevertheless, Theorem~\ref{thm:symmetry_protected_ergotropy} shows that, for a fixed finite group commuting with the Hamiltonian, the asymptotic ergotropy per copy coincides with its unrestricted value.
Corollary~\ref{cor:complete_passivity} then implies that finite-group symmetry does not enlarge the set of completely passive states.
This statement includes states supported on a degenerate ground eigenspace.
The result complements previous work on symmetry-protected thermal equilibrium and work extraction in Ref.~\cite{mitsuhashi2022characterizing}.

The finite-group assumption is essential to the mechanism developed here.
For a fixed finite group, the orbit contains only finitely many alternatives, so the amount of classical information required to identify the symmetry sector does not grow with the number of input copies.
This argument does not directly extend to continuous groups.
Indeed, pure-state asymptotic conversion under compact Lie-group symmetries can exhibit nontrivial finite rates characterized by the quantum geometric tensor, as established in Ref.~\cite{yamaguchi2026quantum}.
For mixed-state conversions, the optimal rate is characterized by a family of metric-adjusted quantum geometric tensors; in general, no fixed finite subset of this family suffices for all input and output states~\cite{yamaguchi2026quantifyingsymmetrybreakingmetric}.
In that setting, increasingly precise estimation of a continuous group parameter may consume a non-negligible asymptotic resource.
An important direction for future work is therefore to understand how the geometric constraints arising for Lie-group asymmetry interact with additional restrictions such as LOCC, thermodynamic operations, and stabilizer operations.

Several finite-group questions also remain open.
The discrimination protocols constructed here were designed to establish vanishing error with sublinear copy cost rather than to optimize the error exponent.
Determining the optimal discrimination exponent under LOCC, stabilizer, or thermodynamic restrictions could lead to sharper finite-blocklength conversion bounds.
It would also be interesting to investigate regimes in which the group, its representation, or the local Hilbert-space dimension grows with the number of copies.
In such regimes, the amount of symmetry information need no longer remain $O(1)$, and the distinction between finite and continuous symmetry may become less sharp.

Another direction is to characterize more generally the resource theories for which accessible symmetry information can be used for conditional control.
Theorem~\ref{thm:rate_lower_bound_general} gives a sufficient criterion for preserving an achievable rate under an additional finite-group constraint.
Failure of this criterion alone does not imply a reduction of the conversion rate.
A more complete characterization would require distinguishing limitations of the discrimination-and-control construction from restrictions that apply to all conversion protocols.

In summary, finite-group symmetry constraints exhibit a common structure across entanglement, magic, and thermodynamic resource theories.
Theorems~\ref{thm:LOCC_rate_map_level}, \ref{thm:magic_map_rates}, \ref{thm:GP_conversion_rate}, \ref{thm:TCGP_conversion_rate}, and \ref{thm:covariant_TO_conversion_rate} determine when the corresponding map-level asymptotic conversion rates are preserved, while Theorems~\ref{thm:primitive-conversion-rate}, \ref{thm:magic_Pauli_primitive_rates}, and~\ref{thm:TO_rate_primitive_level} establish primitive-level results under their respective assumptions.
Theorem~\ref{thm:symmetry_protected_ergotropy} extends the analysis to asymptotic work extraction.
These results identify broad settings in which symmetry information can be obtained and used at a sublinear copy cost, while separating the exact compatibility conditions from the additional requirements imposed by restricted physical implementations.

\section*{Acknowledgements}
The authors wish to thank 
Kaito Watanabe, 
Ryota Matsuda, 
Koji Yamaguchi, 
Yui Kuramochi, 
and Tomohiro Shitara 
for insightful discussions. 
Y.M. was supported by 
the RIKEN SPDR Program and 
the Hakubi projects of RIKEN.
H.T. was supported by 
JSPS Grants-in-Aid for Scientific Research No. JP25K00924, 
MEXT KAKENHI Grant-in-Aid for Transformative Research Areas B “Quantum Energy Innovation” Grant Numbers 24H00830 and 24H00831, 
JST FOREST No. JPMJFR2365, 
JST MOONSHOT No. JPMJMS256E and 
Royal Society International Collaboration Awards 2025 Flexigrant number ICA/R2/252240.

\bibliography{bib.bib}

@misc{yamaguchi2026quantifyingsymmetrybreakingmetric,
      title={Quantifying Symmetry Breaking with Metric Adjusted Quantum Geometric Tensors}, 
      author={Koji Yamaguchi and Hiroyasu Tajima},
      year={2026},
      eprint={2609.11926},
      archivePrefix={arXiv},
      primaryClass={quant-ph},
      url={https://arxiv.org/abs/2609.11926}, 
}

@article{tajima_coherence_2018,
  title = {Uncertainty Relations in Implementation of Unitary Operations},
  author = {Tajima, Hiroyasu and Shiraishi, Naoto and Saito, Keiji},
  journal = {Phys. Rev. Lett.},
  volume = {121},
  issue = {11},
  pages = {110403},
  numpages = {6},
  year = {2018},
  month = {Sep},
  publisher = {American Physical Society},
  doi = {10.1103/PhysRevLett.121.110403},
  url = {https://link.aps.org/doi/10.1103/PhysRevLett.121.110403}
}

@article{tajima_coherence_2020,
	title = {Coherence cost for violating conservation laws},
	volume = {2},
	url = {https://link.aps.org/doi/10.1103/PhysRevResearch.2.043374},
	doi = {10.1103/PhysRevResearch.2.043374},
	number = {4},
	urldate = {2022-01-05},
	journal = {Physical Review Research},
	author = {Tajima, Hiroyasu and Shiraishi, Naoto and Saito, Keiji},
	year = {2020},
	pages = {043374},
}

@misc{tajima2025universaltradeoffstructuresymmetry,
      title={Universal trade-off structure between symmetry, irreversibility, and quantum coherence in quantum processes}, 
      author={Hiroyasu Tajima and Ryuji Takagi and Yui Kuramochi},
      year={2025},
      eprint={2206.11086},
      archivePrefix={arXiv},
      primaryClass={quant-ph},
      url={https://arxiv.org/abs/2206.11086}, 
}

@article{Tajima2025PRLa,
  title = {Gibbs-Preserving Operations Requiring Infinite Amount of Quantum Coherence},
  author = {Tajima, Hiroyasu and Takagi, Ryuji},
  journal = {Phys. Rev. Lett.},
  volume = {134},
  issue = {17},
  pages = {170201},
  numpages = {7},
  year = {2025},
  month = {Apr},
  publisher = {American Physical Society},
  doi = {10.1103/PhysRevLett.134.170201},
  url = {https://link.aps.org/doi/10.1103/PhysRevLett.134.170201}
}

@article{chitambar2019quantum,
  title = {Quantum resource theories},
  author = {Chitambar, Eric and Gour, Gilad},
  journal = {Reviews of Modern Physics},
  volume = {91},
  number = {2},
  pages = {025001},
  year = {2019},
  doi = {10.1103/RevModPhys.91.025001}
}

@article{bennett1996concentrating,
  title = {Concentrating partial entanglement by local operations},
  author = {Bennett, Charles H. and Bernstein, Herbert J. and Popescu, Sandu and Schumacher, Benjamin},
  journal = {Physical Review A},
  volume = {53},
  number = {4},
  pages = {2046--2052},
  year = {1996},
  doi = {10.1103/PhysRevA.53.2046}
}

@article{brandao2013resource,
  title = {Resource theory of quantum states out of thermal equilibrium},
  author = {Brand{\~a}o, Fernando G. S. L. and Horodecki, Micha{\l} and Oppenheim, Jonathan and Renes, Joseph M. and Spekkens, Robert W.},
  journal = {Physical Review Letters},
  volume = {111},
  number = {25},
  pages = {250404},
  year = {2013},
  doi = {10.1103/PhysRevLett.111.250404}
}

@article{veitch2014resource,
  title = {The resource theory of stabilizer computation},
  author = {Veitch, Victor and Mousavian, S. A. Hamed and Gottesman, Daniel and Emerson, Joseph},
  journal = {New Journal of Physics},
  volume = {16},
  number = {1},
  pages = {013009},
  year = {2014},
  doi = {10.1088/1367-2630/16/1/013009}
}

@article{bartlett2007reference,
  title = {Reference frames, superselection rules, and quantum information},
  author = {Bartlett, Stephen D. and Rudolph, Terry and Spekkens, Robert W.},
  journal = {Reviews of Modern Physics},
  volume = {79},
  number = {2},
  pages = {555--609},
  year = {2007},
  doi = {10.1103/RevModPhys.79.555}
}

@article{gour2008resource,
  title = {The resource theory of quantum reference frames: manipulations and monotones},
  author = {Gour, Gilad and Spekkens, Robert W.},
  journal = {New Journal of Physics},
  volume = {10},
  number = {3},
  pages = {033023},
  year = {2008},
  doi = {10.1088/1367-2630/10/3/033023}
}

@article{faist2015gibbs,
  title = {Gibbs-preserving maps outperform thermal operations in the quantum regime},
  author = {Faist, Philippe and Oppenheim, Jonathan and Renner, Renato},
  journal = {New Journal of Physics},
  volume = {17},
  number = {4},
  pages = {043003},
  year = {2015},
  doi = {10.1088/1367-2630/17/4/043003}
}

@article{lostaglio2015description,
  title = {Description of quantum coherence in thermodynamic processes requires constraints beyond free energy},
  author = {Lostaglio, Matteo and Jennings, David and Rudolph, Terry},
  journal = {Nature Communications},
  volume = {6},
  pages = {6383},
  year = {2015},
  doi = {10.1038/ncomms7383}
}

@article{shitara2025iid,
  title = {The i.i.d. state convertibility in the resource theory of asymmetry for finite groups and Lie groups},
  author = {Shitara, Tomohiro and Mitsuhashi, Yosuke and Tajima, Hiroyasu},
  journal = {arXiv preprint arXiv:2312.15758},
  year = {2025},
  eprint = {2312.15758},
  archivePrefix = {arXiv},
  primaryClass = {quant-ph}
}

@article{yamaguchi2026quantum,
  title = {Quantum geometric tensor determines the pure-state i.i.d. conversion rate in the resource theory of asymmetry for any compact Lie group},
  author = {Yamaguchi, Koji and Mitsuhashi, Yosuke and Shitara, Tomohiro and Tajima, Hiroyasu},
  journal = {Physical Review X},
  volume = {16},
  number = {3},
  pages = {031028},
  year = {2026},
  doi = {10.1103/qqf6-x85b}
}

@article{li2016discriminating,
  title = {Discriminating quantum states: The multiple Chernoff distance},
  author = {Li, Ke},
  journal = {The Annals of Statistics},
  volume = {44},
  number = {4},
  pages = {1661--1679},
  year = {2016},
  doi = {10.1214/16-AOS1436}
}

@article{pusz1978passive,
  title = {Passive states and KMS states for general quantum systems},
  author = {Pusz, Wies{\l}aw and Woronowicz, Stanis{\l}aw L.},
  journal = {Communications in Mathematical Physics},
  volume = {58},
  number = {3},
  pages = {273--290},
  year = {1978},
  doi = {10.1007/BF01614224}
}

@article{alicki2013entanglement,
  title = {Entanglement boost for extractable work from ensembles of quantum batteries},
  author = {Alicki, Robert and Fannes, Mark},
  journal = {Physical Review E},
  volume = {87},
  number = {4},
  pages = {042123},
  year = {2013},
  doi = {10.1103/PhysRevE.87.042123}
}

@article{mitsuhashi2022characterizing,
  title = {Characterizing symmetry-protected thermal equilibrium by work extraction},
  author = {Mitsuhashi, Yosuke and Kaneko, Kazuya and Sagawa, Takahiro},
  journal = {Physical Review X},
  volume = {12},
  number = {2},
  pages = {021013},
  year = {2022},
  doi = {10.1103/PhysRevX.12.021013}
}

@article{mitsuhashi2023clifford,
  title = {Clifford group and unitary designs under symmetry},
  author = {Mitsuhashi, Yosuke and Yoshioka, Nobuyuki},
  journal = {PRX Quantum},
  volume = {4},
  number = {4},
  pages = {040331},
  year = {2023},
  doi = {10.1103/PRXQuantum.4.040331}
}

@article{seddon2019quantifying,
  title = {Quantifying magic for multi-qubit operations},
  author = {Seddon, James R. and Campbell, Earl T.},
  journal = {Proceedings of the Royal Society A},
  volume = {475},
  number = {2227},
  pages = {20190251},
  year = {2019},
  doi = {10.1098/rspa.2019.0251}
}

@article{gour2020optimal,
  title = {Optimal extensions of resource measures and their applications},
  author = {Gour, Gilad and Tomamichel, Marco},
  journal = {Phys. Rev. A},
  volume = {102},
  issue = {6},
  pages = {062401},
  numpages = {13},
  year = {2020},
  month = {Dec},
  publisher = {American Physical Society},
  doi = {10.1103/PhysRevA.102.062401},
  url = {https://link.aps.org/doi/10.1103/PhysRevA.102.062401}
}

@article{coecke2016mathematical,
  author = {Bob Coecke and Tobias Fritz and Robert W. Spekkens},
  title = {A mathematical theory of resources},
  journal = {Information and Computation},
  volume = {250},
  pages = {59--86},
  year = {2016},
  doi = {10.1016/j.ic.2016.02.008}
}

@article{horodecki2009quantum,
  author = {Ryszard Horodecki and Pawe{\l} Horodecki and Micha{\l} Horodecki and Karol Horodecki},
  title = {Quantum entanglement},
  journal = {Reviews of Modern Physics},
  volume = {81},
  pages = {865--942},
  year = {2009},
  doi = {10.1103/RevModPhys.81.865}
}

@article{bennett1996mixed-state,
  author = {Charles H. Bennett and David P. DiVincenzo and John A. Smolin and William K. Wootters},
  title = {Mixed-state entanglement and quantum error correction},
  journal = {Physical Review A},
  volume = {54},
  pages = {3824--3851},
  year = {1996},
  doi = {10.1103/PhysRevA.54.3824}
}

@article{nielsen1999conditions,
  author = {Michael A. Nielsen},
  title = {Conditions for a Class of Entanglement Transformations},
  journal = {Physical Review Letters},
  volume = {83},
  pages = {436--439},
  year = {1999},
  doi = {10.1103/PhysRevLett.83.436}
}

@article{chitambar2014everything,
  author = {Eric Chitambar and Debbie Leung and Laura Man{\v c}inska and Maris Ozols and Andreas Winter},
  title = {Everything You Always Wanted to Know About LOCC (But Were Afraid to Ask)},
  journal = {Communications in Mathematical Physics},
  volume = {328},
  pages = {303--326},
  year = {2014},
  doi = {10.1007/s00220-014-1953-9}
}

@article{janzing2000thermodynamic,
  author = {Dominik Janzing and Pawel Wocjan and Robert Zeier and Rubino Geiss and Thomas Beth},
  title = {Thermodynamic Cost of Reliability and Low Temperatures: Tightening Landauer's Principle and the Second Law},
  journal = {International Journal of Theoretical Physics},
  volume = {39},
  pages = {2717--2753},
  year = {2000},
  doi = {10.1023/A:1026422630734}
}

@article{horodecki2013fundamental,
  author = {Micha{\l} Horodecki and Jonathan Oppenheim},
  title = {Fundamental limitations for quantum and nanoscale thermodynamics},
  journal = {Nature Communications},
  volume = {4},
  pages = {2059},
  year = {2013},
  doi = {10.1038/ncomms3059}
}

@article{brandao2015second,
  author = {Fernando G. S. L. Brand{\~a}o and Micha{\l} Horodecki and Nelly Ng and Jonathan Oppenheim and Stephanie Wehner},
  title = {The second laws of quantum thermodynamics},
  journal = {Proceedings of the National Academy of Sciences},
  volume = {112},
  pages = {3275--3279},
  year = {2015},
  doi = {10.1073/pnas.1411728112}
}

@article{goold2016role,
  author = {John Goold and Marcus Huber and Arnau Riera and L{\'i}dia del Rio and Paul Skrzypczyk},
  title = {The role of quantum information in thermodynamics---a topical review},
  journal = {Journal of Physics A: Mathematical and Theoretical},
  volume = {49},
  pages = {143001},
  year = {2016},
  doi = {10.1088/1751-8113/49/14/143001}
}

@article{lostaglio2019introductory,
  author = {Matteo Lostaglio},
  title = {An introductory review of the resource theory approach to thermodynamics},
  journal = {Reports on Progress in Physics},
  volume = {82},
  pages = {114001},
  year = {2019},
  doi = {10.1088/1361-6633/ab46e5}
}

@article{bravyi2005universal,
  author = {Sergey Bravyi and Alexei Kitaev},
  title = {Universal quantum computation with ideal Clifford gates and noisy ancillas},
  journal = {Physical Review A},
  volume = {71},
  pages = {022316},
  year = {2005},
  doi = {10.1103/PhysRevA.71.022316}
}

@article{veitch2012negative,
  author = {Victor Veitch and Christopher Ferrie and David Gross and Joseph Emerson},
  title = {Negative quasi-probability as a resource for quantum computation},
  journal = {New Journal of Physics},
  volume = {14},
  pages = {113011},
  year = {2012},
  doi = {10.1088/1367-2630/14/11/113011}
}

@article{howard2014contextuality,
  author = {Mark Howard and Joel Wallman and Victor Veitch and Joseph Emerson},
  title = {Contextuality supplies the `magic' for quantum computation},
  journal = {Nature},
  volume = {510},
  pages = {351--355},
  year = {2014},
  doi = {10.1038/nature13460}
}

@article{ahmadi2018quantification,
  author = {Mehdi Ahmadi and Hoan Bui Dang and Gilad Gour and Barry C. Sanders},
  title = {Quantification and manipulation of magic states},
  journal = {Physical Review A},
  volume = {97},
  pages = {062332},
  year = {2018},
  doi = {10.1103/PhysRevA.97.062332}
}

@article{heinrich2019robustness,
  author = {Markus Heinrich and David Gross},
  title = {Robustness of Magic and Symmetries of the Stabiliser Polytope},
  journal = {Quantum},
  volume = {3},
  pages = {132},
  year = {2019},
  doi = {10.22331/q-2019-04-08-132}
}

@article{sparaciari2020first,
  author = {Carlo Sparaciari and L{\'i}dia del Rio and Carlo Maria Scandolo and Philippe Faist and Jonathan Oppenheim},
  title = {The first law of general quantum resource theories},
  journal = {Quantum},
  volume = {4},
  pages = {259},
  year = {2020},
  doi = {10.22331/q-2020-04-30-259}
}

@article{guryanova2016thermodynamics,
  author = {Yelena Guryanova and Sandu Popescu and Anthony J. Short and Ralph Silva and Paul Skrzypczyk},
  title = {Thermodynamics of quantum systems with multiple conserved quantities},
  journal = {Nature Communications},
  volume = {7},
  pages = {12049},
  year = {2016},
  doi = {10.1038/ncomms12049}
}

@article{yungerhalpern2016microcanonical,
  author = {Nicole {Yunger Halpern} and Philippe Faist and Jonathan Oppenheim and Andreas Winter},
  title = {Microcanonical and resource-theoretic derivations of the thermal state of a quantum system with noncommuting charges},
  journal = {Nature Communications},
  volume = {7},
  pages = {12051},
  year = {2016},
  doi = {10.1038/ncomms12051}
}

@article{marvian2013theory,
  author = {Iman Marvian and Robert W. Spekkens},
  title = {The theory of manipulations of pure state asymmetry: I. Basic tools, equivalence classes and single copy transformations},
  journal = {New Journal of Physics},
  volume = {15},
  pages = {033001},
  year = {2013},
  doi = {10.1088/1367-2630/15/3/033001}
}

@article{gour2009measuring,
  author = {Gilad Gour and Iman Marvian and Robert W. Spekkens},
  title = {Measuring the quality of a quantum reference frame: The relative entropy of frameness},
  journal = {Physical Review A},
  volume = {80},
  pages = {012307},
  year = {2009},
  doi = {10.1103/PhysRevA.80.012307}
}

@article{marvian2014extending,
  author = {Iman Marvian and Robert W. Spekkens},
  title = {Extending Noether's theorem by quantifying the asymmetry of quantum states},
  journal = {Nature Communications},
  volume = {5},
  pages = {3821},
  year = {2014},
  doi = {10.1038/ncomms4821}
}

@article{marvian2014modes,
  author = {Iman Marvian and Robert W. Spekkens},
  title = {Modes of asymmetry: The application of harmonic analysis to symmetric quantum dynamics and quantum reference frames},
  journal = {Physical Review A},
  volume = {90},
  pages = {062110},
  year = {2014},
  doi = {10.1103/PhysRevA.90.062110}
}

@article{verstraete2003quantum,
  author = {Frank Verstraete and J. Ignacio Cirac},
  title = {Quantum Nonlocality in the Presence of Superselection Rules and Data Hiding Protocols},
  journal = {Physical Review Letters},
  volume = {91},
  pages = {010404},
  year = {2003},
  doi = {10.1103/PhysRevLett.91.010404}
}

@article{bartlett2003entanglement,
  author = {Stephen D. Bartlett and Howard M. Wiseman},
  title = {Entanglement Constrained by Superselection Rules},
  journal = {Physical Review Letters},
  volume = {91},
  pages = {097903},
  year = {2003},
  doi = {10.1103/PhysRevLett.91.097903}
}

@article{schuch2004nonlocal,
  author = {Norbert Schuch and Frank Verstraete and J. Ignacio Cirac},
  title = {Nonlocal Resources in the Presence of Superselection Rules},
  journal = {Physical Review Letters},
  volume = {92},
  pages = {087904},
  year = {2004},
  doi = {10.1103/PhysRevLett.92.087904}
}

@article{schuch2004quantum,
  author = {Norbert Schuch and Frank Verstraete and J. Ignacio Cirac},
  title = {Quantum entanglement theory in the presence of superselection rules},
  journal = {Physical Review A},
  volume = {70},
  pages = {042310},
  year = {2004},
  doi = {10.1103/PhysRevA.70.042310}
}

@article{lostaglio2015quantum,
  author = {Matteo Lostaglio and Kamil Korzekwa and David Jennings and Terry Rudolph},
  title = {Quantum Coherence, Time-Translation Symmetry, and Thermodynamics},
  journal = {Physical Review X},
  volume = {5},
  pages = {021001},
  year = {2015},
  doi = {10.1103/PhysRevX.5.021001}
}

@article{jonathan1999entanglement-assisted,
  author = {Daniel Jonathan and Martin B. Plenio},
  title = {Entanglement-Assisted Local Manipulation of Pure Quantum States},
  journal = {Physical Review Letters},
  volume = {83},
  pages = {3566--3569},
  year = {1999},
  doi = {10.1103/PhysRevLett.83.3566}
}

@article{aberg2014catalytic,
  author = {Johan {\AA}berg},
  title = {Catalytic Coherence},
  journal = {Physical Review Letters},
  volume = {113},
  pages = {150402},
  year = {2014},
  doi = {10.1103/PhysRevLett.113.150402}
}

@article{marvian2019no-broadcasting,
  author = {Iman Marvian and Robert W. Spekkens},
  title = {No-Broadcasting Theorem for Quantum Asymmetry and Coherence and a Trade-off Relation for Approximate Broadcasting},
  journal = {Physical Review Letters},
  volume = {123},
  pages = {020404},
  year = {2019},
  doi = {10.1103/PhysRevLett.123.020404}
}

@article{lostaglio2019coherence,
  author = {Matteo Lostaglio and Markus P. M{\"u}ller},
  title = {Coherence and Asymmetry Cannot be Broadcast},
  journal = {Physical Review Letters},
  volume = {123},
  pages = {020403},
  year = {2019},
  doi = {10.1103/PhysRevLett.123.020403}
}

@article{audenaert2007discriminating,
  author = {Koenraad M. R. Audenaert and John Calsamiglia and Ramon Mu{\~n}oz-Tapia and Emilio Bagan and Lluis Masanes and Antonio Ac{\'i}n and Frank Verstraete},
  title = {Discriminating States: The Quantum Chernoff Bound},
  journal = {Physical Review Letters},
  volume = {98},
  pages = {160501},
  year = {2007},
  doi = {10.1103/PhysRevLett.98.160501}
}

@article{nussbaum2009chernoff,
  author = {Michael Nussbaum and Arleta Szko{\l}a},
  title = {The Chernoff lower bound for symmetric quantum hypothesis testing},
  journal = {The Annals of Statistics},
  volume = {37},
  pages = {1040--1057},
  year = {2009},
  doi = {10.1214/08-AOS593}
}

@article{bennett1999quantum,
  author = {Charles H. Bennett and David P. DiVincenzo and Christopher A. Fuchs and Tal Mor and Eric Rains and Peter W. Shor and John A. Smolin and William K. Wootters},
  title = {Quantum nonlocality without entanglement},
  journal = {Physical Review A},
  volume = {59},
  pages = {1070--1091},
  year = {1999},
  doi = {10.1103/PhysRevA.59.1070}
}

@article{walgate2000local,
  author = {Jonathan Walgate and Anthony J. Short and Lucien Hardy and Vlatko Vedral},
  title = {Local Distinguishability of Multipartite Orthogonal Quantum States},
  journal = {Physical Review Letters},
  volume = {85},
  pages = {4972--4975},
  year = {2000},
  doi = {10.1103/PhysRevLett.85.4972}
}

@article{hoeffding1963probability,
  author = {Wassily Hoeffding},
  title = {Probability Inequalities for Sums of Bounded Random Variables},
  journal = {Journal of the American Statistical Association},
  volume = {58},
  pages = {13--30},
  year = {1963},
  doi = {10.1080/01621459.1963.10500830}
}

@article{lenard1978thermodynamical,
  author = {A. Lenard},
  title = {Thermodynamical proof of the Gibbs formula for elementary quantum systems},
  journal = {Journal of Statistical Physics},
  volume = {19},
  pages = {575--586},
  year = {1978},
  doi = {10.1007/BF01011769}
}

@article{allahverdyan2004maximal,
  author = {A. E. Allahverdyan and R. Balian and Th. M. Nieuwenhuizen},
  title = {Maximal work extraction from finite quantum systems},
  journal = {Europhysics Letters},
  volume = {67},
  pages = {565--571},
  year = {2004},
  doi = {10.1209/epl/i2004-10101-2}
}

\appendix

\section{Distinguishability measures}

In this appendix, we collect elementary properties of the divergence measures used in the general framework.
These results are used to convert the operator inequalities arising from the discrimination-and-control construction into bounds on the chosen error measure.

\begin{lemma} \label{SMlem:distinguishability_orthogonal_general}
    Let $\mathcal{H}$ and $\mathcal{H}'$ be finite-dimensional Hilbert spaces, $\Delta$ be a distinguishability measure, $\rho, \sigma\in\mathcal{S}(\mathcal{H})$ be orthogonal to each other, and $\rho', \sigma'\in\mathcal{S}(\mathcal{H}')$. 
    Then, for any $x\in [0, \infty]$, 
    \begin{align}
        \Delta(\rho, e^{-x}\rho+(1-e^{-x})\sigma)\geq \Delta(\rho', e^{-x}\rho'+(1-e^{-x})\sigma'). 
    \end{align}
\end{lemma}

\begin{proof}
    Since $\rho$ and $\sigma$ are orthogonal to each other, the projection operator $\Pi$ onto the support of $\rho$ satisfies 
    \begin{align}
        \mathrm{tr}(\Pi \rho)=1, \\
        \mathrm{tr}(\Pi \sigma)=0. 
    \end{align}
    We define a CPTP map $\mathcal{E}$ by 
    \begin{align}
        \mathcal{E}(L)
        :=\mathrm{tr}(\Pi L)\rho'
        +\mathrm{tr}((I-\Pi) L)\sigma'\ \forall L\in\mathcal{L}(\mathcal{H}). 
    \end{align}
    Then, $\mathcal{E}$ satisfies 
    \begin{align}
        \mathcal{E}(\rho)=\rho', \\
        \mathcal{E}(\sigma)=\sigma'. 
    \end{align}
    Therefore, by the definition of the distinguishability measure, we get 
    \begin{align}
        \Delta(\rho, e^{-x}\rho+(1-e^{-x})\sigma)\geq \Delta(\rho', e^{-x}\rho'+(1-e^{-x})\sigma') 
    \end{align}
    for all $x\in [0, \infty]$. 
\end{proof}

When both pairs of states $(\rho, \sigma)$ and $(\rho', \sigma')$ are orthogonal to each other, applying Lemma~\ref{SMlem:distinguishability_orthogonal_general} directly leads to the following corollary.

\begin{corollary} \label{SMcor:distinguishability_state_independence}
    Let $\mathcal{H}$ and $\mathcal{H}'$ be finite-dimensional Hilbert spaces, $\Delta$ be a distinguishability measure, $\rho, \sigma\in\mathcal{S}(\mathcal{H}')$ be orthogonal to each other, and $\rho', \sigma'\in\mathcal{S}(\mathcal{H})$ be orthogonal to each other. 
    Then, for any $x\in [0, \infty]$, 
    \begin{align}
        \Delta(\rho, e^{-x}\rho+(1-e^{-x})\sigma)= \Delta(\rho', e^{-x}\rho'+(1-e^{-x})\sigma'). 
    \end{align}
\end{corollary}

In the following two lemmas, we see two basic properties of the orthogonal distinguishability function. 
The first is the monotonicity. 
We note that the following lemma only shows that $f_\Delta$ is an increasing function, and $f_\Delta$ is not always a strictly increasing function.

\begin{lemma} \label{SMlem:orthogonal_distingishability_function_monotonicity}
    Let $\Delta$ be a distinguishability measure, and $f_\Delta$ be its orthogonal distinguishability function. 
    Then, $f_\Delta$ is an increasing function. 
\end{lemma}

\begin{proof}
    We take arbitrary $x, y$ such that $0\leq x< y\leq\infty$, and arbitrary orthogonal states $\rho$ and $\sigma$. 
    We define a CPTP map $\mathcal{E}$ by 
    \begin{align}
        \mathcal{E}(L):=\frac{1-e^{-x}}{1-e^{-y}}L+\left(1-\frac{1-e^{-x}}{1-e^{-y}}\right)\mathrm{tr}(L)\rho. 
    \end{align}
    Then, the monotonicity of $\Delta$ implies 
    \begin{align}
        f_\Delta(y) 
        =\Delta(\rho, e^{-y}\rho+(1-e^{-y})\sigma) 
        \geq \Delta(\mathcal{E}(\rho), \mathcal{E}(e^{-y}\rho+(1-e^{-y})\sigma)) 
        =\Delta(\rho, e^{-x}\rho+(1-e^{-x})\sigma) 
        =f_\Delta(x), 
    \end{align}
    which means that $f_\Delta$ is an increasing function. 
\end{proof}

As the second property of $f_\Delta$, we show that $\Delta(\rho, \sigma)$ can be upper and lower bounded by $f_\Delta(D^\mathrm{max}(\rho\|\sigma))$ and $f_\Delta(D^\mathrm{min}(\rho\|\sigma))$, respectively. 
The following bound is closely related to the extremal role of the min- and max-relative entropies established in Ref.~\cite{gour2020optimal}.
Here, we use only the data-processing property of a general distinguishability measure and express the resulting bounds through its orthogonal distinguishability function.

\begin{lemma} \label{lem:distingishability_evaluation_Dmax_Dmin} 
    Let $\Delta$ be a distinguishability measure, $f_\Delta$ be its orthogonal distinguishability function, and $\rho, \sigma\in\mathcal{S}(\mathcal{H})$. 
    Then, 
    \begin{align}
        f_\Delta(D^\mathrm{min}(\rho||\sigma))
        \leq\Delta(\rho, \sigma)\leq 
        f_\Delta(D^\mathrm{max}(\rho||\sigma)). \label{SMeq:lem:distingishability_evaluation_Dmax_Dmin1}
    \end{align}
\end{lemma}

\bigskip

When $\Delta$ is trace distance or infidelity, $f_\Delta(x):=1-e^{-x}$. 
Thus, this lemma includes the following relations. 
\begin{align}
    &1-e^{-D^\mathrm{min}(\rho\| \sigma)}
    \leq T(\rho, \sigma)
    \leq 1-e^{-D^\mathrm{max}(\rho\|\sigma)}, \\
    &1-e^{-D^\mathrm{min}(\rho\|\sigma)}
    \leq 1-F(\rho, \sigma)
    \leq 1-e^{-D^\mathrm{max}(\rho\|\sigma)}. 
\end{align}

\bigskip

\begin{proof}
    First, we show the left inequality in Eq.~\eqref{SMeq:lem:distingishability_evaluation_Dmax_Dmin1}. 
    We define a CPTP map $\mathcal{M}$ by 
    \begin{align}
        \mathcal{M}(L):=\mathrm{tr}(L\Pi_\rho)\ket{0}\bra{0}+\mathrm{tr}(L(I-\Pi_\rho))\ket{1}\bra{1}. 
    \end{align}
    Then, we have 
    \begin{align}
        \Delta(\rho, \sigma) 
        \geq &\Delta(\mathcal{M}(\rho), \mathcal{M}(\sigma)) \nonumber\\
        =&\Delta(\ket{0}\bra{0}, \mathrm{tr}(\sigma\Pi_\rho)\ket{0}\bra{0}+(1-\mathrm{tr}(\sigma\Pi_\rho))\ket{1}\bra{1}) \nonumber\\
        =&\Delta\left(\ket{0}\bra{0}, e^{-D^\mathrm{min}(\rho ||\sigma)}\ket{0}\bra{0}+\left(1-e^{-D^\mathrm{min}(\rho ||\sigma)}\right)\ket{1}\bra{1}\right) \nonumber\\
        =&f_\Delta(D^\mathrm{min}(\rho ||\sigma)). 
    \end{align}
    We note that this argument holds for the case when $D^\mathrm{min}(\rho \|\sigma)=\infty$, i.e., $\mathrm{tr}(\Pi_\rho \sigma)=0$.

    Next, we show the second inequality in Eq.~\eqref{SMeq:lem:distingishability_evaluation_Dmax_Dmin1}. 
    When $\rho=\sigma$, we can directly confirm the inequality as follows: 
    \begin{align}
        \Delta(\rho, \sigma) 
        =\Delta(\rho, \rho) 
        \leq\Delta(\ket{0}\bra{0}, \ket{0}\bra{0}) 
        =f_\Delta(0)
        =f_\Delta(D^\mathrm{max}(\rho\|\sigma)). 
    \end{align}
    In the following, we consider the case where $\rho\neq\sigma$, i.e., $D^\mathrm{max}(\rho\|\sigma)>0$. 
    Then, we have 
    \begin{align}
        \sigma\geq e^{-D^\mathrm{max}(\rho\|\sigma)}\rho. 
    \end{align}
    We define a state 
    \begin{align}
        \sigma':=\frac{\sigma-e^{-D^\mathrm{max}(\rho\|\sigma)}\rho}{1-e^{-D^\mathrm{max}(\rho\|\sigma)}}, 
    \end{align}
    where we note that the denominator is nonzero, since we exclude the case where $\rho=\sigma$. 
    Then, we have 
    \begin{align}
        \sigma=e^{-D^\mathrm{max}(\rho\|\sigma)}\rho+\left(1-e^{-D^\mathrm{max}(\rho\|\sigma)}\right)\sigma', 
    \end{align}
    which implies that 
    \begin{align}
        \Delta(\rho, \sigma)
        =&\Delta\left(\rho, e^{-D^\mathrm{max}(\rho\|\sigma)}\rho+\left(1-e^{-D^\mathrm{max}(\rho\|\sigma)}\right)\sigma'\right) \nonumber\\
        \leq &\Delta\left(\ket{0}\bra{0}, e^{-D^\mathrm{max}(\rho\|\sigma)}\ket{0}\bra{0}+\left(1-e^{-D^\mathrm{max}(\rho\|\sigma)}\right)\ket{1}\bra{1}\right) \nonumber\\
        =&f_\Delta(D^\mathrm{max}(\rho\|\sigma)). 
    \end{align}

\end{proof}

Lemma~\ref{lem:distingishability_evaluation_Dmax_Dmin} directly implies the relation among the conditions about the upper bounds of $D^\mathrm{max}(\rho\|\sigma)$, $\Delta(\rho, \sigma)$, and $D^\mathrm{min}(\rho\|\sigma)$.

\begin{corollary} \label{cor:distingishability_relation_Dmax_general_Dmin}
    Let $\Delta$ be a distinguishability measure, $f_\Delta$ be its orthogonal distinguishability function, $x\in [0, \infty]$, and the conditions about the distinguishability between states $\rho$ and $\sigma$ be given as follows: 
    \begin{align}
        &(i)\ D^\mathrm{max}(\rho\|\sigma)\leq x, \\
        &(ii)\ \Delta(\rho, \sigma)\leq f_\Delta(x), \\
        &(iii)\ D^\mathrm{min}(\rho\|\sigma)\leq x. 
    \end{align}
    Then, (i) implies (ii). 
    When $f_\Delta$ is strictly increasing, (ii) implies (iii). 
\end{corollary}

When $\Delta$ is trace distance, this corollary means that $D^\mathrm{max}(\rho\|\sigma)\leq x$ implies $T(\rho, \sigma)\leq 1-e^{-x}$, and $T(\rho, \sigma)\leq 1-e^{-x}$ implies $D^\mathrm{min}(\rho\|\sigma)\leq x$. 
When $\Delta$ is infidelity, it means that $D^\mathrm{max}(\rho\|\sigma)\leq x$ implies $1-F(\rho, \sigma)\leq 1-e^{-x}$, and $1-F(\rho, \sigma)\leq 1-e^{-x}$ implies $D^\mathrm{min}(\rho\|\sigma)\leq x$.

\begin{proof}
    First, we show (i) $\rightarrow$ (ii). 
    We take arbitrary states $\rho$ and $\sigma$ such that 
    \begin{align}
        D^\mathrm{max}(\rho\|\sigma)\leq x. \label{SMeq:cor:distingishability_relation_Dmax_general_Dmin1}
    \end{align}
    Since $f_\Delta$ is an increasing function by Lemma~\ref{SMlem:orthogonal_distingishability_function_monotonicity}, Eq.~\eqref{SMeq:cor:distingishability_relation_Dmax_general_Dmin1} implies 
    \begin{align}
        f_\Delta(D^\mathrm{max}(\rho\|\sigma))\leq f_\Delta(x). \label{SMeq:cor:distingishability_relation_Dmax_general_Dmin2}
    \end{align}
    By Lemma~\ref{lem:distingishability_evaluation_Dmax_Dmin}, we have 
    \begin{align}
        \Delta(\rho, \sigma)\leq f_\Delta(D^\mathrm{max}(\rho\|\sigma)). \label{SMeq:cor:distingishability_relation_Dmax_general_Dmin3}
    \end{align}
    By combining Eqs.~\eqref{SMeq:cor:distingishability_relation_Dmax_general_Dmin2} and \eqref{SMeq:cor:distingishability_relation_Dmax_general_Dmin3}, we get 
    \begin{align}
        \Delta(\rho, \sigma)\leq f_\Delta(x). 
    \end{align}

    Next, we show (ii) $\rightarrow$ (iii) when $f_\Delta$ is strictly increasing. 
    We take arbitrary states $\rho$ and $\sigma$ such that 
    \begin{align}
        \Delta(\rho, \sigma)\leq f_\Delta(x). \label{SMeq:cor:distingishability_relation_Dmax_general_Dmin4}
    \end{align}
    By Lemma~\ref{lem:distingishability_evaluation_Dmax_Dmin}, we have
    \begin{align}
        f_\Delta(D^\mathrm{min}(\rho\|\sigma))\leq\Delta(\rho, \sigma). \label{SMeq:cor:distingishability_relation_Dmax_general_Dmin5}
    \end{align}
    By combining Eqs.~\eqref{SMeq:cor:distingishability_relation_Dmax_general_Dmin4} and \eqref{SMeq:cor:distingishability_relation_Dmax_general_Dmin5}, we get 
    \begin{align}
        f_\Delta(D^\mathrm{min}(\rho\|\sigma))\leq f_\Delta(x). \label{SMeq:cor:distingishability_relation_Dmax_general_Dmin6}
    \end{align}
    Since $f_\Delta$ is a strictly increasing function, Eq.~\eqref{SMeq:cor:distingishability_relation_Dmax_general_Dmin6} implies 
    \begin{align}
        D^\mathrm{min}(\rho\|\sigma)\leq x. 
    \end{align}
\end{proof}

\section{Technical lemmas for symmetric implementations}
\label{sec:LOCC_technical_lemmas}

In this appendix, we establish the technical results used in the primitive-level implementation of covariant LOCC channels.
The basic idea is to attach a register carrying the left regular representation and to use its basis states to label rotated versions of the elementary operations in a Stinespring realization.

Let $G$ be a finite group, and let $\{\ket{\xi_g}\}_{g\in G}$ be an orthonormal basis of a register $R$ carrying the left regular representation $L(h)\ket{\xi_g}=\ket{\xi_{hg}}$. 
Throughout this appendix, the representations acting on the physical systems may be projective.
Whenever the same physical system appears before and after an elementary operation, we use the same projective representation on that system.

The construction is organized so that the label $g$ stored in $R$ is preserved throughout the protocol.
Conditioned on $R$ being in $\ket{\xi_g}$, every elementary operation in the original realization is replaced by its $g$-rotated version.
For a unitary $V$, this means replacing $V$ by the corresponding conjugated unitary $U(g)VU(g)^\dag$. 
For a projective measurement $\{P_x\}_x$, it means replacing each projection $P_x$ by its conjugate.
For state introduction, a slightly different construction is required: we first introduce a symmetric state containing an additional regular-representation register, align its group label with the label already stored in $R$ by a symmetric unitary, and finally discard the additional register.

Since the same label $g$ is used at every step, the complete Stinespring realization in the $g$th block implements the $g$-rotated version $\mathcal{U}'_g\circ\mathcal{E}\circ\mathcal{U}_{g^{-1}}$ of the original channel.
After discarding the environment, the resulting channel on the principal system is $|G|^{-1}\sum_{g\in G} \mathcal{U}'_g\circ\mathcal{E}\circ\mathcal{U}_{g^{-1}}$. 
Consequently, if $\mathcal{E}$ is $G$-covariant, the action on the principal system is independent of $g$.
The regular-representation register can therefore be retained throughout the protocol and returned unchanged.

\subsection{Symmetric extension of elementary operations}

\subsubsection{Unitary operations}

Let $Q$ be a system carrying a projective unitary representation $U_Q(g)$ of $G$, and let $V$ be a unitary operation on $Q$.
We define its regular extension on $QR$ by
\begin{align}
    \widetilde{V}:=\sum_{g\in G} U_Q(g)VU_Q(g)^\dag \otimes\ket{\xi_g}\bra{\xi_g}.
\end{align}
Since $\{\ket{\xi_g}\}_{g\in G}$ is mutually orthogonal and $U_Q(g)VU_Q(g)^\dag$ is unitary for all $g\in G$, $\widetilde{V}$ is unitary. 
Moreover, $\widetilde{V}$ is $G$-invariant under the joint symmetry.
Indeed, for every $h\in G$,
\begin{align}
    \left(
    U_Q(h)\otimes L(h)
    \right)
    \widetilde{V}
    \left(
    U_Q(h)^\dag\otimes L(h)^\dag
    \right)
    =&\sum_{g\in G} U_Q(h)U_Q(g)VU_Q(g)^\dag U_Q(h)^\dag\otimes\ket{\xi_{hg}}\bra{\xi_{hg}} \nonumber\\
    =&\sum_{g\in G} U_Q(hg)V U_Q(hg)^\dag \otimes\ket{\xi_{hg}}\bra{\xi_{hg}}\nonumber\\
    =&\widetilde{V}.
\end{align}
In the second equality, the projective phase in
$U_Q(h)U_Q(g)$ cancels with its adjoint.
Conditioned on the regular-representation register being in
$\ket{\xi_g}$, the extended unitary therefore acts as $U_Q(g)VU_Q(g)^\dag$, while the group label stored in $R$ is preserved, i.e., 
\begin{align}
    \widetilde{V}(\rho\otimes \ket{\xi_g}\bra{\xi_g}
    =(U_Q(g)VU_Q(g)^\dag)\rho (U_Q(g)VU_Q(g)^\dag)^\dag\otimes \ket{\xi_g}\bra{\xi_g} 
\end{align}
for all $g\in G$.

\subsubsection{Projective measurements}

Let $\{P_x\}_{x\in X}$ be a projective measurement on $Q$.
For each $x\in X$, define
\begin{align}
    \widetilde{P}_x:=\sum_{g\in G} U_Q(g)P_xU_Q(g)^\dag\otimes\ket{\xi_g}\bra{\xi_g} .
\end{align}
The collection $\{\widetilde{P}_x\}_{x\in X}$ is a projective
measurement on $QR$.
We can easily check
\begin{align}
    &\widetilde{P}_x\widetilde{P}_y
    =\delta_{x,y}\widetilde{P}_x, \\
    &\sum_{x\in X}\widetilde{P}_x
    =I_Q\otimes I_R.
\end{align}
Each projection is invariant under the joint symmetry.
For every $h\in G$,
\begin{align}
    \left(
    U_Q(h)\otimes L(h)
    \right)
    \widetilde{P}_x
    \left(
    U_Q(h)^\dag\otimes L(h)^\dag
    \right)
    =\sum_{g\in G} U_Q(hg)P_xU_Q(hg)^\dag \otimes\ket{\xi_{hg}}\bra{\xi_{hg}} 
    =\widetilde{P}_x.
\end{align}
Again, the projective phase cancels under conjugation.
Conditioned on $R$ being in $|\xi_g\rangle$, the measurement is
therefore the $g$-rotated projective measurement $U_Q(g)P_xU_Q(g)^\dag$, while the regular-representation label itself is unchanged.

\subsubsection{State introduction}

State introduction requires a slightly different construction because the state introduced in the $g$th regular-representation block must be $U_E(g)\tau U_E(g)^\dag$, while an asymmetric state $\tau$ itself is not assumed to be freely introduced.

Let $E$ be the system of the state to be introduced, carrying a projective unitary representation $U_E(g)$ of $G$.
We introduce an additional regular-representation register $R'$ with basis $\{\ket{\xi_h}\}_{h\in G}$ and prepare the state
\begin{align}
    \widetilde{\tau}_{ER'}
    :=&\frac{1}{|G|} \sum_{h\in G} U_E(h)\tau U_E(h)^\dag\otimes \ket{\xi_h}\bra{\xi_h}. 
\end{align}
This state is invariant under the joint action
$U_E(k)\otimes L_{R'}(k)$ for every $k\in G$.
Thus, $\widetilde{\tau}_{ER'}$ can be introduced as a symmetric
ancillary state.

We next align the group label $h$ of the newly introduced register with the label $g$ already stored in $R$.
For $g,h\in G$, define
\begin{align}
    C_{g,h}
    :=U_E(g)U_E(g^{-1}h)^\dag U_E(g)^\dag.
\end{align}
We then define the controlled unitary on $ERR'$ by
\begin{align}
    W:=\sum_{g,h\in G} C_{g,h}\otimes
    \ket{\xi_g}\bra{\xi_g}_R
    \otimes
    \ket{\xi_h}\bra{\xi_h}_{R'}.
\end{align}
Since every $C_{g,h}$ is unitary and the control subspaces are mutually orthogonal, $W$ is unitary.
We next show that $W$ is symmetric.
The relative group element satisfies $(kg)^{-1}(kh)=g^{-1}h$. 
Therefore,
\begin{align}
    U_E(k)C_{g,h}U_E(k)^\dag =C_{kg,kh},
\end{align}
where the projective phases cancel under conjugation.
It follows that
\begin{align}
    \left(
    U_E(k)\otimes L_R(k)\otimes L_{R'}(k)
    \right)
    W
    \left(
    U_E(k)^\dag\otimes L_R(k)^\dag\otimes L_{R'}(k)^\dag
    \right)
    =W.
\end{align}

For fixed $g,h\in G$, the definition of $C_{g,h}$ gives
\begin{align}
    C_{g,h}
    U_E(h)\tau U_E(h)^\dag
    C_{g,h}^\dag
    =U_E(g)\tau U_E(g)^\dag.
\end{align}
Hence, for every $g\in G$,
\begin{align}
    W
    \left(
    |\xi_g\rangle\langle\xi_g|_R \otimes \widetilde{\tau}_{ER'}
    \right)
    W^\dag
    =|\xi_g\rangle\langle\xi_g|_R \otimes U_E(g)\tau U_E(g)^\dag \otimes \frac{I_{R'}}{|G|}.
\end{align}
After discarding $R'$, we therefore obtain
\begin{align}
    |\xi_g\rangle\langle\xi_g|_R\to 
    U_E(g)\tau U_E(g)^\dag\otimes |\xi_g\rangle\langle\xi_g|_R.
\end{align}
Thus, the $g$th regular-representation block introduces precisely the $g$-rotated version of the state $\tau$, using only symmetric state introduction, a symmetric unitary, and discarding.

\subsection{Extension of a Stinespring realization}

Consider a local channel $\mathcal{E}$ and fix a realization of $\mathcal{E}$ in terms of state introduction, unitary operations, projective measurements, classical control, and discarding.
Replace every quantum primitive by the symmetric construction given above, using the same regular-representation register $R$ throughout the realization.
Conditioned on $R$ being in $|\xi_g\rangle$, every state introduced in the realization is replaced by its $g$-rotated state, every unitary is replaced by its $g$-rotated unitary, and every projective measurement is replaced by its $g$-rotated projective measurement.
Since the label $g$ is preserved at every step, these replacements combine into the $g$-rotated version of the complete realization.
After all ancillary systems are discarded, the resulting channel on
the principal system is
\begin{align}
    \mathcal{E}_g
    =\mathcal{U}'_g\circ\mathcal{E}\circ\mathcal{U}_{g^{-1}}.
\end{align}
Here, the unitary actions on the discarded ancillary systems disappear
under the partial trace.
If $\mathcal{E}$ is $G$-covariant, then $\mathcal{E}_g=\mathcal{E}$ for all $g\in G$.

Consequently, if the regular-representation register is initialized in
a $G$-invariant classical mixture
\begin{align}
    \Gamma_R:=\frac{1}{|G|}\sum_{g\in G} |\xi_g\rangle\langle\xi_g|,
\end{align}
the symmetric realization satisfies
\begin{align}
    \widetilde{\mathcal{E}}
    \left(
    \rho\otimes\Gamma_R
    \right)
    =
    \mathcal{E}(\rho)\otimes\Gamma_R.
\end{align}
Thus, the original covariant channel is implemented exactly on the
principal system, while the regular-representation register is
returned unchanged.

\section{Stabilizer memory states for general Clifford symmetries}
\label{app:clifford_memory}

We show that stabilizer memory states satisfying the control condition can be constructed for an arbitrary finite-group symmetry represented by qubit Clifford operators.
The required memory states must transform covariantly under the group action, must be distinguishable by stabilizer measurements, and must allow stabilizer operations to be classically controlled by the identified group label.
We construct such memory states directly by assigning a separate stabilizer block to every pair of distinguishable group elements.

Let
\begin{align}
    K_{\mathrm{proj}}:=\{g\in G \,|\, U(g)\propto I\}
\end{align}
be the projective kernel of the representation, and define the effective group by
\begin{align}
    \overline{G}:=G/K_{\mathrm{proj}}.
\end{align}
Elements belonging to the same coset of $K_{\mathrm{proj}}$ act identically on every quantum state and therefore cannot be distinguished by any memory state.
It is thus sufficient and necessary to construct a memory for the effective group $\overline{G}$.
For notational simplicity, we choose one representative $g\in G$ for each element of $\overline{G}$.

We first establish a two-copy stabilizer witness for an arbitrary nontrivial Clifford operator.

\begin{lemma}
    \label{lem:Clifford-stabilizer-witness}
    Let $C$ be a qubit Clifford operator satisfying $C\not\propto I$.
    There exist an integer $n_C\leq2$, a pure stabilizer state $\lvert\phi_C\rangle$ on $n_C$ copies of the original system, and a Hermitian Pauli operator $A_C$ such that
    \begin{align}
        A_C\lvert\phi_C\rangle=\lvert\phi_C\rangle
    \end{align}
    and
    \begin{align}
        A_CC^{\otimes n_C}\lvert\phi_C\rangle=-C^{\otimes n_C}\lvert\phi_C\rangle.
    \end{align}
    In particular, $\lvert\phi_C\rangle$ and $C^{\otimes n_C}\lvert\phi_C\rangle$ are perfectly distinguished by measuring $A_C$.
\end{lemma}

\begin{proof}
    Suppose first that $C$ is proportional to a nonidentity Pauli operator $R$.
    Choose a Hermitian Pauli operator $P$ that anticommutes with $R$, and let $\lvert\phi_C\rangle$ be a pure stabilizer state satisfying
    \begin{align}
        P\lvert\phi_C\rangle=\lvert\phi_C\rangle.
    \end{align}
    Setting $n_C:=1$ and $A_C:=P$, we obtain
    \begin{align}
        A_CC\lvert\phi_C\rangle=-CA_C\lvert\phi_C\rangle=-C\lvert\phi_C\rangle.
    \end{align}

    Suppose next that $C$ is not proportional to a Pauli operator.
    There exists a Hermitian Pauli operator $P$ such that
    \begin{align}
        Q:=C^\dag PC\neq\pm P.
    \end{align}
    Define the two-copy Pauli operators
    \begin{align}
        &A:=P\otimes P, \\
        &B:=Q\otimes Q.
    \end{align}
    The operators $A$ and $B$ commute even when $P$ and $Q$ anticommute, and they are independent because $Q\neq\pm P$.
    They can therefore be extended to a maximal set of independent commuting Pauli operators.
    Hence, there exists a pure stabilizer state $\lvert\phi_C\rangle$ satisfying
    \begin{align}
        &A\ket{\phi_C}=\ket{\phi_C}, \\
        &B\ket{\phi_C}=-\ket{\phi_C}.
    \end{align}
    Setting $n_C:=2$ and $A_C:=A$, we obtain
    \begin{align}
        A_CC^{\otimes2}\ket{\phi_C}
        =C^{\otimes2}B\ket{\phi_C} 
        =-C^{\otimes2}\ket{\phi_C}.
    \end{align}
\end{proof}

We apply Lemma~\ref{lem:Clifford-stabilizer-witness} to every unordered pair of distinct elements of $\overline{G}$.
For each pair $\{g,h\}$, define the relative Clifford operator by
\begin{align}
    C_{g,h}:=U(g)^\dag U(h).
\end{align}
Since $g$ and $h$ represent distinct elements of $\overline{G}$, we have
\begin{align}
    C_{g,h}\not\propto I.
\end{align}
Lemma~\ref{lem:Clifford-stabilizer-witness} gives an integer $n_{g,h}\leq2$, a pure stabilizer state $\lvert\phi_{g,h}\rangle$, and a Hermitian Pauli operator $A_{g,h}$ satisfying
\begin{align}
    A_{g,h}\lvert\phi_{g,h}\rangle=\lvert\phi_{g,h}\rangle
\end{align}
and
\begin{align}
    A_{g,h}C_{g,h}^{\otimes n_{g,h}}\lvert\phi_{g,h}\rangle
    =-C_{g,h}^{\otimes n_{g,h}}\lvert\phi_{g,h}\rangle.
\end{align}

We define the seed memory state by
\begin{align}
    \lvert\Phi\rangle:=\bigotimes_{\{g,h\}\subset\overline{G}}\lvert\phi_{g,h}\rangle.
\end{align}
The tensor product is taken over all unordered pairs of distinct elements of $\overline{G}$.
The state $\lvert\Phi\rangle$ is a pure stabilizer state.
The total number of copies of the original system used by the memory is
\begin{align}
    N:=\sum_{\{g,h\}\subset\overline{G}}n_{g,h}\leq2\binom{|\overline{G}|}{2}=|\overline{G}|(|\overline{G}|-1).
\end{align}

For every $k\in\overline{G}$, define the corresponding memory state by
\begin{align}
    \lvert M_k\rangle:=U(k)^{\otimes N}\lvert\Phi\rangle.
\end{align}
Since $U(k)$ is a Clifford operator and $\lvert\Phi\rangle$ is a stabilizer state, every $\lvert M_k\rangle$ is a stabilizer state.
Moreover, for every $a,k\in\overline{G}$, the projective representation property gives
\begin{align}
    U(a)^{\otimes N}\lvert M_k\rangle\propto\lvert M_{ak}\rangle.
\end{align}
Thus, the group action permutes the memory labels according to left multiplication.

It remains to construct a stabilizer readout of the memory label.
For each unordered pair $\{g,h\}$, define a Pauli operator acting on the corresponding block by
\begin{align}
    P_{g,h}:=U(g)^{\otimes n_{g,h}}A_{g,h}U(g)^{\dag\otimes n_{g,h}}.
\end{align}
This is a Pauli operator because $U(g)$ is Clifford.
On the memory state with label $g$, the corresponding block satisfies
\begin{align}
    P_{g,h}U(g)^{\otimes n_{g,h}}\lvert\phi_{g,h}\rangle
    =U(g)^{\otimes n_{g,h}}\lvert\phi_{g,h}\rangle.
\end{align}
On the memory state with label $h$, the same block satisfies
\begin{align}
    P_{g,h}U(h)^{\otimes n_{g,h}}\lvert\phi_{g,h}\rangle
    =-U(h)^{\otimes n_{g,h}}\lvert\phi_{g,h}\rangle.
\end{align}
Therefore, measuring $P_{g,h}$ perfectly distinguishes the labels $g$ and $h$.

The operators $P_{g,h}$ associated with different unordered pairs act on different memory blocks.
They can therefore be measured simultaneously by a stabilizer measurement.
For every pair of distinct labels $g$ and $h$, the outcome of $P_{g,h}$ is different on $\lvert M_g\rangle$ and $\lvert M_h\rangle$.
Consequently, the joint measurement outcomes identify the memory label uniquely.
Classical post-processing of these outcomes defines a stabilizer POVM $\{\Pi_g\}_{g\in\overline{G}}$ satisfying
\begin{align}
    \mathrm{tr}\left(\Pi_g\lvert M_h\rangle\langle M_h\rvert\right)=\delta_{g,h}.
\end{align}

The readout also satisfies the control condition required in our definition of a memory.
Let $\{\mathcal{E}_g\}_{g\in\overline{G}}$ be an arbitrary family of stabilizer operations.
After applying the stabilizer POVM $\{\Pi_g\}_{g\in\overline{G}}$ to the memory, we may apply $\mathcal{E}_g$ conditionally on the classical outcome $g$.
The resulting channel is
\begin{align}
    \mathcal{C}(\tau_{MS})
    :=\sum_{g\in\overline{G}}\mathcal{E}_g\left(\mathrm{tr}_M\left((\Pi_g\otimes I_S)\tau_{MS}\right)\right).
\end{align}
Since stabilizer operations are closed under stabilizer measurements, classical feed-forward, and composition, the channel $\mathcal{C}$ is a stabilizer operation.
Thus, the memory label can be used to control arbitrary stabilizer operations without leaving the free-operation class.

We summarize the construction in the following theorem.

\begin{theorem}
    \label{thm:general-Clifford-memory}
    Let $G$ be a finite group with a projective representation $U(g)$ by qubit Clifford operators, and let
    \begin{align}
        \overline{G}:=G/K_{\mathrm{proj}},
        \qquad
        K_{\mathrm{proj}}:=\{g\in G \,|\, U(g)\propto I\}.
    \end{align}
    There exist a pure stabilizer state $\lvert\Phi\rangle$ on at most $|\overline{G}|(|\overline{G}|-1)$ copies of the original system and a family of stabilizer memory states
    \begin{align}
        \lvert M_g\rangle:=U(g)^{\otimes N}\lvert\Phi\rangle,
        \qquad
        g\in\overline{G},
    \end{align}
    satisfying the following properties.
    
    \begin{enumerate}
        \item The group action permutes the memory states according to
        \begin{align}
            U(a)^{\otimes N}\lvert M_g\rangle\propto\lvert M_{ag}\rangle.
        \end{align}
        \item There exists a stabilizer POVM $\{\Pi_g\}_{g\in\overline{G}}$ such that
        \begin{align}
            \mathrm{tr}\left(\Pi_g\lvert M_h\rangle\langle M_h\rvert\right)=\delta_{g,h}.
        \end{align}
        \item For every family of stabilizer operations $\{\mathcal{E}_g\}_{g\in\overline{G}}$, measuring the memory by $\{\Pi_g\}_{g\in\overline{G}}$ and applying $\mathcal{E}_g$ conditionally on the outcome defines a stabilizer operation.
    \end{enumerate}
\end{theorem}

This theorem constructs stabilizer memory states with free readout and free classical control for every finite-group symmetry represented by qubit Clifford operators.
Here, ``free'' refers to the underlying stabilizer theory.
The Pauli measurements used for readout are not required to commute with the symmetry representation.

The existence of such memory states therefore addresses a different question from the existence of a symmetric primitive implementation.
In particular, it does not imply that the readout can be performed using the symmetric Pauli measurements required in Sec.~\ref{sec:magic_primitive}.
Nor does it imply that an arbitrary coherent controlled-Clifford unitary is Clifford.
Accordingly, the general memory construction is consistent with restricting our primitive-level conversion theorem to Pauli representations, without asserting a separation of asymptotic conversion rates for other Clifford representations.

\section{Symmetry subgroup of dephased states} \label{sec:dephased_state_symetry}

The aim of this section is to prove Lemma~\ref{lem:dephased_state_symmetry_subgroup}, which we restate here: 

\setcounter{lemma}{9}

\begin{lemma} 
    Let $G$ be a finite group, $H$ be a Hamiltonian, $U$ be a projective unitary representation of $G$ satisfying $[U(h), H]=0$ for all $h\in G$, and $g\in G$. 
    Then, there exists some $m\in\mathbb{N}$ such that 
    \begin{align}
        \mathrm{Sym}\left(\mathcal{D}_{\mathcal{H}^{(m)}}(\rho^{\otimes m})\right) 
        =K. 
    \end{align}
\end{lemma}

\setcounter{lemma}{25}

For the proof of this lemma, we define the compact closure of the time evolution operators: 
\begin{align}
    \mathbb{T}_H:=\overline{\left\{e^{-itH}\ |\ t\in\mathbb{R}\right\}}. \label{eq:time_evolution_compact}
\end{align}
We denote its normalized Haar measure by $\mu$.

\begin{lemma} \label{lem:dephasing_Haar_average}
    Let $m\in\mathbb{N}$, $\rho$ be a state, $H$ be a Hamiltonian, $H^{\times m}$ be the sum of Hamiltonians on $m$ copies, and $\mathbb{T}_H$ be defined by Eq.~\eqref{eq:time_evolution_compact}. 
    Then, 
    \begin{align}
        \mathcal{D}_{H^{(m)}}(L) 
        =\int_{V\in\mathbb{T}_H} V^{\otimes m} L V^{\dag\otimes m}d\mu(V). \label{eq:lem:dephasing_Haar_average01}
    \end{align}
\end{lemma}

\begin{proof}
    By the left invariance of the Haar measure, for any $t\in\mathbb{R}$, we have 
    \begin{align}
        \int_{V\in\mathbb{T}_H} V^{\otimes m}LV^{\dag\otimes m}d\mu (V) 
        =\int_{V\in\mathbb{T}_H} \left(e^{-itH}V\right)^{\otimes m}L\left(e^{-itH}V\right)^{\dag\otimes m}d\mu (V). \label{eq:lem:dephasing_Haar_average02}
    \end{align}
    For any $s\in\mathbb{R}$, $e^{-isH}$ commutes with $e^{-itH}$. 
    Therefore, every $V\in\mathbb{T}_H$ commutes with $e^{-itH}$. 
    Then, we have 
    \begin{align}
        \int_{V\in\mathbb{T}_H} \left(e^{-itH}V\right)^{\otimes m}L\left(e^{-itH}V\right)^{\dag\otimes m}d\mu (V) 
        =&\int_{V\in\mathbb{T}_H} \left(Ve^{-itH}\right)^{\otimes m}L\left(Ve^{-itH}\right)^{\dag\otimes m}d\mu (V) \nonumber\\
        =&\int_{V\in\mathbb{T}_H} V^{\otimes m}e^{-itH^{(m)}}L e^{itH^{(m)}}V^{\dag\otimes m}d\mu (V). \label{eq:lem:dephasing_Haar_average03}
    \end{align}
    We denote the spectral decomposition of $H^{(m)}$ as 
    \begin{align}
        H^{(m)}=\sum_{j=1}^{J_m} E_{m, j}\Pi_{m, j}, \label{eq:lem:dephasing_Haar_average04}
    \end{align}
    where $E_{m, j}\neq E_{m, k}$ if $j\neq k$ and $\Pi_{m, j}$ is the projection operator onto the eigenspace with eigenvalue $E_{m, j}$. 
    Equation~\eqref{eq:lem:dephasing_Haar_average04} implies 
    \begin{align}
        e^{-itH^{(m)}}L e^{itH^{(m)}} 
        =\sum_{j, k=1}^{J_m} e^{-it(E_{m, j}-E_{m, k})}\Pi_{m, j}L\Pi_{m, k}. \label{eq:lem:dephasing_Haar_average05}
    \end{align}
    By plugging Eq.~\eqref{eq:lem:dephasing_Haar_average05} into Eq.~\eqref{eq:lem:dephasing_Haar_average03} we get
    \begin{align}
        \int_{V\in\mathbb{T}_H} \left(e^{-itH}V\right)^{\otimes m}L\left(e^{-itH}V\right)^{\dag\otimes m}d\mu (V) 
        =&\int_{V\in\mathbb{T}_H} V^{\otimes m}\left(\sum_{j, k=1}^{J_m} e^{-it(E_{m, j}-E_{m, k})}\Pi_{m, j}L\Pi_{m, k}\right) V^{\dag\otimes m}d\mu (V) \nonumber\\
        =&\sum_{j, k=1}^{J_m} e^{-it(E_{m, j}-E_{m, k})}\int_{V\in\mathbb{T}_H} V^{\otimes m}\Pi_{m, j}L\Pi_{m, k} V^{\dag\otimes m}d\mu (V). \label{eq:lem:dephasing_Haar_average06}
    \end{align}
    By Eqs.~\eqref{eq:lem:dephasing_Haar_average02} and \eqref{eq:lem:dephasing_Haar_average06}, we have 
    \begin{align}
        \int_{V\in\mathbb{T}_H} V^{\otimes m}LV^{\dag\otimes m}d\mu (V)
        =\sum_{j, k=1}^{J_m} e^{-it(E_{m, j}-E_{m, k})}\int_{V\in\mathbb{T}_H} V^{\otimes m}\Pi_{m, j}L\Pi_{m, k} V^{\dag\otimes m}d\mu (V), \label{eq:lem:dephasing_Haar_average07}
    \end{align}
    which implies 
    \begin{align}
        \int_{V\in\mathbb{T}_H} V^{\otimes m}LV^{\dag\otimes m}d\mu (V) 
        =&\sum_{j, k=1}^{J_m} \left(\frac{1}{2T}\int_{-T}^T e^{-it(E_{m, j}-E_{m, k})}dt\right)\int_{V\in\mathbb{T}_H} V^{\otimes m}\Pi_{m, j}L\Pi_{m, k} V^{\dag\otimes m}d\mu (V). \label{eq:lem:dephasing_Haar_average08}
    \end{align}
    By taking the limit $T\to\infty$, we get 
    \begin{align}
        \int_{V\in\mathbb{T}_H} V^{\otimes m}LV^{\dag\otimes m}d\mu (V) 
        =&\sum_{j, k=1}^{J_m} \delta_{j, k}\int_{V\in\mathbb{T}_H} V^{\otimes m}\Pi_{m, j}L\Pi_{m, k} V^{\dag\otimes m} d\mu(V) \nonumber\\
        =&\sum_{j=1}^{J_m} \int_{V\in\mathbb{T}_H} V^{\otimes m}\Pi_{m, j}L\Pi_{m, j} V^{\dag\otimes m} d\mu(V). \label{eq:lem:dephasing_Haar_average09}
    \end{align}
    By noting that for any $s\in\mathbb{R}$, $(e^{-isH})^{\otimes m}\Pi_{m, j}L\Pi_{m, j} (e^{isH})^{\otimes m}=\Pi_{m, j}L\Pi_{m, j}$, we have 
    \begin{align}
        V^{\otimes m}\Pi_{m, j}L\Pi_{m, j} V^{\dag\otimes m}=\Pi_{m, j}L\Pi_{m, j} \label{eq:lem:dephasing_Haar_average10}
    \end{align}
    for all $V\in\mathbb{T}_H$. 
    Thus, we have 
    \begin{align}
        \int_{V\in\mathbb{T}_H} V^{\otimes m}\Pi_{m, j}L\Pi_{m, j} V^{\dag\otimes m} d\mu(V) 
        =\int_{V\in\mathbb{T}_H} \Pi_{m, j}L\Pi_{m, j} d\mu(V) 
        =\Pi_{m, j}L\Pi_{m, j}. \label{eq:lem:dephasing_Haar_average11}
    \end{align}
    By plugging Eq.~\eqref{eq:lem:dephasing_Haar_average11} into Eq.~\eqref{eq:lem:dephasing_Haar_average09}, we get 
    \begin{align}
        \mathcal{D}_{H^{(m)}}(L) 
        =\int_{V\in\mathbb{T}_H} V^{\otimes m} L V^{\dag\otimes m}d\mu(V). 
    \end{align}
\end{proof}

\begin{lemma} \label{lem:time_evolution_orbit_closure}
    Let $\rho$ be a state, $H$ be a Hamiltonian, and $\mathbb{T}_H$ be defined by Eq.~\eqref{eq:time_evolution_compact}. 
    Then, 
    \begin{align}
        \overline{\left\{e^{-itH}\rho e^{itH} \ |\ t\in\mathbb{R}\right\}}
        =\{V\rho V^\dag\ |\ V\in\mathbb{T}_H\}. 
    \end{align}
\end{lemma}

\begin{proof}
    First, we show that $\overline{\left\{e^{-itH}\rho e^{itH} \ |\ t\in\mathbb{R}\right\}}\subset\{V\rho V^\dag\ |\ V\in\mathbb{T}_H\}$. 
    We take arbitrary $\tau\in\overline{\left\{e^{-itH}\rho e^{itH} \ |\ t\in\mathbb{R}\right\}}$. 
    Then, there exists some real sequence $(t_j)_{j=1}^\infty$ such that 
    \begin{align}
        \lim_{j\to\infty} e^{-it_j H}\rho e^{it_jH}=\tau. 
    \end{align}
    Since the unitary group is compact, there exists some subsequence $(t_{j(k)})_{k=1}^\infty$ such that $e^{-it_{j(k)} H}$ converges in the limit of $k\to\infty$. 
    We denote $V:=\lim_{k\to\infty} e^{-it_{j(k)} H}$. 
    Then, $V\in\mathbb{T}_H$ and $V\rho V^\dag=\tau$, which implies $\tau\in\{V\rho V^\dag\ |\ V\in\mathbb{T}_H\}$.

    Next, we show that $\overline{\left\{e^{-itH}\rho e^{itH} \ |\ t\in\mathbb{R}\right\}}\supset\{V\rho V^\dag\ |\ V\in\mathbb{T}_H\}$. 
    We take arbitrary $\tau\in\{V\rho V^\dag\ |\ V\in\mathbb{T}_H\}$. 
    By definition, $\tau$ can be written as $\tau=V\rho V^\dag$ with some $V\in\mathbb{T}_H$. 
    By the definition of $V$, there exists some real sequence $(t_j)_{j=1}^\infty$ such that 
    \begin{align}
        \lim_{j\to\infty} e^{-it_j H}=V, 
    \end{align}
    which implies 
    \begin{align}
        \lim_{j\to\infty} e^{-it_j H}\rho e^{it_j H}=V\rho V^\dag=\tau. 
    \end{align}
    We therefore have $\tau\in\overline{\left\{e^{-itH}\rho e^{itH} \ |\ t\in\mathbb{R}\right\}}$. 
\end{proof}

\noindent 
\textit{Proof of Lemma~\ref{lem:dephased_state_symmetry_subgroup}.}
    We prove this lemma in two steps. 
    In the first step, we show that 
    \begin{align}
        \forall m\in\mathbb{N},\ \mathrm{Sym}\left(\mathcal{D}_{\mathcal{H}^{(m)}}(\rho^{\otimes m})\right)\supset K. \label{eq:lem:dephased_state_symmetry_subgroup02}
    \end{align}
    In the second step, we show that 
    \begin{align}
        \exists m\in\mathbb{N} \text{ such that } \mathrm{Sym}\left(\mathcal{D}_{\mathcal{H}^{(m)}}(\rho^{\otimes m})\right)\subset K. \label{eq:lem:dephased_state_symmetry_subgroup03}
    \end{align}

    Before going into the details of the two steps, we explain two points that we commonly use in both steps. 
    We first define 
    \begin{align}
        \Omega_\rho^{(m)}:=\mathcal{D}_{H^{(m)}}(\rho^{\otimes m}). \label{eq:lem:dephased_state_symmetry_subgroup04}
    \end{align}
    Since $H$ commutes with $U(g)$ for all $g\in G$, we have 
    \begin{align}
        \mathcal{U}_g^{\otimes m}(\mathcal{D}_{H^{(m)}}(\rho^{\otimes m})) 
        =\mathcal{D}_{H^{(m)}}(\mathcal{U}_g(\rho)^{\otimes m}) 
        =\Omega_{\mathcal{U}_g(\rho)}^{(m)}, \label{eq:lem:dephased_state_symmetry_subgroup05}
    \end{align}
    which implies 
    \begin{align}
        \mathrm{Sym}(\mathcal{D}_{\mathcal{H}^{(m)}}(\rho^{\otimes m})) 
        =\left\{g\in G\ \middle|\ \Omega_{\mathcal{U}_g(\rho)}^{(m)}=\Omega_\rho^{(m)}\right\}. \label{eq:lem:dephased_state_symmetry_subgroup06}
    \end{align} 
    Next, we note that by Lemma~\ref{lem:time_evolution_orbit_closure}, $g\in K$ is equivalent to 
    \begin{align}
        \exists V\in\mathbb{T}_H \text{ such that } \mathcal{U}_g(\rho)=V\rho V^\dag, \label{eq:lem:dephased_state_symmetry_subgroup07}
    \end{align}
    which means 
    \begin{align}
        K=\{g\in G\ |\ \exists V\in\mathbb{T}_H \text{ such that }\mathcal{U}_g(\rho)=V\rho V^\dag\}. \label{eq:lem:dephased_state_symmetry_subgroup08}
    \end{align}

    As the first step, we show Eq.~\eqref{eq:lem:dephased_state_symmetry_subgroup02}. 
    We take arbitrary $g\in K$ and $m\in\mathbb{N}$. 
    By Eq.~\eqref{eq:lem:dephased_state_symmetry_subgroup08}, there exists some $V\in\mathbb{T}_H$ such that 
    \begin{align}
        \mathcal{U}_g(\rho)=V\rho V^\dag. \label{eq:lem:dephased_state_symmetry_subgroup09}
    \end{align}
    Thus, by Lemma~\ref{lem:dephasing_Haar_average}, we have 
    \begin{align}
        \Omega_{\mathcal{U}_g(\rho)}^{(m)}
        =\int_{W\in\mathbb{T}_H}
        \left(
            WV\rho V^\dag W^\dag
        \right)^{\otimes m}
        d\nu(W)
        =\int_{W\in\mathbb{T}_H}
        \left(
            W\rho W^\dag
        \right)^{\otimes m}
        d\nu(W)
        =\Omega_\rho^{(m)}, \label{eq:lem:dephased_state_symmetry_subgroup10}
    \end{align}
    where we used the right invariance of the Haar measure in the second equality. 
    By Eqs.~\eqref{eq:lem:dephased_state_symmetry_subgroup06} and \eqref{eq:lem:dephased_state_symmetry_subgroup10}, we have $g\in\mathrm{Sym}\left(\mathcal{D}_{\mathcal{H}^{(m)}}(\rho^{\otimes m})\right)$.

    As the second step, we show Eq.~\eqref{eq:lem:dephased_state_symmetry_subgroup03}. 
    We take arbitrary $g\in G$ and assume $g\not\in K$. 
    By Eq.~\eqref{eq:lem:dephased_state_symmetry_subgroup06}, it is sufficient to show that $\Omega_\rho^{(m)}\neq\Omega_{\mathcal{U}_g(\rho)}^{(m)}$ with some $m\in\mathbb{N}$. 
    To show this, we prove that 
    \begin{align}
        \mathrm{tr}\left(\Omega_\rho^{(m)}\rho^{\otimes m}\right)
        >\mathrm{tr}\left(\Omega_{\mathcal{U}_g(\rho)}^{(m)}\rho^{\otimes m}\right). \label{eq:lem:dephased_state_symmetry_subgroup11}
    \end{align}
    We suppose that there exists some $V\in\mathbb{T}_H$ such that $VU(g)\rho U(g)^\dag V^\dag=\rho$. 
    Then, we have $U(g)\rho U(g)^\dag=V^\dag\rho V$, which implies $g\in K$. 
    Since this contradicts with the assumption, for any $V\in\mathbb{T}_H$, we have $VU(g)\rho U(g)^\dag V^\dag\neq\rho$. 
    Therefore, we get 
    \begin{align}
        \mathrm{tr}(VU(g)\rho U(g)^\dag V^\dag \rho)
        =\mathrm{tr}(\rho^2)-\frac{1}{2}\|\rho-VU(g)\rho U(g)^\dag V^\dag\|_2^2 
        <\mathrm{tr}(\rho^2). \label{eq:lem:dephased_state_symmetry_subgroup12}
    \end{align}
    Since $\mathbb{T}_H$ is compact, this implies 
    \begin{align}
        \max_{V\in\mathbb{T}_H} \mathrm{tr}(VU(g)\rho U(g)^\dag V^\dag \rho)
        <\mathrm{tr}(\rho^2). \label{eq:lem:dephased_state_symmetry_subgroup13}
    \end{align}
    We define 
    \begin{align}
        \mathcal{N}:=
        \left\{
            V\in\mathbb{T}_H \ |\ \mathrm{tr}\left(V\rho V^\dag \rho\right)>c \label{eq:lem:dephased_state_symmetry_subgroup14}
        \right\} 
    \end{align}
    with 
    \begin{align}
        &a:=\mathrm{tr}(\rho^2), \label{eq:lem:dephased_state_symmetry_subgroup15}\\
        &b:=\max_{V\in\mathbb{T}_H} \mathrm{tr}(VU(g)\rho U(g)^\dag V^\dag \rho), \label{eq:lem:dephased_state_symmetry_subgroup16}\\
        &c:=\frac{a+b}{2}. \label{eq:lem:dephased_state_symmetry_subgroup17}
    \end{align}
    For any $m\in\mathbb{N}$, we have
    \begin{align}
        \mathrm{tr}\left(\Omega_\rho^{(m)}\rho^{\otimes m}\right) 
        =\int_{V\in\mathbb{T}_H} \left(\mathrm{tr}\left(V\rho V^\dag \rho\right)\right)^m d\nu(V) 
        \geq \int_{V\in\mathcal{N}} \left(\mathrm{tr}\left(V\rho V^\dag \rho\right)\right)^m d\nu(V) 
        \geq c^m \nu(\mathcal{N}). \label{eq:lem:dephased_state_symmetry_subgroup18}
    \end{align}
    We note that $\mathcal{N}$ is an open subset of $\mathbb{T}_H$ containing the identity, which implies $\nu(\mathcal{N})>0$. 
    On the other hand, we have 
    \begin{align}
        \mathrm{tr}\left(\Omega_{\mathcal{U}_g(\rho)}^{(m)} \rho^{\otimes m}\right) 
        =\int_{V\in\mathbb{T}_H} \left(\mathrm{tr}\left(V\mathcal{U}_g(\rho) V^\dag \rho\right)\right)^m d\nu(V) 
        \leq b^m. \label{eq:lem:dephased_state_symmetry_subgroup19}
    \end{align}
    Since $\nu(\mathcal{N})>0$ and $b<c$, we can take some $m\in\mathbb{N}$ such that $c^m \nu(\mathcal{N})>b^m$, which implies Eq.~\eqref{eq:lem:dephased_state_symmetry_subgroup11}. 

\hfill $\square$

\section{Construction of mutually orthogonal state set}

Finally, we prove the existence of the orthogonal orbit vectors used in the ergotropy construction.
The result provides a finite-copy state whose group orbit is perfectly distinguishable while retaining nonzero overlap with the corresponding tensor power of the input state.

\begin{lemma} \label{SMlem:regular_basis_const_first_prep}
    Let $a\in\mathbb{N}$, $\ket{\phi}$ and $\ket{\phi'}$ be eigenstates of $Z\in\mathcal{U}(\mathcal{H})$ with different eigenvalues $z$ and $z'$ satisfying $z^a=z'^a$, and $c_0, c_1, ..., c_{a-1}\in\mathbb{C}$ satisfy $|c_b|=1$ for all $b\in\{0, 1,.., a-1\}$, 
    $\ket{\psi}\in\mathcal{H}^{\otimes a-1}$ be defined by 
    \begin{align}
        \ket{\psi}:=\frac{1}{\sqrt{a}}\sum_{b=0}^{a-1} c_b \ket{\phi}^{\otimes a-b-1}\otimes \ket{\phi'}^{\otimes b}. 
    \end{align}
    Then, 
    \begin{align}
        \bra{\psi}Z^{\otimes a-1}\ket{\psi}=0. 
    \end{align}
\end{lemma}

\begin{proof}
    Since $\ket{\phi}$ and $\ket{\phi'}$ are eigenstates of a unitary $Z$ with different eigenvalues, $\ket{\phi}$ and $\ket{\phi'}$ are orthogonal to each other, which implies that $\{\ket{\phi}^{\otimes a-b-1}\otimes\ket{\phi'}^{\otimes b}\}_{b=0}^{a-1}$ forms an orthogonal set of states. 
    Thus we get 
    \begin{align}
        \bra{\psi}Z^{\otimes a-1}\ket{\psi}
        =&\left(\frac{1}{\sqrt{a}}\sum_{b=0}^{a-1} c_b^*\bra{\phi}^{\otimes a-b-1}\otimes \bra{\phi'}^{\otimes b}\right)
        \left(\frac{1}{\sqrt{a}}\sum_{b=0}^{a-1} c_b z^{a-b-1} z'^b \ket{\phi}^{\otimes a-b-1}\otimes \ket{\phi'}^{\otimes b}\right) \nonumber\\
        =&\frac{1}{a}\sum_{b=0}^{a-1} |c_b|^2 z^{a-b-1} z'^b \nonumber\\
        =&\frac{1}{a}\sum_{b=0}^{a-1} z^{a-b-1} z'^b \nonumber\\
        =&\frac{1}{a}\cdot\frac{z^a-z'^a}{z-z'} \nonumber\\
        =&0, 
    \end{align}
    where we used $z\neq z'$ in the fourth equality. 
\end{proof}

\begin{lemma} \label{SMlem:regular_basis_const_second_prep}
    Let $a\in\mathbb{N}$, $\ket{\phi}$ and $\ket{\phi'}$ be eigenstates of $Z\in\mathcal{U}(\mathcal{H})$ with different eigenvalues $z$ and $z'$ satisfying $z^a=z'^a$, and $\rho\in\mathcal{S}(\mathcal{H})$ satisfy $\bra{\phi}\rho\ket{\phi}>0$. 
    Then, there exists some $\ket{\psi}\in\mathcal{H}^{\otimes a-1}$ such that 
    \begin{align}
        &\bra{\psi}Z^{\otimes a-1}\ket{\psi}=0, \\
        &\bra{\psi}\rho^{\otimes a-1}\ket{\psi}>0. 
    \end{align}
\end{lemma}

\begin{proof}
    We define 
    \begin{align}
        \ket{\psi_\pm}:=\frac{1}{\sqrt{a}}\left(\ket{\phi}^{\otimes a-1}\pm\ket{\Phi}\right) 
    \end{align}
    with
    \begin{align}
        \ket{\Phi}:=\sum_{b=1}^{a-1} \ket{\phi}^{\otimes a-b-1}\otimes \ket{\phi'}^{\otimes b}. 
    \end{align}
    By Lemma~\ref{SMlem:regular_basis_const_first_prep}, we have $\bra{\psi_+}Z^{\otimes a-1}\ket{\psi_+}=\bra{\psi_-}Z^{\otimes a-1}\ket{\psi_-}=0$. 
    By the definition of $\ket{\psi_\pm}$, we have 
    \begin{align}
        &\bra{\psi_+}\rho^{\otimes a-1}\ket{\psi_+}+\bra{\psi_-}\rho^{\otimes a-1}\ket{\psi_-} \nonumber\\
        =&\frac{1}{\sqrt{a}}\left(\bra{\phi}^{\otimes a-1}+\bra{\Phi}\right)\rho^{\otimes a-1}\frac{1}{\sqrt{a}}\left(\ket{\phi}^{\otimes a-1}+\ket{\Phi}\right)
        +\frac{1}{\sqrt{a}}\left(\bra{\phi}^{\otimes a-1}-\bra{\Phi}\right)\rho^{\otimes a-1}\frac{1}{\sqrt{a}}\left(\ket{\phi}^{\otimes a-1}-\ket{\Phi}\right) \nonumber\\
        =&\frac{2}{a}\left(\bra{\phi}^{\otimes a-1}\rho^{\otimes a-1}\ket{\phi}^{\otimes a-1}+\bra{\Phi}\rho^{\otimes a-1}\ket{\Phi}\right) \nonumber\\
        \geq &\frac{2}{a}\bra{\phi}\rho\ket{\phi}^{a-1} \nonumber\\
        >&0, 
    \end{align}
    which implies $\bra{\psi_+}\rho^{\otimes a-1}\ket{\psi_+}>0$ or $\bra{\psi_-}\rho^{\otimes a-1}\ket{\psi_-}>0$. 
\end{proof}

\begin{lemma} \label{lem:left_regular_basis_construction}
    Let $G$ be a finite group, $U$ be a faithful projective unitary representation of $G$ on $\mathcal{H}$, and $\rho\in\mathcal{S}(\mathcal{H})$. 
    Then, there exist some $N\in\mathbb{N}$ and some pure state $\ket{\eta}\in\mathcal{H}^{\otimes N}$ satisfying 
    \begin{align}
        &\bra{\eta}U(g')^{\dag\otimes N}U(g)^{\otimes N}\ket{\eta}=\delta_{g, g'}, \label{SMeq:lem:left_regular_basis_construction1}\\
        &\bra{\eta}\rho^{\otimes N}\ket{\eta}>0. \label{SMeq:lem:left_regular_basis_construction2}
    \end{align} 
\end{lemma}

By the construction method in the following proof, we can take $N$ such that $N\leq (|G|-1)^2$.

\begin{proof}
    For each $h\in G-\{e\}$, we consider the eigenvalue decomposition of $U(h)$, i.e., 
    \begin{align}
        U(h)=\sum_{j=0}^{d-1} u_{h, j}\ket{\phi_{h, j}}\bra{\phi_{h, j}}. 
    \end{align}
    with $u_{h, j}\in\mathbb{C}$ and an orthonormal basis $\{\ket{\phi_{h, j}}\}_{j=0}^{d-1}$. 
    Since $\rho$ is a nonzero positive semidefinite operator, we can take some $j\in\{0, 1, ..., d-1\}$ such that $\bra{\phi_{h, j}}\rho\ket{\phi_{h, j}}>0$. 
    Since $U$ is a faithful representation up to phase, there exists some $j'\in\{0, 1, ..., d-1\}$ such that $u_{h, j}\neq u_{h, j'}$. 
    Since $G$ is a finite group, there exists some $a\in\mathbb{N}$ such that $h^a=e$. 
    This means that $U(h)^a\propto U(e)\propto I$, which implies $(u_{h, j})^a=(u_{h, j'})^a$. 
    Thus, we can take some $a_h\in\mathbb{N}$ such that $(u_{h, j})^{a_h}=(u_{h, j'})^{a_h}$. 
    By Lemma~\ref{SMlem:regular_basis_const_second_prep}, we can take some $\ket{\psi_h}\in\mathcal{H}^{\otimes a_h-1}$ such that $\bra{\psi_h}U(h)^{\otimes a_h-1}\ket{\psi_h}=0$ and $\bra{\psi_h}\rho^{\otimes a_h-1}\ket{\psi_h}>0$. 
    We define $N:=\sum_{h\in G-\{e\}} (a_h-1)$ and 
    \begin{align}
        \ket{\eta}:=\bigotimes_{h\in G-\{e\}} \ket{\psi_h}. 
    \end{align}
    Then, it trivially holds that $\bra{\eta}U(g)^{\dag\otimes N}U(g)^{\otimes N}\ket{\eta}=1$ for all $g\in G$. 
    For any $g, g'\in G$ satisfying $g\neq g'$, we have 
    \begin{align}
        \bra{\eta}U(g)^{\dag\otimes N}U(g')^{\otimes N}\ket{\eta} 
        =&\bra{\eta}(\omega(g, g^{-1}g')^{-1}U(g^{-1}g'))^{\otimes N}\ket{\eta} \nonumber\\
        =&\omega(g, g^{-1}g')^{-N}\left(\bigotimes_{h\in G-\{e\}} \bra{\psi_h}\right)\left(\bigotimes_{h\in G-\{e\} } U(g^{-1} g')^{\otimes a_h-1}\right)\left(\bigotimes_{h\in G-\{e\}} \ket{\psi_h}\right) \nonumber\\
        =&\omega(g, g^{-1}g')^{-N}\prod_{h\in G-\{e\}} \bra{\psi_h} U(g^{-1} g')^{\otimes a_h-1}\ket{\psi_h} \nonumber\\
        =&0, 
    \end{align}
    where we used $U(g)U(g^{-1}g')=\omega(g, g^{-1}g')U(g')$ in the first equality. 
    Equation~\eqref{SMeq:lem:left_regular_basis_construction2} can be directly confirmed by 
    \begin{align}
        \bra{\eta}\rho^{\otimes N}\ket{\eta} 
        =\left(\bigotimes_{h\in G-\{e\}} \bra{\psi_h}\right)\left(\bigotimes_{h\in G-\{e\} } \rho^{\otimes a_h-1}\right)\left(\bigotimes_{h\in G-\{e\}} \ket{\psi_h}\right)
        =\prod_{h\in G-\{e\}} \bra{\psi_h} \rho^{\otimes a_h-1}\ket{\psi_h} 
        >0. 
    \end{align}
\end{proof}

\end{document}